\RequirePackage{amssymb,amsthm,amsfonts,units,nicefrac}
\RequirePackage{stmaryrd}
\documentclass[]{article}

\usepackage{authblk}

\AtBeginDocument{}

\usepackage{subcaption} 
\usepackage[]{hyperref}

\usepackage{xcolor,bm,tikz,stackrel}
\usepackage[inline]{enumitem}

\usepackage[all]{xy}
\usepackage{tikz-cd}
\usepackage{xcolor}
\definecolor{darkolivegreen}{rgb}{0.33, 0.42, 0.18}
\definecolor{darkblue}{rgb}{0.0, 0.0, 0.55}
\definecolor{falured}{rgb}{0.5, 0.09, 0.09}	
\definecolor{tan}{rgb}{0.82, 0.71, 0.55}
\definecolor{turquoise}{rgb}{0.19, 0.84, 0.78}
\definecolor{lightgray}{rgb}{0.83, 0.83, 0.83}
\definecolor{babyblue}{rgb}{0.54, 0.81, 0.94}
\definecolor{applegreen}{rgb}{0.55, 0.71, 0.0}
\definecolor{amber}{rgb}{1.0, 0.75, 0.0}
\definecolor{atomictangerine}{rgb}{1.0, 0.6, 0.4}

\usepackage[inline]{enumitem}
\usepackage{tasks}
\usetikzlibrary{arrows,fit}
\usepackage{rotating}
\usetikzlibrary{trees,decorations.pathmorphing}
\usetikzlibrary{arrows.meta,shapes,positioning,calc,automata,quotes,patterns}
\usepackage{mathtools}
\usepackage{csquotes}

\usepackage[numbers]{natbib}
\usepackage{centernot}

\NewDocumentCommand\define{mo}
{\IfNoValueTF {#2}
{{\label{#1}\textbf{#1}}}
{{\label{#1}\textbf{#2}}}}

\NewDocumentCommand{\term}{mo}
{\IfNoValueTF {#2}
{{\hyperref[#1]{#1}}}
{\hyperref[#1]{#2}}}

\newcommand\altdefine[2]{\label{#1}\textbf{#2}}
\NewDocumentCommand{\altterm}{mm}{\term{#1}[#2]}

\newcommand{\isf}{\ensuremath{\mathsf{IS4}}}
\newcommand{\sfi}{\ensuremath{\mathsf{S4I}}}
\newcommand{\csf}{\ensuremath{\mathsf{CS4}}}
\newcommand{\gsf}{\ensuremath{\mathsf{GS4}}}
\newcommand{\gsfc}{\ensuremath{\mathsf{GS4^c}}}
\newcommand{\csff}{$\csf$ frame}

\newcommand{\rcm}[2]{#1_{#2}}

\newcommand{\bfrm}[1]{#1^\nec}
\newcommand{\dfrm}[1]{#1^{\centernot{\ps}}}

\newcommand{\rel}{\sqsubseteq}
\newcommand{\ler}{\sqsupseteq}
\newcommand\lb{\left\llbracket}
\newcommand\rb{\right\rrbracket}
\newcommand{\val}[1]{\lb #1 \rb}

\newcommand{\iiff}{\leftrightarrow}
\newcommand{\ineg}{\neg}
\newcommand{\imp}{\mathbin \rightarrow}

\newcommand{\lanfull}{{\mathcal L}}
\newcommand{\seq}{\succcurlyeq}
\newcommand{\lgt}[1]{|#1|}
\newcommand{\nec}{\Box}
\newcommand{\ps}{\Diamond}
\newcommand{\peq}{\preccurlyeq}
\newcommand{\from}{\colon}
\newcommand{\fallible}[1]{#1^\bot}
\makeatletter
\newcommand{\mlabel}[2]{#2\def\@currentlabel{#2}\label{#1}}
\makeatother

\newcommand\ignore[1]{}

\DeclareSymbolFont{AMSb}{U}{msb}{m}{n}
\DeclareMathSymbol{\N}{\mathbin}{AMSb}{"4E}
\DeclareMathSymbol{\Z}{\mathbin}{AMSb}{"5A}
\DeclareMathSymbol{\R}{\mathbin}{AMSb}{"52}
\DeclareMathSymbol{\Q}{\mathbin}{AMSb}{"51}
\DeclareMathSymbol{\I}{\mathbin}{AMSb}{"49}

\usepackage[nameinlink]{cleveref}

\newtheorem{remark}{Remark}
\newtheorem{definition}{Definition}

\newtheorem{example}{Example}
\newtheorem{lemma}{Lemma}
\newtheorem{theorem}{Theorem}
\newtheorem{proposition}{Proposition}
\newtheorem{corollary}{Corollary}

\title{Constructive S4 modal logics with the finite birelational frame property}

\author[1]{Philippe Balbiani. \footnote{\href{mailto:philippe.balbiani@irit.fr}{\tt philippe.balbiani@irit.fr}}}
\author[2]{Mart\'in Di\'eguez \footnote{\href{mailto:martin.dieguezlodeiro@univ-angers.fr}{\tt martin.dieguezlodeiro@univ-angers.fr}}}
\author[3]{David Fern\'andez-Duque \footnote{\href{mailto:fernandez-duque@ub.edu}{\tt fernandez-duque@ub.edu}}}
\author[4]{Brett McLean \footnote{\href{mailto:brett.mclean@ugent.be}{\tt brett.mclean@ugent.be}}}
\affil[1]{IRIT, Toulouse University. Toulouse, France}
\affil[2]{LERIA, University of Angers, Angers, France}
\affil[3]{Department of Philosophy, University of Barcelona, Barcelona, Spain}
\affil[4]{Department of Mathematics, Ghent University. Gent, Belgium}
\date{~}

\begin{document}

\maketitle

\begin{abstract}
The logics $\csf$ and $\isf$ are the two leading intuitionistic variants of the modal logic $\mathsf{S4}$.
Whether the finite model property holds for each of these logics have been long-standing open problems. It was recently shown that $\isf$ has the finite frame property and thus the finite model property. In this paper, we prove that $\csf$ also enjoys the finite frame property.

 Additionally, we investigate the following three logics closely related to $\isf$. The logic $\gsf$ is obtained by adding the G\"odel--Dummett axiom to $\isf$; it is both a superintuitionistic and a fuzzy logic and has previously been given a real-valued semantics. We provide an alternative birelational semantics and prove strong completeness with respect to this semantics. The extension $\gsfc$ of $\gsf$ corresponds to requiring a crisp accessibility relation on the real-valued semantics. We give $\gsfc$ a birelational semantics corresponding to an extra confluence condition on the $\gsf$ birelational semantics and prove strong completeness.  Neither of these two logics have the finite model property with respect to their real-valued semantics, but we prove that they have the finite frame property for their birelational semantics. Establishing the finite birelational frame property immediately establishes decidability, which was previously open for these two logics. Our proofs yield NEXPTIME upper bounds for $\csf$, $\gsf$, and $\gsfc$. The logic $\sfi$ is obtained from $\isf$ by reversing the roles of the modal and intuitionistic relations in the birelational semantics. We also prove the finite frame property, and thereby decidability, for $\sfi$, although our proof does not yield an elementary complexity bound.

In summary, we prove that $\csf$, $\sfi$, $\gsf$, and $\gsfc$ all have the finite birelational frame property, with the latter two having a NEXPTIME validity problem.
\end{abstract}


\maketitle
\section{Introduction}

{
Intuitionistic logic, with its roots in constructive reasoning, has found important and profound applications in computer science, most notably through the Curry--Howard correspondence. This correspondence, often termed the ``propositions-as-types'' principle, establishes a deep connection between logic and the type theory of programming languages. 

In the context of intuitionistic \emph{modal} logic, this correspondence has led to the development of expressive and powerful frameworks for reasoning about computation and program behavior. For example, modal types can represent computations that are guaranteed to terminate or those that might diverge. 
More generally, modal type theory allows modalities to capture program invariants, preconditions, and postconditions, allowing for a more fine-grained and expressive specification of program behavior.

Davis and Pfenning~\cite{davpfe01} extended the Curry--Howard correspondence to the case of a constructive version of the well known modal logic $\sf S4$ in order to achieve a modal account of \emph{staged computation}.\footnote{A staged computation is a computation that proceeds as a sequence of multiple stages, where each stage produces the input for the next. These stages can be carried out on different machines by different users---think, for example, of the (various) compilation and execution stages that are necessary to arrive at the final output of a program.} The resulting system, $\mathsf{Mini}\text{-}\mathsf{ML}^\square$, can be used to introduce the notion of partially evaluated functions in functional programming.
Similarly, Choudhury and Krishnaswami~\cite{choudhury2020recovering} proposed extending the constructive $\sf S4$ typed calculus to obtain a modal type system 
 for impure functional programming languages that can distinguish between expressions with and without side effects. 
 To do so, the authors use modalities to capture important notions such as \emph{safety} or \emph{purity} in a comonadic type framework.

Within the context of computer security, Miyamoto and Igarashi~\cite{miyamoto2004modal} used a constructive $\sf S4$ to propose a typed $\lambda$-calculus to capture typed-based information flow analysis. In their approach, security levels are mapped to Kripke worlds and the intuitionistic accessibility relation simulates the heritage between the security levels, i.e.~any information available at a lower level is also available at higher levels.

Another important field of application for constructive $\sf S4$ is the so-called Fitch-style modal $\lambda$-calculus~\cite{borghuis1994coming,Martini1996}, where Fitch-style calculi~\cite{fitch} are extended with two functions, `later' and `constant', 
 which can be captured in terms of the $\nec$ modality in intuitionistic modal logic. In~\cite{ranald18}, Clouston studies a variation of the modal $\mathsf{S4}$ Fitch-style calculus where the $\nec$ modality is idempotent (i.e. $\nec \varphi$ and $\nec \nec \varphi$ are not only logical equivalent but also identical). 
Most of the applications we have mentioned focus on the $\ps$-free fragment of constructive modal logic, but following on previous work of Fitch~\cite{fitch1948intuitionistic} and Wijesekera~\cite{wijesekera1990constructive}, constructive and intuitionistic modal logic may be enriched with the $\ps$ `possiblity' operator.
Unlike in the classical setting, the modalities $\nec$ and $\ps$ are not inter-definable, and whereas $\nec$ allows us to reason about properties that are ensured throughout a computation, $\ps$ instead indicates outcomes that {\em may} happen, including, for example, the possibility of a crash or of finding a solution to an \textsc{np}-complete problem.

Regarding  $\mathsf{S4}$ specifically, there are (at least) three prominent constructive/intuitionistic variants in the literature.
This paper investigates one of them, known as \altterm{def:csf}{$\csf$} \cite{AlechinaMPR01}, along with some variants of a second, known as \altterm{def:isf}{$\isf$} \cite{Simpson94}.
Each of these logics has a natural axiomatisation and is sound and \term{strongly complete} for a class of Kripke structures based on two preorders: one preorder $\peq$ for the intuitionistic implication, and another preorder $\rel$ for the modalities $\ps$ and $\nec$ (so the structures are particular types of \emph{\term{birelational}} frames).
However, the questions of whether the logics \altterm{def:csf}{$\csf$} and \altterm{def:isf}{$\isf$} enjoy the \emph{finite model property} remained open for twenty years or more: in the case of \altterm{def:csf}{$\csf$} since at least 2001 \cite{AlechinaMPR01}, and of \altterm{def:isf}{$\isf$} at least since 1994 \cite{Simpson94}. The question for \altterm{def:isf}{$\isf$} was recently resolved in the positive in \cite{DBLP:conf/lics/GirlandoKMMS23}. In this paper we do the same for \altterm{def:csf}{$\csf$}, solving the other main open problem concerning non-classical variants of $\mathsf{S4}$.
The third preexisting constructive/intuitionistic variant of $\mathsf{S4}$, which we do not investigate in this paper, is the logic $\mathsf{IntS4}$ studied (along with related variants of other modal logics) by Wolter and Zakharyaschev~\cite{Wolter1997,Wolter1999}.
In contrast to \altterm{def:csf}{$\csf$} and \altterm{def:isf}{$\isf$}, the logic $\mathsf{IntS4}$ was already known to have the finite model property, although the semantics are quite different, in particular employing different binary relations for each of $\ps$ and $\nec$.

Along with \altterm{def:csf}{$\csf$}, we also study three variants of \altterm{def:isf}{$\isf$} and show that they too all enjoy the finite model property.\footnote{For all of these results, including for \altterm{def:csf}{$\csf$}, we prove the stronger \emph{finite frame property}.}
The first two variants are modal extensions of \emph{G\"odel(--Dummett) logic,} a type of fuzzy logic. According to Caicedo and Ram\'irez~\cite{Caicedo2010StandardGM}, such fuzzy versions of modal logic may be used to model situations where modal notions may be applied to vague predicates.
One may consider propositions such as~\emph{``The weather is always warm in Tulum''} or \emph{``Alice knows that Bob is trustworthy'',} where to be `warm' or `trustworthy' are treated as vague (they can be true to varying degrees), and `always' and `knows' are treated as modalities.

Among many approaches for dealing with vagueness,  G\"odel logics~\cite{Godel32} are particularly appealing versions of fuzzy logic, because they can alternatively be viewed as superintuitionistic logics~\cite{heyting30}.
Indeed, propositional G\"odel logic lies strictly between intuitionistic propositional logic and classical propositional logic, and is characterised by a semantics that assigns to each proposition a truth value in the interval $[0,1]$.
If we further assume that modalities are interpreted using a \emph{crisp} accessibility relation (that is, expressions of the form $x\mathrel R y$ take values in $\{0,1\}$), then G\"odel modal logic may be seen as a specialisation of intuitionistic modal logic~\cite{servi1977modal,Simpson94}, where semantically the intuitionistic partial order validates an \term{upward linear}[upward linearity] condition. 

Some modal extensions of G\"odel logic have already been studied in the literature, including combinations of G\"odel logic with $\mathsf K$~\cite{MO11,MetcalfeO09,CaicedoMRR13}, $\mathsf{S4}$~\cite{Caicedo2010StandardGM}, $\mathsf{S5}$~[\citealp{Caicedo2010StandardGM},~\citealp{Hajek1998}, Chap\-ter~8], and with linear temporal logic~\cite{AguileraDFM25}.
In the specific case of G\"odel $\mathsf K$, it has been proved that the complexity of the validity problem is \textsc{pspace}-complete for the box fragment~\cite{MetcalfeO09} (where the logics of the \term{accessibility-crisp} and accessibility-fuzzy frames coincide) and also the diamond fragment~\cite{MO11} (for both \term{accessibility-crisp} frames and accessibility-fuzzy frames). Modal logics over accessibility-fuzzy frames have been axiomatised by Caicedo and Rodr\'iguez~\cite{Caicedo2010StandardGM,CaicedoBiModal}, and over \term{accessibility-crisp} frames by Rodr\'iguez and Vidal \cite{RodriguezV21}. 

Our first \altterm{def:isf}{$\isf$} variant, \altterm{def:gsf}{$\gsf$}, which we prove to be equivalent to G\"odel $\mathsf{S4}$, has semantics defined over a subclass of \altterm{def:isf}{$\isf$} frames: those where the intuitionistic relation is \emph{\term{upward linear}}, so that the logic validates the G\"odel--Dummett axiom $(p\imp q)\vee(q \imp p)$. A second variant, \altterm{def:gsfc}{$\gsfc$}, is defined by imposing a \altterm{forth--down confluent}{`forth--down' confluence} condition on \altterm{def:gsf}{$\gsf$} frames, and \altterm{def:gsfc}{$\gsfc$} is equivalent to \term{accessibility-crisp}} G\"odel $\mathsf{S4}$.

Our third \altterm{def:isf}{$\isf$} variant, \altterm{def:sfi}{$\sfi$}, is defined over the same class of frames as \altterm{def:isf}{$\isf$}, except that the roles of the intuitionistic and the modal preorders are interchanged.
These frame conditions for \altterm{def:sfi}{$\sfi$} are natural from a technical point of view, as they are the minimal conditions required so that both modalities $\ps$ and $\nec$ can be evaluated `classically', i.e.~using (there exists an accessible world\dots) and (for all accessible worlds\dots).

Our finite frame property proofs involve first showing that each logic is sound and complete for a class of birelational frames. 
We do this via canonical model constructions. However, after this the proofs diverge, and we use three different methods to obtain the finite frame property for our four logics, choosing the method that give the simplest argument and best complexity bound for each logic.

In the case of \altterm{def:csf}{$\csf$}, the canonical model may be \emph{relativised} to a finite set of formulas, directly yielding the finite frame property.
In the case for the two G\"odel logics \altterm{def:gsf}{$\gsf$} and \altterm{def:gsfc}{$\gsfc$}, the canonical model cannot be relativised in this way.
However, the associated birelational semantics is restrictive enough that we can obtain finite models directly by taking an appropriate \emph{bisimulation quotient} of a model.
The case of \altterm{def:sfi}{$\sfi$} is somewhat more involved, as we can neither relativise the canonical model nor apply a bisimulation quotient directly.
Instead, we first show that it is sufficient that the logic enjoys the \emph{\term{shallow} frame property}, meaning that any non-\term{valid} formula may be falsified on a frame where the length of any $\prec$-chain is bounded (as usual, $w\prec v$ means that $w\peq v$ but $v\not\peq w$).%
\footnote{The descriptor \emph{\term{shallow}} has been used in a similar way in the context of classical modal logic \cite{SP08}.}
While \term{shallow} models may in principle be infinite, their quotients modulo bisimulation 
 are always finite, thus reducing the problem of proving the finite frame property to that of proving the \term{shallow} frame property, which we then proceed to do.

\ignore{First, for each $\Lambda\in \{\isf, \gsf, \gsfc, \sfi, \csf\}$, we construct a canonical model $\mathcal M^\Lambda_{\mathrm c}  = (W_{\mathrm c} ,{\peq_{\mathrm c} },{\rel_{\mathrm c} },V_{\mathrm c} )$ using fairly standard techniques as found in for example~\cite{Simpson94,AlechinaMPR01}.
Based on this canonical model, we fix a finite set of formulas $\Sigma$ and construct a \term{shallow} model $\mathcal M^\Lambda_\Sigma = (W_\Sigma,{\peq_\Sigma},{\rel_\Sigma},V_\Sigma)$.
The details of the construction vary for each of the three logics we consider, but with the general theme that $w\prec_\Sigma v$ may only hold if there is some $\varphi\in \Sigma$ that holds on $v$ but not on $w$.
Having placed this restriction, it is readily seen that any chain
\[w_0 \prec_\Sigma w_1 \prec_\Sigma  \ldots \prec_\Sigma w_n\]
gives rise to distinct formulas $\varphi_0,\ldots,\varphi_{n-1}$ of $\Sigma$ with the property that $\varphi_i$ holds on $w_{i+1}$ but not on $w_n$.
It follows that the length of the chain is bounded by $|\Sigma|+ 1$.

}

\paragraph*{\bf{Structure of paper}}\

\smallskip
\noindent\underline{\Cref{SecBasic}}: We present the five logics \altterm{def:csf}{$\csf$}, \altterm{def:isf}{$\isf$}, \altterm{def:gsf}{$\gsf$}, \altterm{def:gsfc}{$\gsfc$}, and \altterm{def:sfi}{$\sfi$} studied in this paper. We first give syntactic definitions using Hilbert-style deductive calculi, and then give semantics in terms of classes of birelational structures. We also present the alternative real-valued semantics for \altterm{def:gsf}{$\gsf$} and \altterm{def:gsfc}{$\gsfc$}.
\smallskip

\noindent\underline{\Cref{secSound}}: We prove the soundness of the deductive systems with respect to the associated birelational semantics. We also verify the relationships between the five logics, in particular confirming that they are all distinct.

\smallskip

\noindent\underline{\Cref{secCompCS4}}: We prove the \term{strong completeness}[strong completeness] and finite frame property (hence \term{decidability}[decidability]) of \altterm{def:csf}{$\csf$}. Both proofs are obtained via a canonical model argument, where the strong completeness proof (\Cref{cs4_completeness}) uses a standard canonical model construction, while while the finite frame property (\Cref{thmCS4fin}) is obtained via a \term{relativised canonical model}[relativisation of the canonical model] to a finite set of formulas.
The finite frames are exponentially bounded in size, which implies decidability in \textsc{nexptime} (\Cref{thmCS4dec}).  

\smallskip

\noindent\underline{\Cref{secCompGS4}}: We present \term{strong completeness}[strong completeness] of the logics \altterm{def:isf}{$\isf$}, \altterm{def:sfi}{$\sfi$}, \altterm{def:gsf}{$\gsf$}, and \term{def:gsfc}[${\gsfc}$], by constructing canonical models for them in a uniform way.

\smallskip

\noindent\underline{\Cref{sec:fmp}}: We prove the finite frame property for the three logics \altterm{def:gsf}{$\gsf$}, \altterm{def:gsfc}{$\gsfc$}, and \altterm{def:sfi}{$\sfi$} (that is, for all the logics we consider except for \altterm{def:csf}{$\csf$}, whose finite frame property we have already established, and  \altterm{def:isf}{$\isf$}, whose finite frame property is studied in~\cite{DBLP:conf/lics/GirlandoKMMS23}).
In \Cref{secSBisim} we introduce notions used in all three finite frame property proofs: two notions of bisimulation and quotients with respect to such bisimulations. In \Cref{sec:godelfinite} we prove the finite frame property for \altterm{def:gsf}{$\gsf$} (\Cref{theorem:g4s_finite_model}) by taking bisimulation quotients of falsifying \altterm{birelational model}{birelational models} from the corresponding class, relying on the soundness and completeness results of the previous two sections. In \Cref{sec:godelcfinite} we prove the finite frame property for \altterm{def:gsfc}{$\gsfc$} (\Cref{thm:gs4c_finite_frame}) in a similar way, although some technical modifications are required. In \Cref{sShallow} we show that, to prove the finite frame property for \altterm{def:sfi}{$\sfi$}, it is sufficient to prove a \emph{\term{shallow}} frame property. Then in \Cref{secFMPS4I} we prove the finite frame property for \altterm{def:sfi}{$\sfi$} (\Cref{thm:s4i_finite_frame}), by showing how to construct \term{shallow} falsifying models by modifying the canonical model.  
\smallskip

\noindent\underline{\Cref{sec:conclusions}}: We finish the paper by drawing conclusions and discussing potential future lines of research. 

\medskip
This paper is an extended version of the two conference papers~\cite{bdd21} and~\cite{DF23}. In this version, we present new, more direct proofs of the finite frame properties for \altterm{def:csf}{$\csf$}, \altterm{def:gsf}{$\gsf$}, and \altterm{def:gsfc}{$\gsfc$}. These new proofs demonstrate that the three logics are in \textsc{nexptime} (Theorems~\ref{thmCS4dec}, \ref{theorem:decide}, and \ref{theorem:g4sc_nexptime}), which did not follow previously.

\section{Syntax and semantics}\label{SecBasic}

In this section we first introduce the propositional modal language shared by all the logics we are interested in. Then we define the logics themselves. Finally, we present the various birelational or real-valued semantics for the logics. 

Fix a \emph{countably infinite} set $\mathbb P$ of propositional variables.\footnote{The restriction to countable $\mathbb P$ can be a genuine restriction for questions about entailment (see, for example, \cite[Proposition~3.1]{CaicedoBiModal}), but not for validity, since formulas are finite.} Then the \define{intuitionistic modal language} $\lanfull$ is defined by the grammar (in Backus--Naur form)
\[\varphi,\psi \coloneqq  \   p \  | \   \bot  \ |  \ \left(\varphi\wedge\psi\right) \  |  \ \left(\varphi\vee\psi\right)  \ |  \ \left(\varphi\imp \psi\right)    \  | \  \ps\varphi \  |  \ \nec\varphi   \]
where $p\in \mathbb P$.
We also use $\ineg\varphi$ as shorthand for $\varphi\imp \bot$ and $\varphi\iiff \psi$ as shorthand for $(\varphi\imp \psi) \wedge (\psi\imp\varphi)$. As usual, the unary modalities bind tighter than the binary connectives; we also assume that $\wedge$ and $\vee$ bind tighter than $\imp$.

In this paper a \define{logic} means a set of formulas closed under substitution (by formulas of the underlying language, which will always be $\lanfull$).

\subsection{Deductive calculi}

We define the logics we are interested in syntactically, via Hilbert-style deductive calculi. 

\begin{definition}
The logic \altdefine{def:csf}{$\bm{\csf}$} is the smallest logic containing all intuitionistic tautologies and closed under the following axioms and rules.
\smallskip

\noindent\fbox{\begin{minipage}{\textwidth}
\begin{enumerate}[itemsep=2pt]
\item[\mlabel{ax:k:box}{\ensuremath{\bm{ \mathrm K_{\nec}}}}] $\nec (p \imp q) \imp(\nec p \imp \nec q)$
\item[\mlabel{ax:k:dia}{\ensuremath{\bm{\mathrm K_{\ps}}}}] $\nec (p \imp q) \imp( \ps p \imp \ps q)$
\end{enumerate}
\begin{tasks}[label={}, label-width=3em,label-align = {right}](3)
\task[\mlabel{rl:mp}{\ensuremath{\bm{\mathrm{MP}}}}] $\dfrac{\varphi \imp \psi \hspace{10pt} \varphi}{\psi}$
\task[\mlabel{ax:ref:box}{\ensuremath{\bm{\mathrm T_{\nec}}}}] $\nec p \imp p$
\task[\mlabel{ax:trans:dia}{\ensuremath{\bm{\mathrm 4_{\ps}}}}] $\ps \ps p \imp \ps p$
\task[\mlabel{rl:nec}{\ensuremath{\bm{\mathrm{Nec}}}}] $\dfrac{\varphi}{\nec \varphi}$
\task[\mlabel{ax:ref:dia}{\ensuremath{\bm{\mathrm T_{\ps}}}}] $p \imp \ps p$
\task[\mlabel{ax:trans:box}{\ensuremath{\bm{\mathrm 4_{\nec}}}}] $\nec p \imp \nec \nec p$
\end{tasks}
\end{minipage}}
\end{definition}

 We now define the additional axioms


\begin{tasks}[label={}, label-width=3em,label-align = {right}](2)
\task[\mlabel{ax:dp}{\ensuremath{\bm{\mathrm{DP}}}}] $\,\ps(p \vee q) \imp \ps p \vee \ps q$;
\task[\mlabel{ax:fs}{\ensuremath{\bm{\mathrm{FS2}}}}] $\,(\ps p \imp \nec q ) \rightarrow \nec (p \imp q )$;
\task[\mlabel{ax:cd}{\ensuremath{\bm{\mathrm{CD}}}}] $\,\nec(p \vee q) \imp \nec p \vee \ps q$;
\task[\mlabel{ax:null}{\ensuremath{\bm{\mathrm{N}}}}] $\,\neg \ps \bot$;

\task[\mlabel{ax:g}{\ensuremath{\bm{\mathrm{GD}}}}] $\,(p \imp q) \vee (q \imp p)$.
\end{tasks}
Here \ref{ax:dp} stands for `disjunctive possibility', \ref{ax:fs} for `Fischer Servi 2'~\cite{servi1984axiomatizations}, \ref{ax:cd} for `constant domain', \ref{ax:null} for `nullary', and~\ref{ax:g} for `G\"odel--Dummett'.

\begin{definition}\label{definition:additional_axioms}
The logics \label{def:isf}{$\bm{\isf}$}, \label{def:sfi}{$\bm{\sfi}$} \label{def:gsf}{$\bm{\gsf}$}, and \label{def:gsfc}{$\bm{\gsfc}$} are defined similarly to $\csf$, with additional axioms as indicated in \Cref{fig:logs}.
\begin{figure}[h!]
\begin{minipage}[l]{0.49\textwidth}
\begin{align*}
\boldsymbol{\isf}& = {\csf} + \ref{ax:dp} + \ref{ax:null} + \ref{ax:fs} \\[0pt]
\boldsymbol{\sfi}& = {\csf} + \ref{ax:dp} +\ref{ax:null}  + \ref{ax:cd} \\[0pt]
\boldsymbol{\gsf} & = {\isf} + \ref{ax:g}\\[0pt]
\boldsymbol{\gsfc} & = {\gsf} + \ref{ax:cd} = {\sfi} + \ref{ax:fs} + \ref{ax:g}
\end{align*}
\end{minipage}
\begin{minipage}{.5\textwidth}\centering\begin{tikzpicture}
\draw[black] (0,0) node[anchor=north]{\altterm{def:csf}{$\csf$}} -- (-.87, .5)node[anchor=east]{\altterm{def:isf}{$\isf$}} -- (-.87, 1.5)node[anchor=east]{\altterm{def:gsf}{$\gsf$}} -- (0, 2)node[anchor=south]{\altterm{def:gsfc}{$\gsfc$}} -- (.87, 1)node[anchor=west]{\altterm{def:sfi}{$\sfi$}} -- (0,0);
\filldraw (0,0) circle (1pt);
\filldraw (-.87, .5) circle (1pt);
\filldraw (-.87, 1.5) circle (1pt);
\filldraw (0, 2) circle (1pt);
\filldraw (.87, 1) circle (1pt);
\end{tikzpicture}
\end{minipage}
\caption{Definition of logics considered and Hasse diagram for these logics}
\label{fig:logs}
\end{figure}
\end{definition}

Thus \altterm{def:csf}{$\csf$} will serve as the `minimal' logic for the purpose of this paper, and the remaining logics we consider are extensions. We will see in \Cref{secSound} that only the inclusions obtained directly from the definitions (and displayed in \Cref{fig:logs}) hold.

We use the standard Gentzen-style notation that defines $\Gamma \vdash_\Lambda \Delta$ to mean $\bigwedge \Gamma' \imp \bigvee \Delta ' \in \Lambda$ for some finite $\Gamma'\subseteq \Gamma$ and $\Delta'\subseteq\Delta$. We call $\vdash_\Lambda$ the \define{syntactic consequence} relation.
The logic $\Lambda$ will usually be clear from context, in which case we will not reflect it in the notation.
When working with a turnstile, we will follow the usual proof-theoretic conventions of writing $\Gamma,\Delta$ instead of $\Gamma \cup \Delta$ and $\varphi$ instead of $\{\varphi\}$.

\subsection{Semantics}\label{sec:semantics}

We will consider several semantics for our intuitionistic variants of $\mathsf{S4}$, some of which are intuitionistic semantics (that is, featuring \altterm{intuitionistic frame}{intuitionistic frames}) and some of which are real-valued semantics. For the various intuitionistic semantics, it will be convenient to introduce a general class of structures that includes each of these semantics as special cases.

\begin{definition}\label{DefSem}
An \define{intuitionistic frame} is a triple $\mathcal F=(W,\fallible W ,{\peq} )$, where $W$ is a set, $\peq $ is a preorder (that is, a reflexive and transitive binary relation) on $W$, and $\fallible{ W} \subseteq W$ is closed under $\peq $, that is,~whenever $w \in \fallible{W}$ and $w \peq v$, we also have that $v \in \fallible{W}$~\cite{ArisakaDS15}.
We say that $\mathcal F$ is \define{upward linear} if $w\peq u$ and $w\peq v$ implies that $u\peq v$ or $v\peq u$.

A \define{birelational frame} is a quadruple $\mathcal F=(W,\fallible W ,{\peq} ,R)$, where $ (W,\fallible W ,\allowbreak{\peq} )$ is an \term{intuitionistic frame} and $ (W,\fallible W ,R)$ is a Kripke frame in which $\fallible W$ is closed under $R$. It is a \define{bi-preorder} if both $ (W,\fallible W ,{\peq} )$ and $ (W,\fallible W ,R)$ are \altterm{intuitionistic frame}{intuitionistic frames} (i.e., if $R$ is also a preorder), in which case we will write $\rel $ instead of $R$.
We call the \term{birelational frame} $\mathcal F$ upward linear if $(W,\fallible W ,\allowbreak{\peq})$ is \term{upward linear}.
\end{definition}

The set $\fallible W$ is called the set of \define{fallible worlds}, as in~\cite{AlechinaMPR01,ArisakaDS15,Veldman76,ilik2010}.
Note that in a \term{birelational frame}, $\fallible W$ is closed under both $\peq$ and $R$.
When $\fallible W=\varnothing$ we omit it, and view $\mathcal F$ as a triple $(W,{\peq} ,R)$.
In this case, we say that $\mathcal F$ is \define{infallible}.

\begin{definition}
Given a \term{birelational frame} $\mathcal F = (W,\fallible W ,{\peq} ,R)$, a \define{valuation} on $\mathcal F$ is a function $V\from \mathbb P \to 2^{W}$ that is \emph{monotone} in the sense that each $V(p)$ is upward closed with respect to $\peq$ and includes $\fallible W$ (to ensure the validity of $\bot\imp \varphi$).

We define the satisfaction relation $\models$ recursively by
\begin{itemize}	
\item $(\mathcal M,w) \models p \in\mathbb P$ if $w\in V(p)$;
\item $(\mathcal M,w) \models \bot $ if $w \in \fallible{W}$; 
\item $(\mathcal M,w) \models \varphi \wedge \psi$ if $(\mathcal M,w) \models \varphi $ and $(\mathcal M,w) \models \psi $;
\item $(\mathcal M,w) \models \varphi \vee \psi$ if $(\mathcal M,w) \models \varphi $ or $(\mathcal M,w) \models \psi $;
\item $(\mathcal M,w) \models \varphi \imp \psi$ if for all $v\seq w$, $(\mathcal M,v) \models \varphi $ implies $(\mathcal M,v) \models \psi $;
\item $(\mathcal M,w) \models \ps \varphi  $ if for all $u\seq w$ there exists $v$ such that $u \mathrel R v$ and $(\mathcal M,v) \models \varphi $;
\item $(\mathcal M,w) \models \nec \varphi $ if for all $u,v$ such that $w\peq u \mathrel R v$, $(\mathcal M,v) \models \varphi $.
\end{itemize}
\end{definition}

It can be easily proved by induction on $\varphi$ that \emph{persistence} (or \emph{monotonicity}) holds: for all $w,v \in W$, if $w \peq v$ and $(\mathcal M,w)\models \varphi$ then $(\mathcal M,v)\models \varphi$.

A \define{birelational model} is a \term{birelational frame} equipped with a \term{valuation}; if the frame is a \term{bi-preorder}, then we speak of a \define{bi-preordered model}.
If $\mathcal M=(W,\fallible W,{\peq},R,V)$ is any model and $\varphi$ any formula, we write $\mathcal M\models\varphi$ if $(\mathcal M,w)\models \varphi$ for every $w\in W$.

Validity is defined as one would expect, but it is worth being precise about the surrounding terminology and notation.

\begin{definition}
Let $\mathfrak C$ be any class of models or class of frames. Let $\Gamma$ be a set of formulas and $\varphi$ a formula. We write $\Gamma \models_\mathfrak C \varphi$  and say that $\varphi$ is a \define{local birelational semantic consequence} of $\Gamma$ if, for each model $\mathcal M =(W,\fallible W,{\peq},R,V)$ from $\mathfrak C$ and each $w \in W$, we have \[\forall \psi \in \Gamma,(\mathcal M, w) \models \psi \implies (\mathcal M, w) \models \varphi.\]

We say that $\varphi$ is \define{valid} on $\mathfrak C$ if $\models_\mathfrak C \varphi$ (that is, if $\varnothing \models_\mathfrak C \varphi$; we may also write $\mathfrak C \models \varphi$), and \define{falsifiable} otherwise. We say that $\varphi$ is \define{satisfiable} on $\mathfrak C$ if $\varphi \not\models_\mathfrak C \bot$. This terminology extends also to single models or frames.
\end{definition}

\begin{definition}
	Let $\mathfrak C$ be a class of frames or class of models. A logic $\Lambda$ is said to be \define{strongly complete} with respect to $\mathfrak C$ if for any set of formulas $\Gamma$ and for any formula $\varphi$, if $\Gamma \models_{\mathfrak C} \varphi$ then $\Gamma \vdash_{\Lambda} \varphi$. 
\end{definition}


Since this paper is only concerned with constructive analogues of $\mathsf{S4}$, we will only be working with \altterm{birelational frame}{birelational frames} that are \altterm{bi-preorder}{bi-preorders}. Note, however, that even the class of \altterm{bi-preorder}{bi-preorders} does not validate some $\mathsf{S4}$ axioms, as the following example shows.

\begin{example}
Consider the \term{bi-preordered model}\footnote{Recall that $\rel$ indicates that the modal relation is a preorder.} $\mathcal M=(W,{\peq} ,{\rel} , V )$, whose corresponding frame is displayed in \Cref{fig:gen-frame}, with $V (p) = \lbrace x, y, z, t \rbrace$.
Note that we omit $\fallible W$: by convention, this means that $\fallible W =\varnothing$.

\begin{figure}[h!]\centering
	
	\begin{tikzpicture}[node distance=1.5cm,auto, minimum width=2mm, text width=2mm]
	\node[state,fill=lightgray] (x) {$x$};
	\node[state,fill=lightgray] (y) [above of = x] {$y$};
	\node[state,fill=lightgray] (z) [right of = y] {$z$};
	\node[state,fill=lightgray] (t) [above of = z] {$t$};
	\node[state,fill=lightgray] (w) [right of = t] {$w$};
		
	\path[->,dashed] 
	(x) edge[dashed] (y)
	(z) edge[dashed] (t);
\path[->]
	(y) edge[] (z)
	(t) edge[] (w);
	\end{tikzpicture}
	\caption{A \term{bi-preorder}. Dashed arrow for $\peq$; arrows for $\rel$. (Reflexive lines and arrows not displayed)}
	\label{fig:gen-frame}
\end{figure}		

Since $p$ is satisfied at $x$, $y$, and $z$, we have $(\mathcal M, x)  \models \nec p $. But since $p$ is not satisfied at $w$, we have $(\mathcal M, z)  \not \models \nec p$, and thus $(\mathcal M, x)  \not \models \nec\nec p $. Hence $(\mathcal M, x) \not \models \nec p \imp \nec \nec p$, and therefore $\ref{ax:trans:box}$ is not \term{valid} on the class of all \altterm{bi-preorder}{bi-preorders}.
\end{example}

In order to make~$\ref{ax:trans:box}$ \term{valid}, it is necessary to enforce additional constraints governing the interaction between $\peq$ and $\rel$.
There are various properties that have been used to this end.

\begin{definition}
Let $\mathcal F=(W,\fallible W,{\peq})$ be an \term{intuitionistic frame} and $R\subseteq W\times W$. We say that $R$ is:
\begin{enumerate}[label=(\roman*)]

\item \define{forth--up confluent} (for $\mathcal F$) if, whenever $w\peq w'$ and $w \mathrel R v$, there is $v'$ such that $v\peq v'$ and $w' \mathrel R v'$;

\item \define{back--up confluent} (for $\mathcal F$) if, whenever $w \mathrel R v \peq v'$, there is $w'$ such that $w \peq w' \mathrel R v'$;

\item \define{forth--down confluent} (for $\mathcal F$) if, whenever $w \peq v \mathrel R v'$, there is $w'$ such that $w \mathrel R w' \peq v'$.

\end{enumerate}
\begin{figure}[h!]
\centering
\begin{tikzpicture}
\node at (1,1)[align=center]{forth--up};
\draw[line width=.35mm,-Straight Barb] (0,2) -- (0,0) -- (2,0);
\draw[line width=.35mm,dashed] (2,0) -- (2,2);
\draw[line width=.35mm,-Straight Barb, dashed] (0,2) -- (2,2);
\node at (0,0)[circle, fill, inner sep = .8pt]{};
\node at (0,2)[circle, fill, inner sep = .8pt]{};
\node at (2,0)[circle, fill, inner sep = .8pt]{};
\node at (2,2)[circle, fill, inner sep = .8pt]{};
\node at (0,1)[anchor=east]{$\peq$};
\node at (2,1)[anchor=west]{$\peq$};
\node at (1,0)[anchor=north, align=center,minimum size=3em]{$R$};
\node at (1,2)[anchor=south]{$R$};
\end{tikzpicture}
\hspace{.3cm}
\begin{tikzpicture}
\node at (1,1)[align=center]{back--up};
\draw[line width=.35mm,-Straight Barb] (0,0) -- (2,0);
\draw[line width=.35mm](2,0) -- (2,2);
\draw[line width=.35mm,dashed,-Straight Barb] (0,0) -- (0,2)  -- (2,2);
\node at (0,0)[circle, fill, inner sep = .8pt]{};
\node at (0,2)[circle, fill, inner sep = .8pt]{};
\node at (2,0)[circle, fill, inner sep = .8pt]{};
\node at (2,2)[circle, fill, inner sep = .8pt]{};
\node at (0,1)[anchor=east]{$\peq$};
\node at (2,1)[anchor=west]{$\peq$};
\node at (1,0)[anchor=north, align=center,minimum size=3em]{$R$};
\node at (1,2)[anchor=south]{$R$};
\end{tikzpicture}
\hspace{.3cm}
\begin{tikzpicture}
\node at (1,1)[align=center]{forth--down};
\draw[line width=.35mm,-Straight Barb] (0,0) -- (0,2) -- (2,2);
\draw[line width=.35mm,dashed,-Straight Barb] (0,0) -- (2,0);
\draw[line width=.35mm,dashed](2,0)  -- (2,2);
\node at (0,0)[circle, fill, inner sep = .8pt]{};
\node at (0,2)[circle, fill, inner sep = .8pt]{};
\node at (2,0)[circle, fill, inner sep = .8pt]{};
\node at (2,2)[circle, fill, inner sep = .8pt]{};
\node at (0,1)[anchor=east]{$\peq$};
\node at (2,1)[anchor=west]{$\peq$};
\node at (1,0)[anchor=north, align=center,minimum size=3em]{$R$};
\node at (1,2)[anchor=south]{$R$};
\end{tikzpicture}
\caption{Confluence conditions}\label{figure:confluence}
\end{figure}
\Cref{figure:confluence} depicts the confluence conditions. A \term{bi-preorder} $\mathcal F=(W,\allowbreak\fallible W,\allowbreak{\peq},{\rel})$ is \term{forth--up confluent} (respectively, \term{back--up confluent}, \term{forth--down confluent}) if $\rel$ is \term{forth--up confluent} (respectively, \term{back--up confluent}, \term{forth--down confluent}) for $(W,\fallible W,{\peq})$.
\end{definition}

The notions of \altterm{forth--up confluent}{forth--up} and \altterm{forth--down confluent}{forth--down confluence} allow us to simplify the semantic clauses for $\ps$ and $\nec$, respectively.

\begin{lemma}\label{lemClassical}
Let $\mathcal M = (W,\fallible W,{\peq},{\rel},V)$ be any \term{bi-preordered model}, $w\in W$ and $\varphi\in\lanfull$.
\begin{enumerate}

\item\label{classical:one} If $\mathcal M$ is \term{forth--up confluent}, then $(\mathcal M,w) \models \ps\varphi$ if and only if $\exists v\ler w$ such that $(\mathcal M,v) \models \varphi$.

\item\label{classical:two} If $\mathcal M$ is \term{forth--down confluent}, then $(\mathcal M,w) \models \nec \varphi$ if and only if $\forall v\ler w$ we have $(\mathcal M,v) \models \varphi$.
\end{enumerate}
\end{lemma}

\begin{proof}
We prove the second claim; the first is proven similarly by dualising.
It follows immediately from the semantics of $\nec$ (and reflexivity of $\peq$) that if $(\mathcal M,w) \models \nec \varphi$ and $ v\ler w$, then $(\mathcal M,v) \models \varphi$.
Conversely, suppose that $\forall v\ler w$, we have $(\mathcal M,v) \models \varphi$.
Let $u\seq w$ and $u'\ler u$.
By \altterm{forth--down confluent}{forth--down confluence}, there is $w'\ler w$ with $w'\peq u'$.
By our assumption, $(\mathcal M,w' ) \models\varphi$.
By monotonicity of the satisfaction relation, $(\mathcal M,u' ) \models\varphi$.
Since $u$ and $u'$ were arbitrary subject to $u\seq w$ and $u'\ler u$, we conclude that $(\mathcal M,w ) \models \nec \varphi$.
\end{proof}

\begin{definition}\label{definition:frame_classes}
We define the following classes of \altterm{birelational frame}{birelational frames}.
\begin{enumerate}\label{testlabel}

\item The class of \altdefine{def:csff}{$\boldsymbol{\csf}$ frames} is the class of \term{back--up confluent} \altterm{bi-preorder}{bi-preorders}.

\item The class of \altdefine{def:isff}{$\boldsymbol{\isf}$ frames} is the class of \term{forth--up confluent}, \term{infallible} \altterm{def:csff}{{\csff}s}

\item The class of \altdefine{def:s4i}{$\boldsymbol{\sfi}$ frames} is the class of \altterm{forth--up confluent}{forth--up} and \term{forth--down confluent}, \term{infallible} \altterm{bi-preorder}{bi-preorders}.

\item The class of \altdefine{def:gs4}{$\boldsymbol{\gsf}$ frames} is the class of \altterm{upward linear}{upward-linear} \altterm{def:isf}{$\isf$} frames. 

\item   The class of \altdefine{def:gs4c}{$\boldsymbol{\gsfc}$ frames} is the class of \term{forth--down confluent} \altterm{def:gsf}{$\gsf$} frames. 
\end{enumerate}
\end{definition}

In the following sections, we will see that these names are appropriate, in the sense that the \term{valid} formulas of each class are precisely the logics they are named after.

\begin{remark}
Even though \altterm{def:csf}{$\csf$} and \altterm{def:isf}{$\isf$} share the same $\ps$-free axioms, $\neg\neg\nec p\to\nec\neg\neg p$ is derivable in \altterm{def:isf}{$\isf$} but not in \altterm{def:csf}{$\csf$} \cite{DasMDiamonds}.
Meanwhile, the $\ps$-free fragments of \altterm{def:csf}{$\csf$} and \altterm{def:sfi}{$\sfi$} coincide, since their $\ps$-free axioms are sound and complete for a class of \term{bi-preorder}[bi-preorders] that are both \term{back--up confluent}[back--up] and \term{forth--down confluent}; in fact, they may be assumed to satisfy the stronger property ${\rel} = {\peq\mathbin{;}\rel\mathbin{;}\peq}$, where $;$ is relational composition~\cite{Wolter1999}.
Figure~\ref{figS4Iremark} provides an example of an \altterm{def:sfi}{$\sfi$} model falsifying $\neg\neg\nec p\to\nec\neg\neg p$.
\end{remark}

\begin{figure}[h!]\centering
		\begin{tikzpicture}[node distance=1.5cm,auto, text width=2mm]
				\node[state,fill=lightgray,line width=0.4mm] (w)  {};
				\node[state,fill=tan,above of=w] (wp)  {$p$};
				\node[state,right of =w,fill=lightgray] (v)  {};
				\node[state,fill=tan,above of=v] (vp)  {$p$};
				\node[state,right of=vp,fill=lightgray] (vs)  {};
					
				\path[->,dashed] (w) edge[] node[] {} (wp);
				\path[->,dashed] (v) edge[] node[] {} (vp);
				\path[->,dashed] (v) edge[] node[] {} (vs);

				\path[->] (w) edge[] node[] {} (v);
				\path[->] (wp) edge[] node[] {} (vp);
	\end{tikzpicture}	
	\caption{\altterm{def:sfi}{$\sfi$} model falsifying $\neg\neg\nec p\to\nec\neg\neg p$ at the evaluation point (indicated with bold outline). Dashed lines represent the $\peq$ relation while solid ones represent the $\rel$ relation, with reflexive arrows omitted. Worlds where $p$ is true are labelled and filled with tan while those were $p$ is false are empty and in light grey. }\label{figS4Iremark}
\end{figure}

The prime interest in G\"odel logics is arguably as fuzzy, rather than superintuitionistic, logics, as can be seen by their alternative real-valued semantics.

\begin{definition}
 A \altdefine{real-valued frame}{real-valued (Kripke) frame} is a pair $\mathcal F = (W,R)$,
	where $ W $ is a set and ${R} \from W \times W \mapsto [0,1]$  is a function. It is \define{reflexive} if for all $w \in W$ we have $R(w,w) = 1$. It is \define{transitive} if for all $u,v,w\in W$ we have $R(u, w) \geq \min\{R(u, v),R(v, w)\}$. 
	If ${R}\from W \times W \mapsto \lbrace 0,1\rbrace$ we will say the frame is \define{accessibility-crisp}[accessibility crisp].

	A \altdefine{g_valuation}{G\"odel valuation} on $\mathcal F$ is a function $V_\cdot(\cdot) \from W\times \lanfull  \to \mathbb R$ such that, for all $w\in  W$,
	\begin{align*}
	V_w(\bot)&=0,\\
	V_w(\varphi\wedge\psi ) &= \min \{ V_w(\varphi  ), V_w( \psi ) \}, \\
	V_w(\varphi \vee \psi ) &= \max \{ V_w(\varphi ) , V_w( \psi ) \},\\
		V_w(\varphi \imp \psi )& =\begin{cases}
			1 & \text{if }V_w(\varphi  ) \le V_w(\psi)\\
			V_w(\psi  )&  \text{otherwise}
		\end{cases}\\
		V_w(\ps\varphi )&=\sup_{v \in W}\lbrace \min\lbrace R(w, v) ,  V_v(\varphi )\rbrace \rbrace,\\
		V_w(\nec\varphi )&=\inf_{v \in W} \lbrace    1 \hbox{ if } R(w, v) \le V_v(\varphi);  V_v(\varphi) \hbox{ otherwise}   \rbrace .
	\end{align*}
	A real-valued Kripke frame equipped with a \altterm{g_valuation}{G\"odel valuation} is a \altdefine{gk_model}{G\"odel--Kripke model}, and an \term{accessibility-crisp} frame equipped with a \altterm{g_valuation}{G\"odel valuation} is an \define{accessibility-crisp model}.
\end{definition}

\begin{definition}
Let $\mathfrak C$ be any class of \altterm{gk_model}{G\"odel--Kripke models} or class of \altterm{real-valued frame}{real-valued frames}. Let $\Gamma$ be a set of formulas and $\varphi$ a formula. We write $\Gamma \models_\mathfrak C \varphi$  and say that $\varphi$ is a \define{local real-valued semantic consequence} of $\Gamma$ if, for each model $\mathcal M$ from $\mathfrak C$ and each $w \in \mathcal M$, we have \begin{equation}\label{lsc}\forall \psi \in \Gamma,V_w(\psi) = 1  \implies V_w(\varphi) = 1.\end{equation}

We say that $\varphi$ is \define{valid} on $\mathfrak C$ if $\models_\mathfrak C \varphi$ (that is, $\varnothing \models_\mathfrak C \varphi$), and \define{falsifiable} otherwise. This terminology extends also to single models or frames.
\end{definition}

The definition \eqref{lsc} of \term{local real-valued semantic consequence} is taken from \cite{CaicedoBiModal}. By a simple `truncate the truth values' argument, over classes of \emph{frames} this definition is equivalent to the \emph{a priori} stronger condition $V_w(\varphi) \geq \inf \{V_w(\psi) \mid \psi \in \Gamma\}$.

The logics \altterm{def:gsf}{$\gsf$} and \altterm{def:gsfc}{$\gsfc$} are known to be \term{strongly complete} with respect to \term{local real-valued semantic consequence} on classes of \altterm{real-valued frame}{real-valued frames}.

\begin{theorem}[strong completeness  for real-valuded semantics]\label{unit_completeness}\ \\*
\noindent(Caicedo and Rodr\'iguez). Let $\models_{\mathbb R\mathsf{S4}}$ denote \term{local real-valued semantic consequence} on the class of \term{reflexive} \term{transitive} \altterm{real-valued frame}{real-valued frames}, and $\vdash_\gsf$ denote the \term{syntactic consequence} relation for the logic \altterm{def:gsf}{$\gsf$}. Then
\begin{equation}\tag{\cite[Theorem~5.1(ii)]{CaicedoBiModal}}\Gamma \models_{\mathbb R\mathsf{S4}} \varphi \iff \Gamma \vdash_\gsf \varphi.\end{equation}

\noindent(Rodr\'iguez and Vidal). Let $\models_{\mathsf c\mathbb R\mathsf{S4}}$ denote \term{local real-valued semantic consequence} on the class of \term{reflexive} \term{transitive} \term{accessibility-crisp} frames, and $\vdash_\gsfc$ denote the \term{syntactic consequence} relation for the logic \altterm{def:gsfc}{$\gsfc$}. Then
\begin{equation}\tag{\cite[Theorem~5.4]{RodriguezV21}}\Gamma \models_{\mathsf c\mathbb R\mathsf{S4}} \varphi \iff \Gamma \vdash_\gsfc \varphi.\end{equation}
\end{theorem}

Note that \Cref{unit_completeness} would not be true if the set of propositional variables were uncountable, even for the modality-free fragment \cite[Proposition~3.1]{CaicedoBiModal}.

\subsubsection*{History of the logics} The logic \altterm{def:csf}{$\csf$} was introduced in \cite{AlechinaMPR01} motivated by syntactic considerations. In the same paper, the class of \altterm{def:csff}{{\csff}s} was defined in order to provide a semantics. A proof of the \term{strongly complete}[strong completeness] of \altterm{def:csf}{$\csf$} with respect to the \altterm{def:csff}{{\csff}s} is sketched in \cite{AlechinaMPR01}. We provide a full proof of this fact in \Cref{secCompCS4}. The logic \altterm{def:isf}{$\isf$} was introduced by Fischer Servi in \cite{servi1984axiomatizations}, where the logic was proven complete with respect to \altterm{def:isf}{$\isf$} frames. The logics \altterm{def:gsf}{$\gsf$} 
 and \altterm{def:gsfc}{$\gsfc$} were introduced in \cite{CaicedoBiModal}\footnote{The axioms and rules in \cite[Theorem~5.1(ii)]{CaicedoBiModal} are slightly different to those we use to define \altterm{def:gsf}{$\gsf$}, but are proven equivalent in \cite[Lemma~2.1]{CaicedoBiModal}.} and \cite{RodriguezV21}, respectively, in order to obtain the completeness results of \Cref{unit_completeness}.

Regarding \altterm{def:sfi}{$\sfi$}, it is easy to check that $(W,{\peq},{\rel})$ is an \altterm{def:sfi}{$\sfi$} frame if and only if $(W,{\rel},{\peq})$ is an \altterm{def:isf}{$\isf$} frame, so this could be viewed as `commuting' the roles of $\mathsf{S4}$ and intuitionistic logic (which we may denote $\mathsf I$).
In fact, in the context of expanding products of modal logics, \altterm{def:isf}{$\isf$} frames are similar to $\mathsf I \times^e \mathsf{S4}$ frames, where $\times^e$ is the `expanding product' as defined in \cite{pml}, and similarly, \altterm{def:sfi}{$\sfi$} frames can be regarded as $\mathsf{S4} \times^e \mathsf I $ frames. Additionally, in view of \Cref{lemClassical}, the \altterm{def:sfi}{$\sfi$} conditions give us the desirable property of being able to evaluate $\ps$ and $\nec$ classically.

\section{Soundness}\label{secSound}

In this section we establish the soundness of the logics \altterm{def:csf}{$\csf$}, \altterm{def:isf}{$\isf$}, \altterm{def:sfi}{$\sfi$}, \altterm{def:gsf}{$\gsf$}, and \altterm{def:gsfc}{$\gsfc$}, with respect to the corresponding classes of \altterm{bi-preorder}{bi-preorders}. Using the soundness results, we also note that these five logics are all distinct and related exactly as shown in \Cref{fig:logs}.

It will be useful distinguish the following type of frame. 
We call a frame \altdefine{regular}{$\ps$-regular} if it is \term{forth--up confluent} and \term{infallible}. Accordingly, a logic is \altdefine{regularlogic}{$\ps$-regular} if it includes \altterm{def:csf}{$\csf$} and contains Axioms \ref{ax:dp} and \ref{ax:null}.
Essentially, on \altterm{regular}{$\ps$-regular} frames, the modality $\ps$ works as one would expect from classical modal logic. The only logic we consider that is not \altterm{regularlogic}{$\ps$-regular} is \altterm{def:csf}{$\csf$}.

We first list some formulas that are \term{valid} over frame classes that need not be \altterm{regular}{$\ps$-regular}.
Of course the validities of any class of \altterm{birelational frame}{birelational frames} are closed under substitution.

\begin{proposition}\label{prop:val-ity0}\leavevmode
\begin{enumerate}[label=(\arabic*)]

\item\label{it:val-ity0a} Axioms~\ref{ax:k:box} and~\ref{ax:k:dia} are \term{valid}, and the inference rules~\ref{rl:mp} and~\ref{rl:nec} preserve validity, over any class of \altterm{birelational frame}{birelational frames}.

\item\label{it:val-ity0b} Axioms~\ref{ax:ref:box},~\ref{ax:ref:dia}, and~\ref{ax:trans:dia} are \term{valid} over the class of \altterm{bi-preorder}{bi-preorders}.

\item\label{it:val-ity1}
Axiom~\ref{ax:trans:box} is \term{valid} over any \term{bi-preorder} that is either \term{back--up confluent} or \term{forth--down confluent}.

\item\label{it:val-itygd}  Axiom~\ref{ax:g} is \term{valid} over the class of \altterm{upward linear}{upward-linear} frames.
\end{enumerate}

\end{proposition}

\begin{proof}
Items \ref{it:val-ity0a} and \ref{it:val-ity0b} are standard (and see for example~\cite{Simpson94}).
We check only~\ref{ax:trans:dia}. Let $\mathcal M = (W,\fallible W,{\peq},{\rel},V)$ be any \term{bi-preordered model}.
Suppose that $(\mathcal M,w) \models \ps\ps p$ and let $v\seq w$.
Then there exists $u\ler v$ such that $(\mathcal M,u) \models \ps p $.
Then since $u\seq u$, there exists $u'\ler u$ such that $(\mathcal M,u') \models p $.
By transitivity, $u'\ler v$, and as $v$ was arbitrary subject to $v\seq w$, we conclude that $(\mathcal M,w) \models \ps p$.

For item \ref{it:val-ity1}, the case for \term{back--up confluent}[back--up confluence] is known, as such a frame is a \altterm{def:csff}{\csff}~\cite{AlechinaMPR01}.
If instead $\mathcal M = (W,\fallible W,{\peq},{\rel},V)$ is \term{forth--down confluent}, we have by \Cref{lemClassical} that for any $w\in W$, $(\mathcal M,w) \models\nec p$ if and only if $\forall v\ler w$, $(\mathcal M,v) \models p$.
Using this characterisation, we may reason as in the classical case to deduce $\mathcal M\models\nec p \imp\nec\nec p$ from the transitivity of $\rel$.

Item \ref{it:val-itygd} is also well known
, but we provide a proof.
Assume for a contradiction that $\ref{ax:g}$ is not \term{valid} on the class of \term{def:gs4}[${\gsf}$ frames].
		This means that  $(\mathcal M,w) \not \models p \imp q$ and  $(\mathcal M,w)\not  \models  p \imp q$ for some  model $\mathcal M$ based on a \term{def:gs4}[${\gsf}$ frame] and some $w\in W$.
		From the former assumption it follows that there is $v\seq w$ such that $(\mathcal M,v)  \models p$ and  $(\mathcal M,v) \not \models q$.
		From the latter assumption it follows that there is $v'\seq w$ such that $(\mathcal M,v')  \models q$ and  $(\mathcal M,v') \not \models p$.
		Since ${\mathcal M}$ is \term{upward linear}, we need to consider two cases: if $v \peq v'$ we get that  $(\mathcal M,v')  \models p$; if $v'\peq v$ we conclude that  $(\mathcal M,v)  \models q$. In either case we reach a contradiction.
\end{proof}	

Note that with \ref{it:val-ity0a}, \ref{it:val-ity0b}, and the \term{back--up confluent}[back--up] part of \ref{it:val-ity1}, we have confirmed the soundness of the logic \altterm{def:csf}{$\csf$} with respect to the class of \altterm{def:csff}{{\csff}s} (and thus with respect to the subclasses of \term{def:isff}[$\isf$ frames], \term{def:gs4}[$\gsf$ frames], and \term{def:gs4c}[$\gsfc$ frames]). With \ref{it:val-ity0a}, \ref{it:val-ity0b}, and the \term{forth--down confluent}[forth--down] part of \ref{it:val-ity1}, we have \emph{also} shown the soundness of  \term{def:csf}[$\csf$] with respect to the class of \term{def:isff}[$\sfi$ frames], even though this class is not a subclass of the \altterm{def:csff}{{\csff}s}. In summary, our minimal logic \term{def:csf}[$\csf$] is sound with respect to all five of the frame classes of \Cref{definition:frame_classes}.

Next we focus our attention on \altterm{regular}{$\ps$-regular} classes.

\begin{proposition}\label{prop:val-ity1}
Let $\mathfrak C$ be any class of \altterm{regular}{$\ps$-regular} \altterm{bi-preorder}{bi-preorders}.

\begin{enumerate}[label=(\alph*)]

\item\label{it:val-ityn} Axiom~\ref{ax:null} is \term{valid} on $\mathfrak C$.

\item\label{it:val-itydp} Axiom~\ref{ax:dp} is \term{valid} on $\mathfrak C$.

		\item\label{itOne}  If the frames of $\mathfrak C$ are \term{back--up confluent}, then Axiom \ref{ax:fs} is \term{valid} on $\mathfrak C$.

		\item\label{itFour}   If the frames of $\mathfrak C$ are \term{forth--down confluent}, then Axiom $\ref{ax:cd}$ is \term{valid} on $\mathfrak C$.

\end{enumerate}

\end{proposition}

\begin{proof}
Item \ref{it:val-ityn} is a straightforward consequence of infallibility. Item \ref{it:val-itydp} is easily derived from \Cref{lemClassical}.\ref{classical:one}, since satisfaction of $\ps (p\vee q)$ requires a single witness, which is then sufficient to witness either $\ps p$ or $\ps q$. Item \ref{itOne} is proven in \cite{Simpson94}.

For \ref{itFour}, assume that $\mathcal M$ is \altterm{forth--up confluent}{forth--up} and \term{forth--down confluent}, and that $(\mathcal M,w) \models \nec (\varphi \vee \psi)$.
If $(\mathcal M,w)  \models \nec \varphi$, there is nothing to prove, so assume that  $(\mathcal M,w) \not \models \nec \varphi$.
		From \Cref{lemClassical}.\ref{classical:two} it follows that there is $v \ler w$ such that $(\mathcal M,v) \not \models \varphi$.
		But $(\mathcal M,v)  \models \varphi \vee\psi$, so $(\mathcal M,v)  \models  \psi$, and from \Cref{lemClassical}.\ref{classical:one} we obtain $(\mathcal M,w)  \models \ps \psi$.
\end{proof}

Given the soundness of \altterm{def:csf}{$\csf$} (with respect to any of our five frame classes), to verify the soundness of the remaining logics with respect to their corresponding frame classes, it suffices to check the remaining axioms in the definition of the logics. Examining \Cref{definition:additional_axioms}, we see that we have soundness of:
\begin{description}
\item[$\isf\phantom{^c}$]by \ref{it:val-ityn}, \ref{it:val-itydp}, and \ref{itOne},
\item[$\sfi\phantom{^c}$]by \ref{it:val-ityn}, \ref{it:val-itydp}, and \ref{itFour},
\end{description}
then building on the soundness of \altterm{def:isf}{$\isf$}:
\begin{description}
\item[$\gsf\phantom{^c}$]by \Cref{prop:val-ity0}\ref{it:val-itygd},
\end{description}
then
\begin{description}
\item[\altterm{def:gsfc}{$\gsfc$}]by \ref{itFour}.
\end{description}
Thus we obtain the following.

\begin{corollary}[soundness]\label{cor:sound}
Each of the logics \altterm{def:csf}{$\csf$}, \altterm{def:isf}{$\isf$}, \altterm{def:sfi}{$\sfi$}, \altterm{def:gsf}{$\gsf$}, and \altterm{def:gsfc}{$\gsfc$} is sound for its respective class of frames, as specified in \Cref{definition:frame_classes}.
\end{corollary}

To end this section, we now verify that the inclusions among our logics are exactly as depicted by the Hasse diagram in \Cref{fig:logs}. With the exception of comparisons involving  \altterm{def:sfi}{$\sfi$}, these relationships are all well known. First we show that the marked inclusions are proper. In each case, we exhibit a formula contained in the larger logic but falsified on the frame class corresponding to smaller logic. By soundness of the smaller logic, it does not contain the formula.
{\parindent=0pt\paragraph{\underline{$\isf \not\subseteq \csf$ and $\sfi \not\subseteq \csf$}}
By definition, $\ref{ax:null} \in \isf$ and $\ref{ax:null} \in \sfi$. Consider the \term{bi-preorder} $(\{i,f\}, \{f\}, \{(i,i),\allowbreak (f,f)\},\allowbreak \{(i,i), (f,f), (i,f)\})$ depicted on the top left of \Cref{fig:false_n}. This frame is \term{back--up confluent} and thus a \altterm{def:csff}{\csff}. The \term{infallible} world $i$ satisfies $\ps\bot$ (under any \term{valuation}), and so falsifies $\neg\ps\bot$. 

\paragraph{\underline{$\gsf \not\subseteq \isf$ and $\gsfc \not\subseteq \sfi$}}
By definition, $\ref{ax:g} \in \gsf$ and $\ref{ax:g} \in \gsfc$. Consider the \term{infallible} \term{bi-preordered model} depicted on the top right of \Cref{fig:false_n}, where $p$ and $q$ are propositional variables. As $\rel$ is trivial, the frame validates all confluence conditions and in particular is both an \altterm{def:isf}{$\isf$} frame and an \altterm{def:sfi}{$\sfi$} frame. The world at the bottom falsifies $(p \imp q) \vee (q \imp p)$.

\paragraph{\underline{$\gsfc \not\subseteq \gsf$}}
By definition, $\ref{ax:cd} \in \gsfc$. Consider the \term{infallible} \term{bi-preordered model} depicted on the bottom left of \Cref{fig:false_n}. It is easy to check that the frame is \term{back--up confluent}, \term{forth--up confluent}, and \term{upward linear}, and thus is a \altterm{def:gsf}{$\gsf$} frame. The world at the bottom satisfies $\nec(p \wedge q)$ but neither $\nec p$ nor $\ps q$, and thus falsifies $\ref{ax:cd}$.
}
\medskip

Finally we verify the incomparability of \altterm{def:sfi}{$\sfi$} with either \altterm{def:isf}{$\isf$} or \altterm{def:gsf}{$\gsf$}.

{\parindent=0pt\paragraph{\underline{$\sfi \not\subseteq \gsf$}} By definition, $\ref{ax:cd} \in \sfi$. As we have just seen, $\ref{ax:cd} \not\in \gsf$.

\paragraph{\underline{$\isf \not\subseteq \sfi$}} By definition, $\ref{ax:fs} \in \isf$. Consider the \term{infallible} \term{bi-preordered model} depicted on the bottom right of \Cref{fig:false_n}. It is easy to check that the frame is \term{back--up confluent} and \term{forth--up confluent}, and thus is a \altterm{def:sfi}{$\sfi$} frame. The bottom left world falsifies $\ps p$ and thus satisfies $\ps p \imp \nec q$. But the same world does not satisfy $\nec (p \imp q)$ and thus falsifies \ref{ax:fs}.
}

\begin{figure}
\centering
\begin{tikzpicture}
\draw (0,0)node[anchor=north]{$i$} -- (2,0)node[anchor=north]{$f$};
\node at (0,0)[circle, fill, inner sep = 1pt]{};
\node at (2,0)[circle, fill, inner sep = 1pt]{};
\node at (1,0)[anchor=north]{$\rel$};
\draw[dashed] (1.5,.5) -- (2.5,.5) -- (2.5,-.5) -- (1.5,-.5) -- cycle;
\node at (2.4,-.4)[anchor=north west]{$\fallible W$};
\end{tikzpicture}
\hspace{1cm}
\begin{tikzpicture}
\draw (-1,2)node[anchor=east]{$p$} -- (0,0) -- (1,2)node[anchor=west]{$q$};
\node at (0,0)[circle, fill, inner sep = 1pt]{};
\node at (-1,2)[circle, fill, inner sep = 1pt]{};
\node at (1,2)[circle, fill, inner sep = 1pt]{};
\node at (-.5,1)[anchor= east]{$\peq$};
\node at (.5,1)[anchor= west]{$\peq$};
\end{tikzpicture}

\vspace{1cm}
\begin{tikzpicture}
\draw (0,0)node[anchor=east]{$p$} --(0,2) node[anchor=east]{$p,q$} -- (2,2)node[anchor=west]{$q$};
\node at (0,0)[circle, fill, inner sep = 1pt]{};
\node at (0,2)[circle, fill, inner sep = 1pt]{};
\node at (2,2)[circle, fill, inner sep = 1pt]{};
\node at (0,1)[anchor= east]{$\peq$};
\node at (1,2)[anchor= north]{$\rel$};
\end{tikzpicture}
\hspace{1.5cm}
\begin{tikzpicture}
\draw (0,0)node[anchor=east]{} --(2,0) node[anchor=east]{} -- (2,2)node[anchor=west]{$p$};
\node at (0,0)[circle, fill, inner sep = 1pt]{};
\node at (2,0)[circle, fill, inner sep = 1pt]{};
\node at (2,2)[circle, fill, inner sep = 1pt]{};
\node at (1,0)[anchor= south]{$\rel$};
\node at (2,1)[anchor= east]{$\peq$};
\end{tikzpicture}
\caption{Frames and models separating the logics. (Reflexive arrows not displayed)}\label{fig:false_n}
\end{figure}

\section{Strong completeness and finite model property for \texorpdfstring{{$\csf$}}{CS4}}\label{secCompCS4}

In this section, we prove that \altterm{def:csf}{$\csf$} is strongly complete with respect to its birelational semantics and also has the finite frame property. The strong completeness proof (\Cref{cs4_completeness}) uses a canonical model construction, while the finite frame property (\Cref{thmCS4fin}) uses the same type of construction, but \term{relativised canonical model}[relativised] to a finite set of formulas.
The finite frames are exponentially bounded in size, which implies decidability in \textsc{nexptime}.

It will be convenient to observe that the following formulas are  contained in \altterm{def:csf}{$\csf$} (and thus also its extensions  \altterm{def:isf}{$\isf$}, \altterm{def:sfi}{$\sfi$}, \altterm{def:gsf}{$\gsf$}, and \altterm{def:gsfc}{$\gsfc$}).
We leave the proofs to the reader.
\begin{proposition}\label{auxiliary}
The following formulas 
belong to \term{def:csf}[$\csf$].
\begin{enumerate}[label=(\arabic*)]
	\item\label{der:ax1} $\ps {\left( p \rightarrow q\right)} \rightarrow (\nec p \rightarrow \ps q)$,\footnote{This is the first of Fischer Servi's two \emph{connecting axioms} \cite{servi1984axiomatizations}.}	
	
	\item \label{der:ax3} $\ps ( p \wedge q) \rightarrow \ps p \wedge \ps q$,

	\item\label{der:box_dist_and} \mbox{$\nec p \wedge \nec q  \rightarrow \nec (p \wedge q)$},
	
	\item\label{der:ax2} \mbox{$\nec p \vee \nec q  \rightarrow \nec (p \vee q)$}.
\end{enumerate}
%
\end{proposition}

\subsection{Prime theories}

We first introduce the notions that will be useful not only for \altterm{def:csf}{$\csf$}, but also for the other logics considered later.
Fix a logic $\Lambda$. At the moment we only need to assume that $\Lambda$ is a logic (in the \term{intuitionistic modal language}[modal language] $\lanfull$) that includes intuitionistic logic.

Let $\Sigma\subseteq \lanfull$ be a set of formulas closed under subformulas.
A set $\Gamma\subseteq \Sigma $ is called a $\Lambda$-\define{theory over} $\Sigma$ if it is closed
under \term{syntactic consequence} within $\Sigma$ ($\Gamma \vdash \varphi \in \Sigma$ implies $\varphi \in  \Gamma$).
If $\Lambda$ is clear from context or if $\Sigma = \lanfull$ we may omit either of them, so that e.g.~a \define{theory} is a $\Lambda$-theory over $\lanfull$ (where the choice of $\Lambda$ is clear or irrelevant).
Further, we say that $\Gamma$ is \define{prime} if $\varphi \vee \psi \in  \Gamma$ implies that either $\varphi \in  \Gamma$ or $ \psi \in  \Gamma$.\footnote{Note that we do not require a prime theory to be a \emph{proper} subset of $\lanfull$, so $\lanfull$ is itself a prime theory.}

\begin{definition}[]\label{def:deduction} 
	Given sets $\Phi\subseteq \Sigma$ and $\Xi$ of formulas, we say that $\Phi$ is $\Xi$\define{-consistent} if 		$\Phi  \not \vdash   \Xi$.
	If 	$\Sigma\subseteq \lanfull$ is closed under subformulas, we say that $\Phi$ is \altdefine{maximal consistent}{$\Sigma$-maximal $\Xi$-consistent} if moreover $\Phi\subseteq\Sigma$, and whenever $\Phi\subsetneq \Phi' \subseteq\Sigma$, it follows that $\Phi'$ is not $\Xi$-consistent.

	We omit mention of $\Sigma$ when $\Sigma=\lanfull$ and mention of $\Xi$ if $\Xi = \varnothing$ (with the understanding that $\bigvee\varnothing \equiv \bot$), so that e.g.~$\Phi$ is maximal consistent if it is $\lanfull$-maximal $\varnothing$-consistent.
		If $\Xi$ is a singleton $\{\xi\}$, we write $\xi$\term{-consistent} instead of $\{\xi\}$\term{-consistent}.
\end{definition}

Note that if $\Phi$ is $\Xi$\term{-consistent}, then necessarily $\bot \not\in\Phi $. It is \term{prime}, rather than \term{maximal consistent}[maximal], \term{theory}[theories] that we are ultimately interested in, but \term{maximal consistent} sets give us a way of obtaining \term{prime} \term{theory}[theories].

\begin{lemma}\label{lem:prime_theory}
Let $\Sigma \subseteq \lanfull$ be closed under subformulas, and let $\Phi$ be a \term{maximal consistent}[$\Sigma$-maximal $\Xi$-consistent] set. Then \begin{enumerate}
	\item\label{one_theory} $\Phi$ is a theory;
	\item\label{two_prime} $\Phi$ is \term{prime}.
\end{enumerate}
\end{lemma}

\begin{proof}
	\eqref{one_theory}: Using some intuitionistic reasoning, it is easy to see that $\Phi$ is closed under $\vdash$ and is thus a \term{theory}. 
	
	\eqref{two_prime}: To see that $\Phi$ is \term{prime}, suppose that $\varphi \vee \psi \in \Phi $.
	We cannot have that $   \Phi  \cup \{ \varphi \}  $ and $   \Phi \cup \{\psi\}$ are both $\Xi$-in\term{consistent}, since by left disjunction introduction (formalisable as a theorem of intuitionistic logic) we would obtain that $   \Phi , \varphi \vee \psi $ is $\Xi$-in\term{consistent}.
	However, $   \Phi , \varphi \vee \psi $ is just $   \Phi  $, contrary to assumption.
	Thus either $   \Phi  , \varphi   $ or $   \Phi ,\psi$ is $\Xi$\term{-consistent}; say, $   \Phi , \varphi $.
	But then by maximality of $\Phi $ and the assumption that $\Sigma$ is closed under subformulas, we must already have $\varphi\in \Phi $, as required.
\end{proof}

	We say that $\Psi$ \define{extends} $\Phi$ if $\Phi \subseteq\Psi $.
	
\begin{lemma}[Lindenbaum lemma]\label{lem:lindenbaum}
Let $\Sigma \subseteq \lanfull$ be closed under subformulas.
Any $\Xi$\term{-consistent} set $\Phi\subseteq\Sigma$ of formulas can be \term{extends}[extended] to a ($\Sigma$\term{-maximal}) $\Xi$-\term{consistent} \term{prime} \term{theory} $ \Phi _* $.
\end{lemma}	
	
	\begin{proof}
	By enumerating all formulas in $\Sigma$ as $(\sigma_i)_{i < \alpha}$ (where $\alpha\leq\omega$) and successively adding $\sigma_i$ to $\Phi$ if the result is $\Xi$\term{-consistent}, one obtains a  set of formulas $ \Phi  _* \supseteq \Phi $. This $ \Phi  _*$ is $\Xi$-\term{consistent} due to the finitary definition of $\vdash$, and $ \Phi  _*$ is $\Sigma$-\term{maximal consistent}[maximal] by construction. 
	Since  $ \Phi  _*$ is \term{maximal consistent}[$\Sigma$-maximal $\Xi$-consistent], \Cref{lem:prime_theory} says that $ \Phi  _*$ is also a \term{prime} \term{theory}.
		\end{proof}



If $\Phi$ is a \term{theory} then we define 
\begin{align*}\bfrm\Phi &\coloneqq \{ \varphi \in \lanfull \mid \nec\varphi\in \Phi  \},\\ \dfrm\Phi &\coloneqq \{ \varphi \in \lanfull \mid \ps\varphi\notin \Phi  \}.
\end{align*}
These sets will be useful in defining the canonical models below.

\subsection{Augmented theories}

In order to prove that \altterm{def:csf}{$\csf$} has the finite model property, we generalize the completeness proof sketched in~\cite{AlechinaMPR01} to function relative to a set $\Sigma$ of formulas.
The proof proceeds by constructing a canonical model falsifying any formula that is not derivable.
If $\Sigma$ is the full language, we obtain strong completeness for \altterm{def:csf}{$\csf$}, but if we choose it to be finite and closed under subformulas then the resulting relativised canonical model is finite.
The main challenge in the construction is dealing with the `non-classical' behaviour of $\ps$ for \altterm{def:csf}{$\csf$}, which requires an additional technical treatment that will not be needed for $\ps$-regular logics which we consider later.

Fix a set of formulas $\Sigma$ closed under subformulas and containing $\nec\bot$.
The set $\Sigma$ could be infinite or finite, with the latter case leading directly to a finite model property proof.
Rather than working with standard \term{theory}[theories], we must work with \define{augmented theories over} $\Sigma$ (or simply \define{augmented theories} when $\Sigma$ is clear), which are ordered pairs of sets of formulas  $\Phi = (\Phi^+; \Phi^{\centernot\ps})$, where $\Phi^+$ is a \term{theory over} $\Sigma$ and $\Phi^{\centernot\ps}\subseteq\Sigma$.
We say that $\Phi$ is \define{prime2}[prime] if $\Phi^+$ is \term{prime}.
The intuition behind this definition is the following: formulas in $\Phi^+$ are the ones satisfied by the \term{theory}, and the formulas $\varphi$ in $\Phi^{\centernot\ps}$ are those  such that $\ps\varphi$ is falsified \emph{directly}, by $\rel$ not having any witnesses for $\varphi$ (as opposed to being falsified via a $ \peq  $-accessible world).
Any standard \term{theory over} $\Sigma$ may be regarded as an \term{augmented theories over}[augmented theory over] $\Sigma$ if we identify $\Phi$ with $(\Phi ;\Phi^{\centernot\ps})$, where $\Phi^{\centernot\ps} \coloneqq \{\varphi\in \Sigma\mid\ps\varphi\notin \Phi\}$. 
However, not all \term{augmented theories} are of this form, and in fact we often work with \term{augmented theories} of the form $(\Phi^+;\varnothing)$ (among others).

For convenience, the following sequence of definitions use the notation $(\Phi^+; \Phi^{\centernot\ps})$ used for \term{augmented theories}, even though they apply to arbitrary pairs of sets of formulas.

\begin{definition}[]\label{def:deduction2} 
	Given a set of formulas $\Xi$, we say that a pair $\Phi = (\Phi^+; \Phi^{\centernot\ps})$ of sets of formulas is \altdefine{-consistent}{$\Xi$-consistent} if for any finite set $\Delta \subseteq \Phi^{\centernot\ps}$,
		\begin{displaymath}
		 \Phi^+ \not \vdash   \Xi,  \ps\bigvee \Delta .
		\end{displaymath}
		We adopt the convention that $\ps {\bigvee \varnothing} \coloneqq  \bot$ (note that this convention overrides $\bigvee\varnothing = \bot$).
		We say that $\Phi$ is \define{consistent} if it is $\varnothing$\term{-consistent}, in other words, if for any finite set $\Delta \subseteq \Phi^{\centernot\ps}$,
		\begin{displaymath}
		 \Phi^+ \not \vdash   \ps\bigvee \Delta .
		\end{displaymath}
		If $\Xi$ is a singleton $\{\xi\}$, we write $\xi$\term{-consistent} instead of $\{\xi\}$\term{-consistent}.
\end{definition}

Note that if $\Phi$ is a $\Xi$\term{-consistent} \term{augmented theories}[augmented \term{theory}], then necessarily $\bot \not\in\Phi^+$.

We say that a pair of sets of formulas $\Psi$ \define{extends} $\Phi$ if $\Phi^+\subseteq\Psi^+$ and $\Phi^{\centernot\ps}\subseteq \Psi^{\centernot\ps}$.
Pairs of sets of formulas satisfy their own version of Lindenbaum's lemma.

\begin{lemma}[Lindenbaum lemma, adapted from~\cite{AlechinaMPR01}]\label{lem:lindenbaumGen}
Let $\Sigma\subseteq\lanfull$ be closed under subformulas.
Any $\Xi$\term{-consistent} pair of sets of formulas $(\Phi^+; \Phi^{\centernot\ps})$ with $ \Phi^+, \Phi^{\centernot\ps}\subseteq \Sigma$ can be extended to a $\Xi$\term{-consistent} \term{prime} \term{augmented theories over}[augmented \term{theory}] $(\Phi^+_*; \Phi^{\centernot\ps})$ over $\Sigma$.
\end{lemma}	
	
	\begin{proof}
Apply \Cref{lem:lindenbaum} to $\Phi^+$ with $\Xi' = \Xi\cup \{\ps{\bigvee \Delta}\mid \Delta\subseteq_{\mathrm{fin}} \Phi^{\centernot\ps}\}$, where $\subseteq_{\mathrm{fin}}$ is the finite subset relation.
		\end{proof}

		\subsection{Relativised canonical models for \texorpdfstring{$\csf$}{CS4}}

We will use augmented theories to define canonical models for \term{def:csf}[$\csf$].
To be precise, each $\Sigma$ closed under subformulas gives rise to a distinct version of the canonical model. In particular, if $\Sigma$ is finite, then this model is of size exponential in $|\Sigma|$.
If $\Phi$ is an \term{augmented theories}[augmented \term{theory}] then we define $\bfrm\Phi \coloneqq \{ \varphi \in \lanfull \mid \nec\varphi\in \Phi^+ \}$.

\begin{definition}
Let $\Sigma\subseteq\lanfull$ contain $\nec\bot$ and be closed under subformulas (so also $\bot \in \Sigma$). Let $\rcm {W^{\centernot\bot}}\Sigma$ be the set of all \term{prime} \term{augmented theories over}[augmented \altterm{def:csf}{$\csf$}-theories over] $\Sigma$. 
We define the $\Sigma$-\define{relativised canonical model} for \term{def:csf}[$\csf$] as $\rcm{\mathcal M }\Sigma ^{{\csf}}=(\rcm W\Sigma ,\rcm{\fallible W }\Sigma ,{\rcm\peq\Sigma} ,{\rcm\rel \Sigma }, \rcm V \Sigma )$, where~
\begin{itemize}

	\item $\rcm W\Sigma = \rcm {W^{\centernot\bot}}\Sigma \cup \rcm {\fallible W}\Sigma  $,
	
	\item $\rcm {\fallible W}\Sigma  = \lbrace (\Sigma;\varnothing)\rbrace$, 

	\item ${\rcm \peq \Sigma } \subseteq \rcm W\Sigma \times \rcm W\Sigma $ is defined by $\Phi \rcm\peq\Sigma \Psi$ if and only if $\Phi^+ \subseteq \Psi^+$,

	\item ${\rcm\rel\Sigma} \subseteq \rcm W\Sigma \times \rcm W\Sigma $ is defined by $\Phi \rcm\rel\Sigma \Psi$ if and only if $\bfrm\Phi\subseteq \bfrm\Psi $ and $\Phi^{\centernot\ps} \subseteq \Psi^{\centernot\ps}$, 

	\item $\rcm V\Sigma $ is defined by $\rcm V\Sigma (p)  = \lbrace \Phi \in \rcm W\Sigma \mid p \in \Phi^+  \rbrace$.

\end{itemize}
\end{definition}

This construction is a variation of a standard canonical model, with two key distinctions.
First, we only work with formulas in $\Sigma$, which when $\Sigma = \lanfull$ will yield the usual construction but for finite $\Sigma$ will immediately produce finite structures.
Second, the constructive semantics of $\Diamond$ mean that we may have $ \Phi \rcm\rel\Sigma \Psi$ and $(\rcm{\mathcal M}\Sigma^{\csf},\Psi)\models \varphi$ but $(\rcm{\mathcal M}\Sigma^{\csf},\Phi) \not \models \varphi$.
In such a case, we need $\tilde\Phi\rcm\seq\Sigma\Phi$ witnessing this failure of $\ps\varphi$, and this $\tilde \Phi$ is obtained by modifying the `$\centernot\ps$-component'.
In this case, we may simply set $\tilde\Phi = (\Phi^+,\{\varphi\}) $.
This is the reason for working with \term{augmented theories}.

We now show that $\rcm {\mathcal M^{\csf}}\Sigma$ is indeed a \term{def:csff}[$\csf$ model], meaning the underlying frame is both a \term{bi-preorder} and \term{back--up confluent}.

\begin{lemma}\label{lemCS4IsModel}
If $\Sigma\subseteq\lanfull$ is closed under subformulas, the $\Sigma$-\term{relativised canonical model} $\rcm {\mathcal M^{\csf}}\Sigma$ is indeed a \term{bi-preordered model}.
\end{lemma}

\begin{proof}
	It is straightforward to check that the singleton $\rcm{\fallible W}\Sigma$ is closed under ${\rcm\peq\Sigma}$. To see that $\rcm{\fallible W}\Sigma$ is also closed under ${\rcm\rel \Sigma }$, note that if $(\Sigma;\varnothing) \rcm\peq\Sigma \Psi$, then $(\Sigma;\varnothing)^\nec \subseteq \Psi^\nec$, and since $\nec\bot \in \Sigma$ this gives $\nec\bot \in \Psi^+$. From Axiom~\ref{ax:ref:box} and $\nec\bot \in \Psi^+$ we obtain $\bot \in \Psi^+$, but then we must have $\Psi^+ = \Sigma$, so $\Psi = (\Sigma;\varnothing)$.   
The monotonicity property of the valuation $\rcm V\Sigma $ is also easily verified. 
Since $\rcm\peq\Sigma$ is defined by $\Phi \rcm\peq\Sigma \Psi$ if and only if $\Phi^+ \subseteq \Psi^+$, and $\subseteq$ is itself a preorder, we see that $\rcm\peq\Sigma$ is a preorder.
Similarly, $\Phi \rcm\rel\Sigma \Psi$ if and only if $ \Phi^\nec\subseteq \Psi^\nec$ and $ \Phi^{\centernot\ps} \subseteq \Psi^{\centernot\ps}$, which is again a preorder since $\subseteq$ is.
\end{proof}

\begin{lemma}\label{prop:back--up}
For $\Sigma\subseteq\lanfull$ closed under subformulas, the $\Sigma$-\term{relativised canonical model} $\rcm {\mathcal M^{\csf}}\Sigma$ is  \term{back--up confluent}.
In particular, if $\Phi \rcm\rel\Sigma \Psi \rcm\peq\Sigma \Omega$, then $\Upsilon \coloneqq  (\Phi^+;\varnothing) $ satisfies $\Phi \rcm\peq\Sigma \Upsilon \rcm\rel\Sigma \Omega$.
\end{lemma}

\begin{proof}
Take $\Phi$, $\Psi$, and $\Omega$ in $\rcm W\Sigma $ such that $\Phi \rcm\rel\Sigma \Psi \rcm\peq\Sigma \Omega$, and let us define $\Upsilon = \left( \Phi^+; \varnothing \right)$.
Now, $\Phi^+$ is already \term{prime}, so $\Upsilon \in \rcm W\Sigma $. By definition $\Phi^+ \subseteq \Upsilon^+$; thus $\Phi \rcm\peq\Sigma \Upsilon$. 
Take $\nec \varphi \in \Upsilon^+$. By definition $\nec \varphi \in \Phi^+$. Since $\Phi \rcm\rel\Sigma \Psi \rcm\peq\Sigma \Omega$, it follows that $\nec \varphi \in \Omega^+$.
Since $\nec \varphi$ was chosen arbitrarily, we have $ \Upsilon^\nec \subseteq \Omega^\nec$.
Clearly also $\Upsilon^{\centernot\ps} = \varnothing \subseteq \Omega^{\centernot\ps}$, so $  \Upsilon \rcm\rel\Sigma \Omega$.
\end{proof}

We have established that $\rcm {\mathcal M^{\csf}}\Sigma$ is always a \term{def:csff}[$\csf$ model]. We now proceed towards a truth lemma for $\rcm {\mathcal M^{\csf}}\Sigma$. Here is where we really put the axioms and inference rules to work.

\begin{lemma}\label{lemDiamCS4}
Let $\Sigma\subseteq\lanfull $ be closed under subformulas.
If $\Gamma$ is a \term{prime} \term{augmented theories}[augmented \altterm{def:csf}{$\csf$}-theory over] $\Sigma$  and $\varphi \in \lanfull$, the following hold.
	\begin{enumerate}[label=(\alph*)]
		\item\label{item:a}  $\ps\varphi\in \Gamma^+$ if and only if $\ps\varphi\in \Sigma$ and for all $\Psi\seq_{\mathrm c}  \Gamma$ there is $\Delta\in \rcm W\Sigma $ such that $\Psi \rcm\rel\Sigma \Delta$ and $ \varphi\in \Delta^+$.
		
		\item\label{item:b} $\nec\varphi\in \Gamma^+$ if and only if  $\nec\varphi\in \Sigma$ and for all $\Psi$ and $\Delta$ such that $\Gamma \rcm\peq\Sigma \Psi\rcm\rel\Sigma \Delta$, we have $ \varphi\in \Delta^+$.
	\end{enumerate}	
\end{lemma}

\begin{proof}\emph{\ref{item:a}.} From left to right, assume that $\ps\varphi \in \Gamma^+$, and let $\Psi\seq_{\mathrm c} \Gamma$. We may assume that $\Psi$ is \term{-consistent}[consistent], otherwise $\Psi = (\Sigma;\varnothing)$ and we can just choose $\Delta$ to also be $(\Sigma;\varnothing)$. 
	Consider the pair $\Upsilon \coloneqq (  \nec \bfrm\Psi, \varphi ;  \Psi^{\centernot\ps})$ of sets of formulas (where $\nec \Theta  := \{\nec\theta:\theta\in \Theta\}$). 
	We show that $\Upsilon$ is \term{consistent}.
	If not, let $\chi_1,\ldots,\chi_n \in  \bfrm\Psi$ and $\psi_1,\ldots,\psi_m \subseteq \Psi^{\centernot\ps}$ be such that, for $\chi \coloneqq \bigwedge_i \chi_i$ and $\psi \coloneqq \bigvee_i \psi_i$, we have $  \chi \wedge \varphi \imp \ps \psi \in {\csf}$, so that $  \chi \imp\left( \varphi \imp \ps\psi\right) \in {\csf}$.
	By~\ref{rl:nec},~\ref{ax:k:box}, and~\ref{rl:mp}, it follows that $\nec \chi \imp \nec (\varphi \imp \ps\psi) \in \csf$. By \Cref{auxiliary}\ref{der:box_dist_and} (and induction up to $n$), we also know that $\bigwedge_i\nec\chi_i \imp \nec \chi \in \csf$, and hence $\bigwedge_i\nec\chi_i \imp \nec (\varphi \imp \ps\psi) \in {\csf}$, or otherwise stated, $\nec\chi_1, \dots,\nec\chi_n \vdash \nec (\varphi \imp \ps\psi)$. 
	We thus obtain $  \Psi^+ \vdash \nec (\varphi \allowbreak\imp \allowbreak\ps\psi)$, and then by \ref{ax:k:dia}, $  \Psi^+ \vdash \ps\varphi \imp \ps\ps\psi $.
	Since $\Gamma \rcm\peq\Sigma \Psi$ and $\ps \varphi \in \Gamma^+$, we also have $\ps \varphi \in \Psi^+$.
We know that $\ps \varphi \wedge (\ps\varphi \imp \ps\ps\psi) \imp \ps\ps\psi \in {\csf}$ by substitution on an intuitionistic theorem. Thus, $  \Psi^+ \vdash \ps \ps \psi$ and by \ref{ax:trans:dia}, $\Psi^+ \vdash \ps \psi$, which contradicts the consistency of $\Psi$. 
	We conclude that $\Upsilon$ is \term{consistent}.
	By the \term{lem:lindenbaumGen}[Lindenbaum \Cref{lem:lindenbaumGen}], $\Upsilon$ can be extended to a \term{consistent} \term{prime} \term{augmented theories}[augmented theory] $\Delta= (\Delta^+;\Delta^{\centernot\ps})$ such that $ \bfrm\Psi \subseteq \Delta^+$ and $\Psi^{\centernot\ps} = \Delta^{\centernot\ps}$ (therefore $\Psi \rcm\rel\Sigma \Delta$), and moreover $\varphi \in \Delta^+$, as needed.
	
	Conversely, let us assume that (for $\ps \varphi \in \Sigma$), we have $\ps \varphi \not\in \Gamma^+$ (so in particular $\Gamma$ is \term{consistent}) and let us define $\Psi = (\Gamma^+;\lbrace \varphi \rbrace)$.
 	It is easy to see that $\Psi$ is \term{consistent}, and since $\Gamma^+$ is \term{prime}, $\Psi \in \rcm W\Sigma $. 
	Moreover, $\Gamma \rcm\peq\Sigma \Psi$.
	We claim that for all $\Delta \in \rcm W\Sigma $, if $\Psi \rcm\rel\Sigma \Delta$, then 
	$\varphi \not \in \Delta^+$.
	To prove this, let us take any $\Delta \in \rcm W\Sigma $ satisfying $\Psi \rcm\rel\Sigma \Delta$. 
	By definition, $\Psi^{\centernot\ps} \subseteq \Delta^{\centernot\ps}$, so $\varphi \in \Delta^{\centernot\ps}$.
	This implies $\Delta \neq(\Sigma; \varnothing)$, and hence $\Delta$ is consistent. 
	By consistency of the \term{augmented theories}[augmented theory] $\Delta$, and using  $\varphi \in \Delta^{\centernot\ps}$ again, we see that $\ps \varphi \not \in \Delta^+$. Then $\varphi \not \in \Delta^+$ because of Axiom~\ref{ax:ref:dia}, as needed.

\paragraph{\ref{item:b}}
From left to right, let $\nec\varphi \in \Gamma^+$, and let $\Psi$ and $\Delta$ be such that $\Gamma \rcm\peq\Sigma \Psi \rcm\rel\Sigma \Delta$.  
	From $\Gamma \rcm\peq\Sigma \Psi \rcm\rel\Sigma \Delta$, we obtain $\Gamma^+ \subseteq \Psi^+$ and $ \bfrm\Psi \subseteq \bfrm\Delta $. Since $\nec \varphi \in \bfrm \Gamma$, we have $\varphi \in \bfrm\Delta$, i.e.~$\nec \varphi \in \Delta^+$.
But $\Delta^+$ is closed under \term{syntactic consequence} within $\Sigma$, and $\Sigma$ is closed under subformulas, so by \ref{ax:ref:box}, $ \varphi\in \Delta^+$.

	Conversely, let us assume that $\nec \varphi \not \in \Gamma^+$ but $\nec\varphi\in\Sigma$, and let us define $\Psi  = (\Gamma^+;\varnothing)$. Clearly $\Psi$ is a \term{prime} \term{theory} and satisfies $\Gamma \rcm\peq\Sigma \Psi$.
	Let us take $\Upsilon = ( \nec \bfrm\Psi , \varnothing )$. $\Upsilon$ is $\varphi$\term{-consistent}, since otherwise \ref{rl:nec} and \ref{ax:k:box} would yield $\nec\bigwedge_i\psi_i \imp \nec \varphi \in \csf$ for some $\psi_1, \dots, \psi_n \in \bfrm\Psi$, from which we can obtain $\nec\psi_1, \dots, \nec\psi_n \vdash \nec\varphi$ and thus $\nec \varphi \in \Gamma^+$---a contradiction.
	By the Lindenbaum \Cref{lem:lindenbaum}, $\Upsilon$ can be extended to a $\varphi$\term{-consistent} \term{prime} augmented \term{theory} $\Delta = (\Delta^+;\varnothing)$ such that $\Upsilon^+ \subseteq \Delta^+$. By the definition of $\Upsilon$, we then have $\Psi \rcm\rel\Sigma \Delta$ and $\varphi \not \in \Delta^+$, as needed.  
\end{proof}

\begin{lemma}[truth lemma for $\rcm {\mathcal M^{\csf}}\Sigma$]\label{lem:truth-lemma}
Let $\Sigma\subseteq\lanfull$ be closed under subformulas.
	For each \term{prime} \term{augmented theories}[augmented \altterm{def:csf}{$\csf$}-theory] $\Phi = (\Phi^+; \Phi^{\centernot\ps})$ and each formula $\varphi \in \Sigma$,
\[\varphi \in \Phi^+ \iff ( \rcm {\mathcal M^{\csf}}\Sigma, \Phi) \models \varphi.\]\end{lemma}
\begin{proof}
By structural induction on $\varphi$.	
	The case of propositional variables holds by the definition of $\rcm V\Sigma $. For $\bot$ we need that $\bot \in \Phi^+ \iff \Phi \in W^\bot$. This holds because the only fallible world is $(\Sigma;\varnothing)$ (and we insisted $\bot \in \Sigma$), and conversely the infallible worlds are \term{consistent} \term{augmented theories} and so satisfy $\bot \notin \Phi^+$. 
	The cases of $\wedge$ and $\vee$ are straightforward.
	
	The case where $\varphi$ is of the form $\psi \imp \chi$ is  as follows.
	From left to right, let us assume towards a contradiction that $ ( \rcm {\mathcal M^{\csf}}\Sigma, \Phi) \not \models \varphi$.
		Therefore, there exists $\Psi\in W_\Sigma$ such that $\Phi\peq_\Sigma \Psi$ and $( \rcm {\mathcal M^{\csf}}\Sigma, \Psi) \models \psi$ but $( \rcm {\mathcal M^{\csf}}\Sigma, \Psi) \not \models \chi$.
		By induction, $\psi \in\Psi^+$ and $\chi \not \in \Psi^+$. However, since $\Phi\peq_\Sigma\Psi$, $\psi\to \chi\in \Psi^+$. By~\ref{rl:mp}, $\chi \in \Psi^+$, which is a contradiction.
		Conversely, let us assume that $\psi \to \chi \not \in \Phi^+$ and let us define $\Omega=(\Phi^+ \cup \lbrace \varphi\rbrace,\lbrace \chi\rbrace)$. 
		$\Omega$ must be consistent, otherwise, there should be $\gamma \in \Phi^+$ such that $\csf \vdash \gamma \wedge \psi \to \chi$. 
		By some reasoning in intuitionistic logic we get that $\csf \vdash \gamma \to \left( \psi \to \chi\right)$.
		From $\gamma \in \Phi^+$ and~\ref{rl:mp}, $\psi \to \chi \in \Phi^+$ and we would reach a contradiction. 
		By Lemma~\ref{lem:lindenbaum}, $\Omega$ can be extended to a consistent prime augmented theory $\Psi=(\Psi^+;\Psi^{\centernot\ps})$, such that $\Phi^+\subseteq \Psi^+$, so $\Phi \peq_\Sigma \Psi$, $\psi\in \Phi^+$ and $\chi \not \in \Psi^+$. 
		By induction, $( \rcm {\mathcal M^{\csf}}\Sigma, \Psi) \models \psi$ but $( \rcm {\mathcal M^{\csf}}\Sigma, \Psi) \not \models \chi$, so $( \rcm {\mathcal M^{\csf}}\Sigma, \Psi) \not \models \psi \to \chi$.

	The case where $\varphi$ is of the form $\ps \psi$ is as follows.
	Since $\ps\psi\in\Sigma$ by assumption, by \Cref{lemDiamCS4}\ref{item:a},  $\ps\psi \in \Phi^+$ if and only if whenever $\Phi \rcm\peq\Sigma \Psi$, there exists $\Psi \rcm\rel\Sigma \Omega$ such that $\psi \in \Omega^+$. By the inductive hypothesis, this happens if and only if whenever $\Phi \rcm\peq\Sigma \Psi$ there exists $\Psi \rcm\rel\Sigma \Omega$ such that $( \rcm {\mathcal M^{\csf}}\Sigma, \Omega ) \models \psi$, which is precisely the definition of $(\rcm {\mathcal M^{\csf}}\Sigma, \Phi) \models \ps \psi$.
	The case where $\varphi$ is of the form $\nec \psi$ is similar. 
By \Cref{lemDiamCS4}\ref{item:b},  $\nec \psi \in \Phi^+$ if and only if whenever $\Phi \rcm\peq\Sigma \Psi \rcm\rel\Sigma \Omega$, we have $\psi \in \Omega^+$. By the inductive hypothesis, this happens if and only if whenever $\Phi \rcm\peq\Sigma \Psi \rcm\rel\Sigma \Omega$, we have $(\rcm {\mathcal M^{\csf}}\Sigma, \Omega) \models \psi$, which is precisely the definition of $(\rcm {\mathcal M^{\csf}}\Sigma, \Phi) \models \nec \psi$.
\end{proof}

By varying our choice of $\Sigma$, we either obtain \term{strongly complete}[strong completeness] or the finite birelational frame property for \altterm{def:csf}{$\csf$} with respect to its class of models.

\begin{theorem}[Alechina, Mendler, De Paiva, and Ritter \cite{AlechinaMPR01}]\label{cs4_completeness}
Let $\models_\csf$ denote \term{local birelational semantic consequence} on the class of~\altterm{def:csff}{{\csff}s}, and $\vdash_\csf$ denote the \term{syntactic consequence} relation for the logic \altterm{def:csf}{$\csf$}. Then for any set of formulas $\Gamma \cup \{\varphi\}$,
\[\Gamma \models_\csf \varphi \iff \Gamma \vdash_\csf \varphi.\]
\end{theorem}

\begin{proof}
	Soundness is \Cref{cor:sound}. For completeness, suppose $\Gamma \not\vdash_\csf \varphi$. 
Set $\Sigma = \lanfull$ and apply \Cref{lem:lindenbaumGen} to find a $\varphi$\term{-consistent} \term{prime} \term{augmented theories over}[augmented \term{theory}] $\Gamma_* = (\Gamma^+_*; \varnothing)$ extending $(\Gamma,\varnothing)$, then apply \Cref{lem:truth-lemma} to see that $(\rcm{\mathcal M}\lanfull^{\csf},\Gamma_*) \models\Gamma$ but $(\rcm{\mathcal M}\lanfull^{\csf},\Gamma_*) \not \models\varphi$.
\end{proof}

\begin{theorem}[{\term{def:csf}[$\csf$]} finite frame property]\label{thmCS4fin}
	The logic \term{def:csf}[$\csf$] has the finite frame property. That is, if a formula $\varphi$ is \term{falsifiable} (respectively~\term{satisfiable}) on a \altterm{def:csff}{\csff}, then $\varphi$ is \term{falsifiable} (respectively \term{satisfiable}) on a finite \altterm{def:csff}{\csff} of size at most $2^{2
	|\varphi|+4}$.
\end{theorem}

\begin{proof}
Suppose that $\term{def:csf}[\csf]\not\vdash\varphi$ and let $\Sigma$ contain $\nec\bot$ and $\varphi$ and be closed under subformulas.
Apply \Cref{lem:lindenbaumGen} to find a $\varphi$\term{-consistent} \term{prime} \term{augmented theories over}[augmented \term{theory}] $\Gamma_* = (\Gamma^+_*; \varnothing)$ over $\Sigma$ extending $(\varnothing,\varnothing)$, then apply \Cref{lem:truth-lemma} to see that $(\rcm{\mathcal M}\Sigma^{\csf},\Gamma_*) \not \models\varphi$.
Moreover, there are at most $2^{|\varphi|+2}\cdot 2^{|\varphi|+2}$ pairs of sets of subsets of $\Sigma$, i.e.~$|\rcm W\Sigma| \leq 2^{2|\varphi|+4}$, as required.
\end{proof}

\begin{theorem}[{\term{def:csf}[$\csf$]} complexity]\label{thmCS4dec}
The logic \term{def:csf}[$\csf$] is decidable in \textsc{nexptime}.
\end{theorem}

\begin{proof}
Let $n$ be the length of a formula $\varphi$; thus $n+2$ bounds the number of subformulas of $\nec\bot$ and $\varphi$. A nondeterministic algorithm runs as follows on input $\varphi$.
\begin{itemize}
\item
Nondeterministically choose a structure $\mathcal M=(W,{\peq},{\rel},V(p_1),\dots, V(p_m))$ of cardinality at most $  2^{ 2 n +4}$, where $p_1, \dots, p_m$ are the variables appearing in $\varphi$. The structure $\mathcal M$ can be stored with an exponential (in $n$) number of bits, and hence choosing $\mathcal M$ can be done in exponential time.
\item
Check that $\mathcal M$ is a \altterm{def:csf}{$\csf$} frame and reject if not. This (deterministic) check takes exponential time.

\item
Calculate the interpretations in $\mathcal M$ of each subformula of  $\varphi$. There are at most $n$ interpretations to calculate, and each calculation can be done deterministically in exponential time provided the interpretations of subformulas of lower structural complexity are calculated first.

\item
Check whether there is any world not in the calculated interpretation of $\varphi$.\qedhere
\end{itemize}
\end{proof}

\begin{remark}
To the best of our knowledge, the decidability of \term{def:csf}[$\csf$] has not previously been explicitly stated in the literature, although it is likely to be extractable from known results, as conjectured by Alechina et al.~\cite{AlechinaMPR01}.
The \textsc{nexptime} upper bound is new, and the exponential upper bound on the size of frames is  optimal, as this is already a lower bound for the modality-free fragment, intuitionistic propositional logic~\cite{zakharyaschev1979complexity,Zakharyaschev2001}.
\end{remark}

\section{Strong completeness of \texorpdfstring{$\ps$-regular}{diamond-regular} logics}\label{secCompGS4}

In this section we show that the \altterm{regularlogic}{$\ps$-regular} logics we consider (namely, \altterm{def:isf}{$\isf$}, \altterm{def:sfi}{$\sfi$}, \altterm{def:gsf}{$\gsf$}, and \altterm{def:gsfc}{$\gsfc$}) are \term{strongly complete} with respect to their corresponding class of \term{birelational frame}[birelational frames] as given in \Cref{definition:frame_classes}.

Recall from Lemma~\ref{lemClassical}\eqref{classical:one} that  a point $w$ in a  $\ps$\term{-regular} model $\mathcal M$ satisfies $\ps\varphi$ if and only if there is $v\ler w$ satisfying $\varphi$.
Now, let $\Sigma $ be the closure under subformulas of $\Sigma_0:=\{ \nec\neg(p\wedge q), \Diamond (p\vee q) , \neg\neg(\Box p\vee\Box q)\}$ and suppose that $w'\seq w \rel v$.
By \term{forth--up confluent}[forth--up confluence], there is $v'$ such that $w'\rel v' \seq v$.
Let us assume that these are all the worlds of our frame, as in Figure~\ref{figPsRegEx}.

We can label worlds of  $\mathcal M$ by letting $\ell(x) $ be the set of formulas of $\Sigma$ true on $\mathcal M$, and suppose that $\ell(w) = \Sigma_0\cup\{\neg(p\wedge q) \}$.
We have  two options for the labels of the other worlds: in the first, we set
\[\ell(w') = \ell(v) = \ell(v') := \ell(w )\cup \{\Box p,p\},\]
and in the second,
\[\ell(w') = \ell(v) = \ell(v'):= \ell(w )\cup \{\Box q,q\}.\]
If for $r\in \{p,q\}$ we put $x\in \val r$ if and only if $r\in \ell(x)$, both options lead to models where $w$ satisfies exactly those formulas of $\Sigma$ belonging to $\ell(w)$.
However, it is not possible to combine the two, as letting $\ell(w')  := \ell(w )\cup \{\Box p,p\}$ and $\ell(v)  := \ell(w )\cup \{\Box q,q\}$ forces us to have both $p$ and $q$ in $\ell(v')$, making $\nec\neg(p\wedge q)$ false on $w$.
This tells us that ensuring \term{forth--up confluent}[forth--up confluence] forces us to look at formulas outside $\Sigma$, since $\ell(w)$ alone does not tell us which of the two options we should choose.

\begin{figure}[h!]\centering
	\begin{tikzpicture}[node distance=1.5cm,auto]
		\begin{scope}[node distance=1.5cm,text width=2mm, local bounding box=scope2 ]
			\node[state] (x1)  {$w'$};
			\node[state,below of=x1] (y1)  {$w$};
			\node[state,right of=x1] (x2)  {$v'$};
			\node[state,below of=x2] (y2)  {$v$};
			
			\path[->,dashed] (y1) edge[] node[] {} (x1);
			\path[->,dashed] (y2) edge[] node[] {} (x2);
			\path[->] (x1) edge[] node[] {} (x2);
			\path[->] (y1) edge[] node[] {} (y2);
		\end{scope}	 			
	\end{tikzpicture}	
	\caption{A simple $\ps$-regular frame. Reflexive arrows are not displayed.}\label{figPsRegEx}
\end{figure}

Thus, unlike in the completeness proof for \term{def:csf}[$\csf$],  we can no longer relativise the canonical models to obtain the finite model property directly, so this relativisation aspect is omitted.
On the other hand, the simplified semantics for $\ps$ allows for the canonical models for \altterm{regularlogic}{$\ps$-regular} logics to be a bit more straightforward than that for \term{def:csf}[$\csf$] and instead they are very similar to those of intuitionistic (rather than constructive) modal logics~\cite{Simpson94}.
Indeed, this subsection may be read as a review of the completeness proof for \altterm{def:isf}{$\isf$} up until \Cref{thm:completeness:is4s4i}, with only a few variations needed to obtain \term{strongly complete}[strong completeness] for \altterm{def:sfi}{$\sfi$} as well.
We remark that we will not use the completeness of \altterm{def:isf}{$\isf$}, as our techniques do not yield the finite model property for this logic (but see~\cite{DBLP:conf/lics/GirlandoKMMS23}), but the proof is a by-product of that for the closely related \altterm{def:gsf}{$\gsf$} and \altterm{def:gsfc}{$\gsfc$}, whose completeness is stated in~\cite{Caicedo2010StandardGM,RodriguezV21} for \altterm{gk_model}{G\"odel--Kripke models} and whose finite \term{birelational frame} property we do establish.
We begin by uniformly defining the canonical model for any \altterm{regularlogic}{$\ps$-regular} logic $\Lambda$.

\begin{definition}
Let $\Lambda$ be a \altterm{regularlogic}{$\ps$-regular} logic.
We define the canonical model for $\Lambda$ as $\mathcal M_{\mathrm c} ^\Lambda=(W_{\mathrm c} ,{\peq_{\mathrm c} },{\rel_{\mathrm c} }, V_{\mathrm c} )$, where
\begin{enumerate}[label=\alph*)]
	\item $W_{\mathrm c} $ is the set of \term{consistent} \term{prime} $\Lambda$-\term{theory}[theories], 
\item 	${\peq_{\mathrm c}} \subseteq W_{\mathrm c} \times W_{\mathrm c} $ is defined by $\Phi \peq_{\mathrm c}  \Psi$ if and only if $\Phi \subseteq \Psi $,
\item ${\rel_{\mathrm c}}  \subseteq W_{\mathrm c} \times W_{\mathrm c} $ is defined by $\Phi \rel_{\mathrm c}  \Psi$ if and only if $  \bfrm\Phi \subseteq  \Psi $ and $  \dfrm\Phi \cap \Psi = \varnothing $, 

\item $V_{\mathrm c} \from \mathbb P \to 2^{W_{\mathrm c}}$ is defined by $V_{\mathrm c}(p):= \{\Phi\in W_{\mathrm c} \mid p\in \Phi\}$.
\end{enumerate}
\end{definition}

As we will now verify, this construction always yields \altterm{regular}{$\ps$-regular} models, given the assumption that $\Lambda$ is \altterm{regularlogic}{$\ps$-regular}.

\begin{lemma}
If $\Lambda$ is any \altterm{regularlogic}{$\ps$-regular} logic, then $\mathcal M^\Lambda_{\mathrm c}$ is a \altterm{regular}{$\ps$-regular} \term{bi-preordered model}.
\end{lemma}

\begin{proof}
It is immediate from the definition that $\peq_{\mathrm c}$ is a preorder. It is also straightforward to show, using Axioms~\ref{ax:trans:box} and \ref{ax:trans:dia}, that $\rel_{\mathrm c}$ is a preorder. Thus $\mathcal M^\Lambda_{\mathrm c}$ is a \term{bi-preordered model}. We have defined $\mathcal M^\Lambda_{\mathrm c}$  to be \term{infallible}, so to prove that $\mathcal M^\Lambda_{\mathrm c}$ is \altterm{regular}{$\ps$-regular}, it only remains to show that $\mathcal M^\Lambda_{\mathrm c}$ is \term{forth--up confluent}.

Let $\Phi$, $\Psi$, and $\Theta$ in $W_{\mathrm c} $ be such that $\Phi \peq_{\mathrm c}  \Psi$ and $\Phi \rel_{\mathrm c}  \Theta$. 
We claim that $   \bfrm\Psi\cup\Theta $ is $\ps \Psi^{\centernot\ps}$\term{-consistent}.
If not, there exist $\varphi \in \bfrm\Psi$, $\theta \in \Theta $, and $  \psi  \in \Psi^{\centernot\ps}$ such that
$\varphi \wedge \theta \rightarrow \ps\psi \in \Lambda$. (Note that we can take single formulas since $\Theta $ is closed under conjunction; by \Cref{auxiliary}\ref{der:box_dist_and}, $ \bfrm\Psi$ is too; and in view of \ref{ax:dp} and \term{prime}[primality] of $\Psi$, we know $\Psi^{\centernot\ps}$ is closed under disjunction.)
Since $\theta \in \Theta $, we then obtain $\varphi \rightarrow \ps\psi \in \Theta $. 
Then by Axiom~\ref{ax:ref:dia}, we get $\ps (\varphi \rightarrow \ps\psi) \in \Theta $.
Therefore $\varphi \rightarrow \ps\psi \not \in  \Theta^{\centernot\ps}$, so $\varphi \rightarrow \ps\psi \not \in  \Phi^{\centernot\ps}$, and hence $\ps (\varphi \rightarrow \ps\psi) \in  \Phi $.
Using \Cref{auxiliary}\ref{der:ax1} we get $\nec \varphi \rightarrow \ps\ps\psi \in \Phi $. 
Since $\Phi \subseteq \Psi $, this gives $\nec \varphi \rightarrow \ps\ps\psi \in \Psi $. As $\varphi \in \bfrm\Psi$---i.e.~$\nec \varphi \in \Psi$---we deduce $\ps\ps \psi \in \Psi $.  By Axiom~\ref{ax:trans:dia}, we then have $\ps \psi \in \Psi$, but contradicts the hypothesis that   $  \psi  \in \Psi^{\centernot\ps}$. 

In view of the Lindenbaum \Cref{lem:lindenbaum}, the set $  \bfrm\Psi\cup\Theta$ can be extended to a $\ps\Psi^{\centernot\ps}$\term{-consistent} \term{prime} \term{theory} $\Upsilon$. 
%
%
%
It is easy to check by our choice of $\Upsilon$ that $\Psi \rel_{\mathrm c}  \Upsilon$ and $\Theta \peq_{\mathrm c}  \Upsilon$.
%
%
%
%
%
\end{proof}

Next we establish that the canonical models have the appropriate additional properties. To help the reader track all the cases, we summarise the extra properties the frames must have, beyond being \altterm{regular}{$\ps$-regular} \term{bi-preorder}[bi-preorders] and the extra axioms the logics contain, beyond \altterm{def:csf}{$\csf$} + \ref{ax:dp} + \ref{ax:null}.

\begin{tabular}{@{}r@{ }l@{\quad}ll@{}}
    $\bullet$\ & \altterm{def:csf}{$\isf$}:  & \ref{ax:fs}              & \term{back--up confluent} \\
    $\bullet$\ & \altterm{def:csf}{$\sfi$}:  & \ref{ax:cd}              & \term{forth--down confluent} \\
    $\bullet$\ & \altterm{def:csf}{$\gsf$}:  & \ref{ax:fs}  + \ref{ax:g}               & \term{back--up confluent} + \term{upward linear} \\
    $\bullet$\ & \altterm{def:csf}{$\gsfc$}: & \ref{ax:fs}  + \ref{ax:cd} + \ref{ax:g} & \term{back--up confluent} + \term{forth--down confluent} + \term{upward linear} \\
\end{tabular}

From the summary, it is clear that the following lemma is sufficient. 


\begin{lemma}
Let $\Lambda$ be any \altterm{regularlogic}{$\ps$-regular} logic.
\begin{enumerate}

\item\label{item:1} If $\ref{ax:fs} \in \Lambda$, then $\mathcal M^\Lambda_{\mathrm c}$ is \term{back--up confluent}.

\item\label{item:2}  If $\ref{ax:cd} \in \Lambda$, then $\mathcal M^\Lambda_{\mathrm c}$ is \term{forth--down confluent}.

\item\label{item:3}  If $\ref{ax:g} \in \Lambda$, then $\mathcal M^\Lambda_{\mathrm c}$ is \term{upward linear}.

\end{enumerate}
\end{lemma}

\begin{proof} 

\eqref{item:1} Suppose that $\Phi\rel_{\mathrm c} \Psi\peq_{\mathrm c}  \Theta$, and let $\Xi \coloneqq  \{\nec\xi \in\lanfull \mid \xi\notin \Theta\}$ and $\Delta = \{\ps \delta \in\lanfull \mid \ps\delta \in \Theta\}$.
We claim that $ \Phi\cup\Delta $ is $\Xi$\term{-consistent}.
If not, there are $\varphi \in\Phi $, $\ps\delta_1,\ldots,\ps\delta_n \in \Delta$, and $\nec\xi\in \Xi$ such that $ \varphi ,\{\ps \delta_i\}_{i=1}^n \vdash \nec \xi $, which using \ref{ax:trans:dia} yields $ \varphi \vdash \bigwedge_{i=1}^n \ps \ps \delta_i \imp \nec \xi $. Then using \Cref{auxiliary}\ref{der:ax3} we can obtain $ \varphi \vdash \ps\bigwedge_{i=1}^n  \ps \delta_i \imp \nec \xi $.
By an application of \ref{ax:fs}, we obtain $\varphi \vdash \nec (\bigwedge_{i=1}^n \ps \delta_i \imp  \xi  )$.
Then since $\varphi\in \Phi $, we obtain $\nec (\bigwedge_{i=1}^n \ps \delta_i \imp  \xi  ) \in\Phi$.
Since $\Phi\rel_{\mathrm c} \Psi$, we get $\bigwedge_{i=1}^n \ps \delta_i \imp  \xi \in \Psi$, and since $\Psi\peq_{\mathrm c} \Theta$ and each $\ps\delta_i\in \Theta$, we get 
 $\xi\in \Theta $, contradicted by the definition of $\Xi$. Hence $\Phi \cup\Delta $ is $\Xi$\term{-consistent}.
 
Let $\Upsilon$ be any $\Xi$\term{-consistent} \term{prime} extension of $\Phi \cup\Delta $. Then $\Phi\peq_{\mathrm c} \Upsilon\rel_{\mathrm c} \Theta$: that $\Phi\peq_{\mathrm c} \Upsilon$ follows from $\Phi \subseteq \Upsilon $, and we check that $\Upsilon\rel_{\mathrm c} \Theta$.
We have $\Upsilon^\nec\subseteq \Theta $, since if $\xi\notin \Theta $ it follows by the definition of $\Xi$ and the fact that $\Upsilon$ is $\Xi$\term{-consistent} that $\xi\notin\Upsilon^\nec$, while if $\delta\notin\Theta^{\centernot\ps}$ it follows that $\ps\delta\in \Theta $; hence $\ps\delta\in \Delta \subseteq \Upsilon $ and $\delta\notin\Upsilon^{\centernot\ps}$.
It follows that $\Upsilon^{\centernot\ps}\subseteq \Theta^{\centernot\ps}$, so $\Upsilon\rel_{\mathrm c} \Theta$.

\eqref{item:2} Suppose that $\Phi \peq_{\mathrm c}  \Psi\rel_{\mathrm c}  \Theta$.
Let $\Xi = \lanfull\setminus \Theta $.
We claim that $ \bfrm\Phi $ is $\ps \Phi^{\centernot\ps}\cup \Xi$\term{-consistent}.
If not, there exist $\nec \varphi \in \Phi $, $\ps \psi \not \in \Phi $, and $\chi \not \in \Theta $ such that $ \varphi \rightarrow   \chi \vee \ps\psi \in \Lambda $.
By~\ref{rl:nec}, Axiom~\ref{ax:k:box} and~\ref{rl:mp} we get that $ \nec \varphi \imp \nec (\chi \vee \ps\psi) \in \Lambda $, and thus $ \nec (\chi \vee \ps\psi) \in \Lambda $.
By~\ref{ax:cd} it follows that $\nec \chi \vee \ps\ps \psi \in \Phi $, and then by \ref{ax:trans:dia}, it follows that $\nec \chi \vee \ps \psi \in \Phi $.
Since $\chi \not \in \Theta $, we know $\nec \chi \not \in \Phi $.
Then by primality, $\ps \psi \in \Phi $\allowbreak---a contradiction. 
Thanks to \Cref{lem:lindenbaum}, $ \bfrm\Phi $ can be extended to a $\ps \Phi^{\centernot\ps}\cup \Xi$\term{-consistent} \term{prime} \term{theory} $\Upsilon$.
It then readily follows that $\Phi\rel_{\mathrm c} \Upsilon\peq_{\mathrm c} \Theta$, as required.

\eqref{item:3} Assume toward a contradiction that $\peq_{\mathrm c} $ is not \term{upward linear}. 
	Let the \term{theory}[theories] $\Phi$, $\Psi$, and $\Omega$ be such that $\Phi \peq_{\mathrm c}  \Psi$ and $\Phi \peq_{\mathrm c}  \Omega$, but $\Psi \not \peq_{\mathrm c}  \Omega$ and $\Omega \not \peq_{\mathrm c}  \Psi$. 
	From the definition of $\peq_{\mathrm c} $ we get that $\Psi  \not \subseteq \Omega $ and $\Omega  \not \subseteq \Psi $.
	So there exist two formulas $\varphi$ and $\psi$ such that $\varphi \in \Psi  \setminus \Omega $ and $\psi \in \Omega  \setminus \Psi $.
	Therefore $\varphi \imp \psi \not \in \Psi  \supseteq \Phi $ and $\psi \imp \varphi \not \in \Omega  \supseteq \Phi $. 
	Consequently, $(\varphi \imp \psi) \vee (\psi \imp \varphi) \not \in \Phi $, contradicting Axiom~\ref{ax:g}.
\end{proof}


It already follows from the above that when $\Lambda$ is any of the logics \altterm{def:isf}{$\isf$}, \altterm{def:sfi}{$\sfi$}, \altterm{def:gsf}{$\gsf$}, or \altterm{def:gsfc}{$\gsfc$}, the structure $\mathcal M^\Lambda_{\mathrm c}$ is a $\Lambda$-model.  Thus to conclude our completeness proof, it remains to establish a standard truth lemma.
The case of formulas of the form $\nec\varphi$ is subtle, and we establish it separately depending on which axioms/confluence conditions we are operating with.
In particular, the argument for \term{back--up confluent} logics such as \altterm{def:isf}{$\isf$} is trickier than that for \term{forth--down confluent} logics such as \altterm{def:sfi}{$\sfi$}.

\begin{lemma}\label{lemBoxWit}
Let $\Lambda$ be a \altterm{regularlogic}{$\ps$-regular} logic containing either \ref{ax:fs} or \ref{ax:cd}.
Then for any $\Phi\in W_{\mathrm c}$ and any formula $\varphi$, we have $\nec\varphi\in \Phi$ if and only if whenever $\Phi\peq_{\mathrm c}\Psi\rel_{\mathrm c}\Theta$, it follows that $\varphi\in \Theta$.
\end{lemma}

\begin{proof}
For the left-to-right direction, 
 suppose that $\nec\varphi\in \Phi$ and $\Phi\peq_{\mathrm c} \Psi\rel_{\mathrm c}\Theta$. 
By the definition of $\peq_{\mathrm c}$, we have $\nec \varphi \in \Psi$, and thus $\varphi \in \bfrm\Psi$. Then by the definition of $\rel_{\mathrm c}$, we have 
{$\varphi\in \Theta$, as needed}.

For the other direction, we consider two cases.
First assume that $ \ref{ax:fs} \in \Lambda$.
If $\nec \varphi \not \in \Phi $, using \Cref{lem:lindenbaum}, let $\Psi\seq_{\mathrm c}  \Phi$ be a maximal $\nec\varphi$\term{-consistent} \term{prime} \term{theory}; note that maximality and $\peq_{\mathrm c}$ are both defined by set inclusion and thus $\Psi$ is $\peq_{\mathrm c}$-maximal.
We claim that
\begin{equation}\label{eqDiamImp}
\chi\in \Psi^{\centernot\ps} \implies \nec(\chi\imp \varphi) \in \Psi .
\end{equation}
Suppose $\chi\in \Psi^{\centernot\ps}$. By maximality of $\Psi$, we have $\Psi,\ps\chi \vdash \nec \varphi $, so $\Psi\vdash \ps\chi\imp\nec\varphi$.
By \ref{ax:fs}, we obtain $\Psi\vdash \nec(\chi\imp\varphi)$ and thus $ \nec(\chi\imp\varphi) \in \Psi$, as needed.

Note by \ref{ax:dp} that $\Psi^{\centernot\ps}$ is closed under disjunction, given that $\Psi$ is \term{prime}. We claim that $\bfrm\Psi$ is $\ps\Psi^{\centernot\ps}\cup\{\varphi\}$\term{-consistent}.
 If not, we would have $\chi\in \bfrm\Psi$ and $\theta \in \Psi^{\centernot\ps}$ such that $ \chi \vdash  \varphi\vee \ps\theta$, and reasoning as before, $ \nec\chi  \vdash \nec( \varphi\vee \ps\theta)$.
It follows that  $ \nec\chi ,\nec(\ps\theta\imp \varphi) \vdash \nec( \varphi\vee \varphi)$, that is~$ \nec\chi ,\nec(\ps\theta\imp \varphi) \vdash \nec \varphi $.
From \ref{ax:trans:dia} we see that $\ps\ps\theta \in \Psi $ implies that $ \ps\theta\in \Psi $, which by contraposition becomes $\theta\in \Psi^{\centernot\ps} $ implies $\ps\theta\in\Psi^{\centernot\ps}$.
Thus we may use \eqref{eqDiamImp} to conclude that $\nec(\ps\theta\imp\varphi) \in \Psi $, so $\Psi  \vdash \nec\varphi$, a contradiction.

Hence there exists a  $\ps \Psi^{\centernot\ps}\cup\{\varphi\}$\term{-consistent} \term{prime} \term{theory} $\Upsilon$ extending $ \bfrm\Psi $.
It follows that $\varphi\notin \Upsilon$, and it is readily verified that $\Phi\peq_{\mathrm c} \Psi\rel_{\mathrm c}  \Upsilon $, as needed.

Next we consider the case where $ \ref{ax:cd} \in \Lambda$.
Suppose that $\nec \varphi \not \in \Phi $;  we claim that $ \bfrm\Phi $ is $\ps \Phi^{\centernot\ps} \cup\{ \varphi\}$\term{-consistent}. 
		If not, there exist $\nec \chi \in \Phi$ and $\ps \theta \not \in \Phi$ such that $ \chi \rightarrow \varphi \vee \ps \theta \in \Lambda$.
		By~\ref{rl:nec},~\ref{ax:k:box} and~\ref{rl:mp} we get $\nec ( \varphi \vee \ps \theta) \in \Phi$. By Axiom~\ref{ax:cd}, this gives $\nec \varphi \vee \ps\ps \theta \in \Phi$---a contradiction since $\Phi$ is \term{prime}, and neither $\nec\varphi\in \Phi$ nor, by \ref{ax:trans:dia}, $\ps\psi\in \Phi$.
		Therefore, $  \bfrm\Phi  $ can be extended to a $\ps \Phi^{\centernot\ps}\cup\varphi$\term{-consistent} \term{prime} \term{theory} $\Upsilon $.
		Clearly, 
		$\varphi \not \in \Upsilon$ and $\Phi\peq_{\mathrm c}\Phi\rel_{\mathrm c} \Upsilon$.		
\end{proof}

With this, we may establish the truth lemma for \altterm{regularlogic}{$\ps$-regular} logics.

\begin{lemma}[truth lemma]\label{lem:truth-lemma:regular}
Let $\Lambda$ be any \altterm{regularlogic}{$\ps$-regular} logic such that $\Lambda$ contains either \ref{ax:fs} or \ref{ax:cd}.
	For any  $\Phi \in W_{\mathrm c} $ and $\varphi \in \lanfull$, 
\[\varphi \in \Phi  \iff ( \mathcal M_{\mathrm c} ^{{\Lambda}}, \Phi ) \models \varphi.\]
\end{lemma}

\begin{proof}
	By structural induction on $\varphi$.	
	The case of propositional variables holds by the definition of $V_{\mathrm c} $.
	The cases of $\wedge$ and $\vee$ are straightforward, and the case for $\varphi=\nec\psi$ follows from \Cref{lemBoxWit}.

We consider the $\imp$ connective next.
First suppose $\psi \imp \chi \in \Phi $. Let $\Psi \seq_{\mathrm c}  \Phi$ be such that $(\mathcal M_{\mathrm c} ^{{\Lambda}}, \Psi )\models \psi$.
By the induction hypothesis, $\psi\in \Psi $.
		Since $\Psi \seq_{\mathrm c}  \Phi$ and $\psi \imp \chi \in \Phi $, we have $\psi \imp \chi \in \Psi $.
		By closure under \term{syntactic consequence}, $\chi \in \Psi $, and then by the induction hypothesis, $(\mathcal M_{\mathrm c} ^{{\Lambda}}, \Psi ) \models  \chi$.
Since $\Psi$ was arbitrary subject to $\Psi \seq_{\mathrm c}  \Phi$, we conclude $(\mathcal M_{\mathrm c} ^{{\Lambda}}, \Phi ) \models \psi\imp \chi$.

Conversely, suppose $\psi \imp \chi \not \in \Phi $.
		We claim that $\Phi\cup\{\psi\}$ is $\chi$\term{-consistent}. 
		If not, by the definition of $\vdash$, there exists $\xi \in \Phi $ such that $ \xi \wedge  \psi \imp \chi \in \Lambda $. 
		It follows that  $ \xi \imp\left(  \psi \imp \chi\right)  \in \Lambda$.
		Since $\xi \in \Phi $, we obtain $\psi \imp \chi \in \Phi $---a contradiction.
		Hence $\Phi\cup\{\psi\}$ is $\chi$\term{-consistent} and so by \Cref{lem:lindenbaum} can be extended to a $\chi$\term{-consistent} \term{prime} \term{theory} $\Psi $, and hence $\Phi \peq_{\mathrm c}  \Psi$.
		By the induction hypotheses for $\psi$ and $\chi$, we deduce $(\mathcal M_{\mathrm c} ^{{\Lambda}}, \Psi) \models \psi$ and  $(\mathcal M_{\mathrm c} ^{{\Lambda}}, \Psi) \not \models \chi$.
		Therefore $(\mathcal M_{\mathrm c} ^{{\Lambda}}, \Phi) \not \models \psi \imp \chi$.
	
For the case that $\varphi$ is of the form $\ps \psi$, if $\ps \psi \in \Phi $, then let us define $\Theta =    \bfrm\Phi \cup \{\psi\}$, and let us assume toward a contradiction that $\Theta \vdash \lbrace \chi \in \lanfull \mid \ps \chi \not \in \Phi \rbrace$. This means that there exist $\nec \theta \in \Phi $ and $\ps \chi \not \in \Phi $ such that $\theta \rightarrow\left( \psi \rightarrow \chi\right) \in \Lambda$. From~\ref{rl:nec} and \ref{ax:k:box}, we obtain $\nec \theta \imp \nec(\psi \imp \chi)$, and thus $\nec(\psi \imp \chi) \in \Phi $. Then by~\ref{ax:k:dia}, we obtain $\ps\psi \imp \ps \chi$, and thus $\ps \chi \in \Phi $---a contradiction. Therefore $\Theta $ can be extended to a \term{prime} \term{theory} $\Upsilon $ such that $\Upsilon  \not \vdash \lbrace \chi \mid \ps \chi \not \in \Phi  \rbrace$.
		It follows that $\Phi\rel_{\mathrm c} \Upsilon$. Moreover, $\psi \in \Upsilon $. By the induction hypothesis, $ ( \mathcal M_{\mathrm c} ^{{\Lambda}}, \Upsilon )   \models \psi$, so $( \mathcal M_{\mathrm c} ^{{\Lambda}}, \Phi )  \models \ps\psi$.

Conversely, assume that $ ( \mathcal M_{\mathrm c} ^{{\Lambda}}, \Phi ) \models \ps\psi$, so $( \mathcal M_{\mathrm c} ^{{\Lambda}}, \Upsilon) \models \psi$ for some $\Phi \rel_{\mathrm c}  \Upsilon$. By the inductive hypothesis, $\psi \in \Upsilon $; hence by \ref{ax:ref:dia}, we have $\ps\psi \in \Upsilon $, and thus $\psi\not\in \Upsilon^{\centernot\ps}$.
It follows that $ \psi\not\in \Phi^{\centernot\ps}$, hence $\ps\psi\in \Phi$, as required.
\end{proof}

It follows from the above considerations that the intuitionistic modal logics \altterm{def:isf}{$\isf$}, \altterm{def:sfi}{$\sfi$}, \altterm{def:gsf}{$\gsf$}, and \altterm{def:gsfc}{$\gsfc$} are \term{strongly complete} with respect to their corresponding class of frames, and in particular any formula that is not derivable is \term{falsifiable}. Completeness of \altterm{def:isf}{$\isf$} is well known and was first proven in \cite{servi1984axiomatizations} and strong completeness follows from such a proof.


\begin{theorem}[strong completeness of $\isf$, $\sfi$, $\gsf$, and $\gsfc$] \label{thm:completeness:is4s4i}
	Let $\Lambda$ be one of the logics $\isf$, $\sfi$, $\gsf$, or $\gsfc$. 
	Let $\models_\Lambda$ denote \term{local birelational semantic consequence} on the class of $\Lambda$ frames as given by \Cref{definition:frame_classes}, and let $\vdash_{\Lambda}$ denote the \term{syntactic consequence} relation for $\Lambda$. Then for any set of formulas $\Gamma \cup \{\varphi\}$,
	\[\Gamma \models_\Lambda \varphi \iff \Gamma \vdash_\Lambda \varphi.\]
\end{theorem}





\begin{remark}
\Cref{thm:completeness:is4s4i} also applies to the logics ${\sfi } +{\ref{ax:g}}$ and ${\isf} +{\ref{ax:cd}}$, but the first logic does not seem to be the logic of any natural class of \altterm{gk_model}{G\"odel--Kripke models}, and our techniques in \Cref{sec:fmp} do not settle whether the second has the finite model property, essentially for the same reason that they do not work for \altterm{def:isf}{$\isf$}.
\end{remark}

\section{Finite frame properties for $\ps$-regular logics}\label{sec:fmp}

In this section we prove that the finite \term{birelational frame} property holds for each of the logics \altterm{def:gsf}{$\gsf$}, \altterm{def:gsfc}{$\gsfc$},  and \altterm{def:sfi}{$\sfi$}, in that order.
As we are specifically interested in $\ps$-\term{regular} logics, throughout this section the models we consider will be \term{infallible}, and we omit $\fallible W$. 
All the proofs make use of a quotient modulo bisimilarity. For \altterm{def:sfi}{$\sfi$} it is first necessary to establish that the logic has the \term{shallow} model property, before applying the quotient to a shallow model.
We begin by discussing the notions of bisimulation quotients that we use.

\subsection{\texorpdfstring{$\Sigma$}{Sigma}-bisimulations}\label{secSBisim}

In this subsection we develop the theory of \emph{$\Sigma$\altterm{-bisimulation}{-bisimulations}}, a key component common to all of our finite frame property proofs for $\ps$-\term{regular} logics.
These are defined as follows.

\begin{definition}\label{defLabel}
	Given a set of formulas $\Sigma$ and a \term{bi-preordered model} $\mathcal M =(W,{\peq},{\rel},V) $, we define the $\Sigma$\define{-label} of $w\in W$ by 
\[\ell (w)  \coloneqq  \{ \varphi\in \Sigma \mid (\mathcal M,w)\models \varphi \}.\]
	A $\Sigma$\define{-bisimulation} on $\mathcal M$ is a \altterm{forth--up confluent}{forth--up} and \term{back--up confluent} relation $Z \subseteq W  \times W $ such that if $w\mathrel Z v$ then $\ell(w) = \ell(v)$. 
\end{definition}

In particular, for each canonical model $\mathcal M^\Lambda_{\mathrm c} $ and $\Phi\in W_{\mathrm c} $, note that $\ell(\Phi)= \Phi\cap \Sigma$.

\begin{definition}
	Given a \term{bi-preordered model} $\mathcal M=(W,{\peq},{\rel},V)$ and an equivalence relation ${\sim}\subseteq W\times W$, we denote the equivalence class of $w\in W$ under $\sim$ by $[w]$.
	We then define the quotient $\nicefrac{\mathcal M}\sim = (\nicefrac W\sim,\nicefrac\peq \sim,\nicefrac\rel \sim,\nicefrac V \sim)$ by
\begin{itemize}
\item
$\nicefrac W \sim \coloneqq  \{[w] \mid w\in W\}$, 
\item
$[w]\mathrel{\nicefrac \peq\sim} [v]$ if there exist $w'\sim w$ and $v'\sim v$ such that $w' \peq v' $, 
\item
$\nicefrac \rel \sim$ is the  {transitive} closure of $\nicefrac {\rel^0} \sim$ defined by $[w] \mathrel{\nicefrac {\rel^0} \sim} [v]$ whenever there are $w',v' $ such that $w\sim w' \rel v' \sim v$,
\item
 $\nicefrac V\sim (p) \coloneqq \{[w] \mid w \in V(p)\}$
\end{itemize}
\end{definition}

Of particular importance is the case where the relation $\sim$ is given by $\Sigma$\term{-bisimulation}.

\begin{lemma}\label{lemQuotLeq}
	Let $\mathcal M=(W,{\peq},{\rel},V)$ be a \term{bi-preordered model}, and suppose that $\sim$ is an equivalence relation that is also a $\Sigma$\term{-bisimulation}.
	Then, 
	\begin{itemize}
	\item
	for all $w,v\in W$, we have $[w]\mathrel{ \nicefrac \peq\sim }[v]$ if and only if there is $v'\sim v$ such that $ w \peq v' $;
	\item
	the relation  $\nicefrac \peq\sim$ is a preorder.
	\end{itemize}
\end{lemma}

\begin{proof}
	For the first item, clearly if there is $v'\sim v$ such that $ w \peq v' $, then $[w]\mathrel{ \nicefrac \peq\sim }[v]$. Conversely, if $[w]\mathrel{ \nicefrac \peq\sim }[v]$, then there are $w'\sim w$ and $v'\sim v$ such that $w'\peq v'$.
	Since $\sim$ is a bisimulation, it satisfies \altterm{forth--up confluent}{forth--up confluence}, so there is $v''\sim v'$ such that $w\peq v''$.
	But then by transitivity of $\sim$, we have $v'' \sim v$, as needed.
	
	For the second item, clearly $ \nicefrac \peq\sim$ is {reflexive}. For transitivity, suppose $[w]\mathrel{ \nicefrac \peq\sim }[v] \mathrel{ \nicefrac \peq\sim }[u]$. Then by the first item, there exists $v' \sim v$ with $w  \peq v'$. Then $[v']\mathrel{ \nicefrac \peq\sim }[u]$, and so, using the first item again, there exists $u' \sim u$ with $v'  \peq u'$. By transitivity of $\peq$, we have $w \peq u'$, and hence $[w] \mathrel{ \nicefrac \peq\sim } [u']=[u]$.
\end{proof}

By the second item of \Cref{lemQuotLeq}, if $\mathcal M$ is a \term{bi-preordered model} and $\sim$ is an equivalence relation that is also a $\Sigma$\term{-bisimulation}, then $\nicefrac{\mathcal M}\sim$ is also a \term{bi-preordered model}.

We now need an auxiliary lemma confirming that confluence notions for relations have nice closure properties.
If $R,S$ are binary relations then let $R\mathbin{;}S$ denote their composition: $x \mathrel{R\mathbin{;}S} y$ if there is $z$ such that $x\mathrel R z$ and $z\mathrel S y$.

\begin{lemma}\label{lemmClosure}
Let $\mathcal F =(W,{\peq})$ be an \term{intuitionistic frame}.
\begin{enumerate}

\item If $R,S\subseteq W\times W$ are \term{forth--up confluent}[forth--up] (respectively~\term{back--up confluent}[back--up], \term{forth--down confluent}[forth--down]) confluent, then so is $R\mathbin;S$.

\item If $(R_i)_{i\in I} \subseteq W\times W$ are \term{forth--up confluent}[forth--up] (respectively~\term{back--up confluent}[back--up], \term{forth--down confluent}[forth--down]) confluent, then so is $\bigcup_{i\in I} R_i$.

\item If $R$ is \term{forth--up confluent}[forth--up] (respectively~\term{back--up confluent}[back--up], \term{forth--down confluent}[forth--down]) confluent, then so is its {transitive} closure $R^+$.

\end{enumerate}
\end{lemma}

\begin{proof}
We prove only that \term{forth--up confluent}[forth--up confluence] is closed under composition and leave other items to the reader.
Suppose that $R,S\subseteq W\times W$ are \term{forth--up confluent} and $v\seq w \mathrel{R\mathbin ;S} w'$.
Then there is $w''$ such that $w\mathrel R w'' \mathrel S w'$.
By \term{forth--up confluent}[forth--up confluence] of $R$, there is $v''$ such that $v\mathrel R v''\seq w $.
By \term{forth--up confluent}[forth--up confluence] of $S$, there is $v'$ such that $v''\mathrel S v'\seq w' $.
But then $v\mathrel{R\mathbin{;}S} v'\seq w'$, as needed.
\end{proof}

Bisimulation quotients preserve all confluence properties we consider, as illustrated in Figure~\ref{figBisQuo} and made precise in Lemma~\ref{lemQuotPreserv} below.

 \begin{figure}[h!]\centering
 	\resizebox{.4\textwidth}{!}{
	\begin{tikzpicture}[node distance=1.5cm,auto, text width=2mm]
		\begin{scope}[node distance=1.5cm,local bounding box=scope2 ]
			\node[state,fill=turquoise] (x1)  {$p$};
			\node[state,fill=turquoise,below of=x1] (y1)  {$p$};
			\node[state,fill=lightgray,below of=y1] (z1)  {};
			
			\node[state,fill=tan, right of=x1] (x2)  {$q$};
			\node[state, fill=lightgray,below of=x2] (y2)  {};
			\node[state, fill=lightgray,below of=y2] (z2)  {};
			
			\path[->,dashed] (y1) edge[] node[] {} (x1);
			\path[->,dashed] (z1) edge[] node[] {} (y1);
			\path[->,dashed] (y2) edge[] node[] {} (x2);
			\path[->,dashed] (z2) edge[] node[] {} (y2);
			\path[->] (x1) edge[] node[] {} (x2);
			\path[->] (y1) edge[] node[] {} (y2);
			\path[->] (z1) edge[] node[] {} (z2);
		\end{scope}	 	
		
		\begin{scope}[node distance=1.5cm,shift={($(x2.east)+(4cm,0)$)},local bounding box=scope3 ]
			\node[state,fill=turquoise,right of=x2, node distance=4cm] (xx1)  {$p$};
			\node[state,fill=lightgray,right of =y2, node distance=4cm] (yy1)  {};
			\node[state,fill=tan, right of=xx1] (xx2)  {$q$};
			\node[state,fill=lightgray, below of=xx2] (yy2)  {};

			\path[->,dashed] (yy1) edge[] node[] {} (xx1);
			\path[->,dashed] (yy2) edge[] node[] {} (xx2);
			\path[->] (xx1) edge[] node[] {} (xx2);
			\path[->] (yy1) edge[] node[] {} (yy2);
			\path[->] (xx1) edge[] node[] {} (yy2);			
		\end{scope}	 	 	
		
		\node[right of=y2,node distance=2cm] (transition2)  {\scalebox{1.9}{$\Longrightarrow$}};
	\end{tikzpicture}}	
	\caption{A bisimulation quotient of a \altterm{def:gsfc}{$\gsfc$} model.
The intuitionistic relation $\peq$ is the transitive, reflexive closure of the relation defined by the dashed arrows, and $\rel$ is the reflexive closure of the relation defined by the solid arrows (which in this case is already transitive). 	
	Note that the quotient is also a  \altterm{def:gsfc}{$\gsfc$} model.}\label{figBisQuo}
\end{figure}

\begin{lemma}\label{lemQuotPreserv}
	Let $\mathcal M=(W,{\peq} ,{\rel} ,V )$ be a \term{bi-preordered model} and $\Sigma$ be a set of formulas.
	Let $\sim$ be a $\Sigma$\term{-bisimulation} that is also an equivalence relation.
	Then:
	\begin{enumerate}
		
		\item\label{itQuotPreservFU} If $\mathcal M$ is \term{forth--up confluent}, then so is $\nicefrac{\mathcal M}{\sim }$.
		
		\item If $\mathcal M$ is \term{back--up confluent}, then so is $\nicefrac{\mathcal M}{\sim }$.
		
		\item If $\mathcal M$ is \term{upward linear}, then so is $\nicefrac{\mathcal M}{\sim }$.
		
	\end{enumerate}
\end{lemma}

\begin{proof}
For \term{forth--up confluent}[forth--up confluence], in view of \Cref{lemmClosure}, it suffices to show that $\nicefrac {\rel^0} \sim$ is \term{forth--up confluent}.
Assume that $[w] \mathrel{\nicefrac {\rel^0} \sim} [u]$ and $[w] \mathrel{\nicefrac\peq\sim} [v]$. 
From $[w] \mathrel{\nicefrac {\rel^0} \sim} [u]$, there exist $w'$ and $u'$ such that $w\sim w'\rel u' \sim u$. 
Since $[w] \mathrel{\nicefrac\peq\sim} [v]$ and $w'\sim w$, we can also say that $[w']\mathrel{\nicefrac\peq\sim} [v]$. 
By \Cref{lemQuotLeq}, there exists $v''$ such that $w' \peq v''\sim v$. 
Since $\mathcal M$ is \term{forth--up confluent}, there exists $z$ such that $u'\peq z$ and $v'' \rel z$. From the definitions of $\mathrel{\nicefrac{\rel^0} \sim}$ and $\mathrel{\nicefrac{\peq}\sim}$ it follows that $[v]\mathrel{\nicefrac{\rel^0}\sim} [z]$ and $[u]\mathrel{\nicefrac{\peq} \sim} [z]$, as needed.
\term{back--up confluent}[Back--up confluence] is treated similarly, and we omit it.
	
For \altterm{upward linear}{upward linearity}, assume that $[w] \mathrel{\nicefrac{\peq}\sim} [u]$ and $[w] \mathrel{\nicefrac{\peq}\sim} [v]$.
Using \Cref{lemQuotLeq}, let $u'$, $v'$ in $W$ be such that $u \sim u' \seq w \peq v \sim v'$.
Hence $u' \peq v'$ or $v' \peq u'$. Consequently $[u]\mathrel{\nicefrac{\peq}\sim}[v]$ or $[v]\mathrel{\nicefrac{\peq}\sim}[u]$.
\end{proof}

\begin{lemma}[truth lemma for quotients]\label{lemQuotTruth'}
	If $\mathcal M = (W, {\peq},{\rel},V)$ is a \term{forth--up confluence}[forth--up confluent] \term{bi-preordered model}, $\Sigma\subseteq\lanfull$ is closed under subformulas, and ${\sim} \subseteq W\times W$ is a $\Sigma$\term{-bisimulation} that is also an equivalence relation, then $\ell(w) = \ell([w])$.
	Equivalently, for all $w \in W$ and $\varphi\in \Sigma$,  we have \[(\mathcal M,w)\models\varphi \iff (\nicefrac{\mathcal M}\sim,[w])\models\varphi.\]
\end{lemma}

\begin{proof}
	Proceed by a standard structural induction on $\varphi$.
	We consider only the interesting cases.
	\medskip

	\noindent \textsc{Case} $\varphi = \psi\imp \theta$. If $(\mathcal M,w)\models \psi\imp \theta$ and $[v] \mathrel {\nicefrac \seq\sim} [w]$, then in view of Lemma \ref{lemQuotLeq} (which we henceforth use without mention), there is $v' \sim v$ such that $w \peq v'$.
	Since $v'\seq w$, either $( \mathcal M,v')\not \models \psi$ or $( \mathcal M,v')   \models \theta$, which by the induction hypothesis and the fact that $[v] = [v']$ yields $(\nicefrac{\mathcal M}\sim,[v])\not \models \psi$ or $(\nicefrac{\mathcal M}\sim,[v]) \models \theta$. Since $[v]  \mathrel {\nicefrac \seq\sim} [w]$ was arbitrary, we conclude that $(\nicefrac{\mathcal M}\sim,[w])\models\varphi$.

	Conversely, if $(\nicefrac{\mathcal M}\sim,[w]) \models \psi \imp \theta$ and $v\seq w$, then $[w]\mathrel {\nicefrac \peq\sim} [v]$, which implies that $( \nicefrac{\mathcal M}\sim,[v]) \not \models \psi$ or $(\nicefrac{\mathcal M} \sim,[v]) \models \theta$, and then by the induction hypothesis, that $( \mathcal M,v)\not \models \psi$ or $( \mathcal M,v)   \models \theta$, as needed.
	\medskip
	
	\noindent \textsc{Case} $\varphi = \ps\psi$.
	Since $\mathcal M$ is assumed \term{forth--up confluent}, Lemma~\ref{lemClassical}(\ref{classical:one}),  $(\mathcal M,w) \models \ps \psi $ if and only if there exits $v\ler w$ such that $(\mathcal M,v) \models \psi $.
In view of Lemma	\ref{lemQuotPreserv}(\ref{itQuotPreservFU}), the analogous claim hods for $\nicefrac{\mathcal M}\sim$.

	Suppose first that $(\mathcal M,w) \models \ps \psi $.
	Then, there is $v\ler w$ such that $(\mathcal M,v) \models \psi $, so $[w] \mathrel{\nicefrac\rel\sim} [v]$ and the induction hypothesis yields $(\nicefrac{\mathcal M}\sim,[v]) \models \psi$.
	We conclude that $(\nicefrac{\mathcal M}\sim,[w]) \models \ps \psi$.

Conversely, suppose that $(\mathcal M,w) \not \models \ps \psi $ and let $v$ be such that $[w]\mathrel{ \nicefrac\rel\sim} [v] $.
We recall that the relation $\mathrel{ \nicefrac\rel\sim}$ corresponds to the {transitive} closure of the relation $\mathrel{ \nicefrac{\rel^0}\sim}$, defined as $[x]\mathrel{\nicefrac{\rel^0}\sim}[y]$ if there exists $x',y'$ such that $x\sim x' \rel y'\sim y$.
Thus, $[w]\mathrel{ \nicefrac\rel\sim} [v]$ means that there exist sequences $(w_i)_{i\leq n}$ and $(v_i)_{i\leq n}$ such that
	\[w = w_0 \sim  v_0 \rel  w_1 \sim  v_1 \rel  \ldots \rel w_n \sim v_n = v. \]
	By induction on $i\leq n$, one readily verifies that $\psi \notin \ell (w_i) \cup \ell (v_i)$, so that, in particular, $  \psi \not \in \ell (v)$.
By the induction hypothesis, $(\nicefrac{\mathcal M}\sim,[v]) \not \models \psi$.
Since $v \ler w$ was arbitrary, we obtain $(\nicefrac{\mathcal M}\sim,[w]) \not \models \ps \psi$.
	\medskip

	\noindent \textsc{Case} $\varphi = \nec\psi $. This case is treated very similarly to the $\ps$ case, but working `dually'. We show only that $(\mathcal M,w) \not \models \nec \psi $ implies that $(\nicefrac{\mathcal M}\sim,[w]) \not \models \nec \psi$ to illustrate.
	If $(\mathcal M,w) \not \models \nec \psi $, there are $v'\ler v\seq w$ such that $(\mathcal M,v') \not \models \nec \psi $.
	But then $ [v'] \mathrel{\nicefrac \ler \sim} [v] \mathrel{\nicefrac \seq \sim} [w]$, and the induction hypothesis yields $(\nicefrac{\mathcal M}\sim,[v']) \not \models \psi$; hence $(\nicefrac{\mathcal M}\sim,[w]) \not \models \nec \psi$.
\end{proof}

Quotients modulo $\Sigma$\term{-bisimulation} will be instrumental in proving the finite frame property for \altterm{def:csf}{$\csf$}.
However, $\Sigma$\term{-bisimulation} does not preserve \altterm{forth--down confluent}{forth--down confluence}, so to treat \altterm{def:sfi}{$\sfi$} and \altterm{def:gsfc}{$\gsfc$}, we will need a stronger notion of bisimulation. This stronger notion will also be convenient in treating \altterm{def:gsf}{$\gsf$}.

\begin{definition}
	A \define{strong} $\Sigma$\term{-bisimulation} is a $\Sigma$\term{-bisimulation} $Z$ such that both $Z$ and $Z^{-1}$ are \term{forth--down confluent} for $\peq$.
We denote the greatest \term{strong} $\Sigma$\term{-bisimulation} by $\approx_\Sigma$.
\end{definition}

 Notice that an arbitrary union of \term{strong} $\Sigma$\altterm{-bisimulation}{-bisimulations} is a $\Sigma$\term{-bisimulation} (and the empty relation is a \term{strong} $\Sigma$\term{-bisimulation}).
 \[\text{Thus a
	 greatest \term{strong} $\Sigma$\term{-bisimulation} exists, which we denote by $\approx_\Sigma$.}\]
	Since the set of \term{strong} $\Sigma$\altterm{-bisimulation}{-bisimulations} clearly contains the identity relation and is closed under compositions and under the converse operation, $\approx_{\Sigma}$ must be an equivalence relation.
The following is proven in essentially the same way as the forth--up confluence preservation clause of \Cref{lemQuotPreserv}.

\begin{lemma}\label{lemStrongQuotPreserve}
	Let $\mathcal M $ be a \term{bi-preordered model} and $\Sigma$ be a set of formulas.
	Let $\approx$ be a \term{strong} $\Sigma$\term{-bisimulation} on $\mathcal M$ that is also an equivalence relation.
	Then if $\mathcal M$ is \term{forth--down confluent}, so is $\nicefrac{\mathcal M}{\approx}$.
\end{lemma}

\begin{lemma}\label{lemSingletons}
	Let $\mathcal M$ be a \term{bi-preordered model}, and $\Sigma$ be a set of formulas closed under subformulas. Then every $\approx_\Sigma$-equivalence class of $\nicefrac{\mathcal M}{\approx_\Sigma}$ is a singleton.
\end{lemma}

\begin{proof}
	Suppose that $[w] \approx_\Sigma [v] $; we must show that $ w  \approx_\Sigma  v  $ as well, to conclude $[w] = [v]$ (note that $\approx_\Sigma$ is defined both on $\mathcal M$ and $\nicefrac{\mathcal M}\sim$). Define a relation $Z \subseteq W\times W$ given by $x \mathrel Z y$ if $[x] \approx_\Sigma [y]$.
	Clearly $ w\mathrel Z v$, so it remains to check that $Z$ is a $\Sigma$\term{-bisimulation} to conclude that $ w \approx_\Sigma v$.
	We have that $Z$ preserves labels in $\Sigma$ by the truth lemma for quotients (\Cref{lemQuotTruth'}), and we check only the `forth' clause, as the `back' clause is symmetric.
	
	Suppose that $x' \seq   x \mathrel Z y$.
	Since $[x] \approx_\Sigma [y] $, there is $y'$ such that $ [x'] \approx_\Sigma [y'] \mathrel {\nicefrac{\seq}{\approx_\Sigma}} [y]$.
	Since $\approx_\Sigma$ is both a $\Sigma$\term{-bisimulation} and an equivalence relation, we can apply \Cref{lemQuotLeq}, which says we may assume that $y'$ is chosen so that $y \peq  y ' $.
	From $[x'] \approx_\Sigma [y']  $ we obtain $x'\mathrel Z y'$, as needed.
\end{proof}

We remark that the quotient $\nicefrac{\mathcal M}{\sim}$ may be infinite, even taking $ {\sim} = {\approx_\Sigma} $.
However, in the following subsections we identify classes of models that do have finite quotients, and this will be the central ingredient in proving the finite \term{birelational frame} property for our \altterm{regular}{$\ps$-regular} logics.

\subsection{The finite frame property for \texorpdfstring{\altterm{def:gsf}{$\gsf$}}{GS4}}\label{sec:godelfinite}

In this subsection, we prove our second finite frame property result, showing that the G\"odel modal logic \altterm{def:gsf}{$\gsf$} indeed has the finite frame property. The \altterm{upward linear}{upward-linear} nature of \altterm{def:gsf}{$\gsf$} frames enables a particularly direct proof.

There are two steps to the proof, illustrated in \Cref{figShallowgs4}.
First we note that we may assume frames are not only upward but also \emph{downward} linear.
This is obtained by an initial transformation, as illustrated in the second model of the figure; note that the first two models are bisimilar, as bisimulations do not `look downward'.\footnote{Here it is worth noting that we do need to `look downward' to some extent to preserve the \term{forth--down confluent}[forth--down confluence] of \altterm{def:gsfc}{$\gsfc$}, which does require a bit of additional care. For this reason, \altterm{def:gsfc}{$\gsfc$} is treated separately in Section~\ref{sec:godelcfinite}, although the broad strokes of the proof are very similar.}
Then, given a subformula-closed subset $\Sigma$ of $\lanfull$, we describe a bisimulation quotient construction that yields finite models whenever $\Sigma$ is finite, as shown in the third model of \Cref{figShallowgs4}.
It is worth noting that as long as $\Sigma$ is finite, this third model is finite even if the second model is not; for example, we would obtain the same result if we replaced the $p,q$-worlds by an infinite chain.
We describe the quotient explicitly, noting afterwards that it is in fact the quotient by the greatest \term{strong} $\Sigma$\term{-bisimulation} $\approx_\Sigma$ defined in \Cref{secSBisim}.

 \begin{figure}[h!]\centering
 	\resizebox{.7\textwidth}{!}{
	\begin{tikzpicture}[node distance=1.5cm,auto,text width=5mm]
		\begin{scope}[node distance=1.5cm,local bounding box=scope2 ]
			\node[state,fill=tan] (x1)  {$p,q$};
			\node[state,fill=turquoise,below of=x1] (y1)  {$p$};
			\node[state,fill=turquoise,below of=y1] (z1)  {$p$};
			\node[state,fill=lightgray,below of=z1] (e1)  {};
			
			\node[state,fill=tan, right of=x1] (x2)  {$p,q$};
			\node[state, fill=turquoise,below of=x2] (y2)  {$p$};
			\node[state, fill=turquoise,below of=y2] (z2)  {$p$};
			
			\path[->,dashed] (y1) edge[] node[] {} (x1);
			\path[->,dashed] (z1) edge[] node[] {} (y1);
			\path[->,dashed] (e1) edge[] node[] {} (z1);
			\path[->,dashed] (y2) edge[] node[] {} (x2);
			\path[->,dashed] (z2) edge[] node[] {} (y2);
		\end{scope}

		\begin{scope}[node distance=1.5cm,local bounding box=scope1]
			\node[state,fill=tan,left of=x1,node distance=4cm] (pq)  {$p,q$};
			\node[state,fill=turquoise,left of=y1, node distance=4cm] (p)  {$p$};
			\node[state,fill=turquoise,left of=z1,node distance=3cm] (p1)  {$p$};
			\node[state,fill=turquoise,left of=z1, node distance=5cm] (p2)  {$p$};
			\node[state,left of =e1,fill=lightgray, node distance=5cm] (e)  {};	
			\path[->,dashed] (p) edge[] node[] {} (pq);		
			\path[->,dashed] (p1) edge[] node[] {} (p);		
			\path[->,dashed] (p2) edge[] node[] {} (p);		
			\path[->,dashed] (e) edge[] node[] {} (p2);				
		\end{scope}
			 			
		\begin{scope}[node distance=1.5cm,shift={($(x2.east)+(4cm,0)$)},local bounding box=scope3 ]
			\node[state,fill=tan, right of=x2, node distance=4cm] (xx1)  {$p,q$};
			\node[state,fill=turquoise,right of =y2, node distance=4cm] (yy1)  {$p$};
			\node[state,below of=yy1, fill=lightgray,node distance=1.5cm] (ee1)  {};
			\node[state,fill=tan, right of=xx1] (xx2)  {$p,q$};
			\node[state,fill=turquoise, below of=xx2] (yy2)  {$p$};

			\path[->,dashed] (yy1) edge[] node[] {} (xx1);
			\path[->,dashed] (ee1) edge[] node[] {} (yy1);
			\path[->,dashed] (yy2) edge[] node[] {} (xx2);
		\end{scope}	 	 	
				
		\node[left of=y1,node distance=2cm] (transition)  {\scalebox{1.9}{$\Longrightarrow$}};
		\node[right of=y2,node distance=2cm] (transition2)  {\scalebox{1.9}{$\Longrightarrow$}};
	\end{tikzpicture}
}	
	\caption{Transformations to obtain a bisimilar shallow model from an \term{upward linear} model. The intuitionistic relation $\peq$ is the transitive, reflexive closure of the relation indicated by dashed arrows.}\label{figShallowgs4}
\end{figure}

\medskip

Before diving into formal details, \Cref{figBisQuo} provides some additional intuition as to why birelational semantics are helpful for obtaining a finite frame property.
In the model on the left of \Cref{figBisQuo}, the modal accessibility relation is the reflexive closure of a bijection, and this would allow us to view this as a \altterm{gk_model}{G\"odel--Kripke model} by assigning truth values in $[0,1]$ to each variable at each linear submodel.
For example, we can let $w$ represent the linear order on the left and $v$ the linear order on the right.
Then, we can set  $V_w(p) :=\nicefrac 23$ and $V_v(q) :=\nicefrac 13$ to obtain a \altterm{gk_model}{G\"odel--Kripke model} validating the same formulas.
This is no longer the case for the right--hand model, where points are not `horizontally aligned'.
The advantage of giving up on concrete truth values is the added flexibility on the semantics needed for our proof techniques.
We remark that it is possible to perform an `unwinding' operation to convert a birelational model into a (typically infinite) \altterm{gk_model}{G\"odel--Kripke model}~\cite{AguileraDFM25}.
However, in our setting we know the two semantics validate the same formulas simply because they share an axiomatisation, so such an unwinding is not needed.
 \smallskip

\noindent\emph{Downward-linear \altterm{def:gsf}{$\gsf$} frames.} First we argue that without loss of generality, models based on \altterm{def:gsf}{$\gsf$} frames may be assumed downward linear. Let $\mathcal M = (W,{\peq},{\rel},V)$ be an arbitrary model based on a \altterm{def:gsf}{$\gsf$} frame. We construct a  model $\mathcal M' = (W',{\peq'},{\rel'},V')$ based on a downward-linear \altterm{def:gsf}{$\gsf$} frame and falsifying exactly the same $\lanfull$-formulas as $\mathcal M$.

The idea is to include copies of each of the (necessarily linear) principal upsets $v^\uparrow$, indexed by the generating world $v$. We define 
\begin{align*}
W' &= \{(v, w) \in W \times W \mid v \peq w\},\\
(v_1, w_1) \peq' (v_2, w_2) &\iff v_1 = v_2\text{ and }w_1 \peq w_2,\\
(v_1, w_1) \rel' (v_2, w_2) &\iff w_1 \rel w_2,\\
V'(v, w) &= V(w).
\end{align*}
As $\peq$ is an \altterm{upward linear}{upward-linear} partial order, clearly $\peq'$ is a partial order that is both upward and downward linear. As $\rel$ is a preorder, clearly $\rel'$ is a preorder. For forth--up confluence, suppose $(u_1, v_1) \rel' (u_2, v_2)$ and that $(u_1, v_1) \peq' (u_1, w_1)$. Then $v_1 \rel v_2$ and $v_1 \peq w_1$, so by forth--up confluence of $\rel$ there exists $w_2$ with $w_1 \rel w_2$ and $v_2 \peq w_2$. By transitivity of $\peq$, we have $u_2 \peq w_2$, so that $(u_2, w_2) \in W'$. Then $(u_1, w_1) \rel' (u_2, w_2)$ and $(u_2, v_2) \peq' (u_2, w_2)$, as required. Back--up confluence is entirely analogous. Thus $(W',{\peq'},{\rel'})$ is a \altterm{def:gsf}{$\gsf$} frame. 

It is straightforward to check by structural induction on formulas that for any $\lanfull$-formula $\varphi$ and any $v, w \in W$, we have \[(\mathcal M', (v, w)) \models  \varphi \iff (\mathcal M, w) \models  \varphi.\]
We verify the left-to-right implication for the case that $\varphi$ is of the form $\nec \psi$, since this will be instructive when we switch our focus to \altterm{def:gsfc}{$\gsfc$}. So suppose $(\mathcal M', (v, w)) \models  \nec\psi$ and, let $w \peq u$ and $u \rel z$. Then $(v, w) \peq' (v,u)$ and $(v, u) \rel' (z,z)$, so $(\mathcal M', (z, z)) \models  \psi$. Thus by the inductive hypothesis $(\mathcal M, z) \models  \psi$. Since $u$ and $z$ were arbitrary, subject to $w \peq u$ and $u \rel z$, we have $(\mathcal M, w) \models  \nec\psi$, as required.

Thus $\mathcal M'$ falsifies precisely the $\lanfull$-formulas falsified by $\mathcal M$. In other words, in proving the finite frame property for \altterm{def:gsf}{$\gsf$} we may work exclusively with models based on downward-linear  \altterm{def:gsf}{$\gsf$} frames.
\smallskip

\noindent\emph{Finite quotients.} Now let $\Sigma$ be a subformula-closed subset of $\lanfull$. We will describe a quotient of any model based on a downward-linear \altterm{def:gsf}{$\gsf$} frame that falsifies the same formulas from $\Sigma$ as the original model and is finite whenever $\Sigma$ is finite. 

Let $\mathcal {M} = (W,{\peq},{\rel}, V)$ be a model based on a downward-linear \altterm{def:gsf}{$\gsf$} frame. For $w\in W$.
Recall that we defined $\ell (w)$ by
\[\ell (w)=\{\varphi\in \Sigma \mid (\mathcal M,w) \models \varphi \}\text{,}\] and define \[L (w) = \{ \ell(v) \mid w \peq v\text{ or }v \peq w\}.\] We define an equivalence relation $\approx$ on $W$ by \[w \approx v \iff (\ell (w), L (w)) = (\ell (v), L (v)).\]
If $\Sigma$ is finite, then clearly $W / {\approx}$ is finite.

It is straightforward to show that $\approx$ is a \term{strong} $\Sigma$\term{-bisimulation}. Conversely, it is clear that if $w$ and $v$ are strongly $\Sigma$-bisimilar, then $\ell (w) = \ell (v)$ and $L (w) = L (v)$, and hence $w \approx v$. Thus $\approx$ is the greatest \term{strong} $\Sigma$\term{-bisimulation} $\approx_\Sigma$.

By \Cref{lemQuotPreserv}, the quotient $\mathcal M /{\approx}$ has all the properties necessary to be a \altterm{def:gsf}{$\gsf$} frame. It is clear that if $\Sigma$ is finite then $\mathcal M /{\approx}$ is finite.

\begin{theorem}[{\altterm{def:gsf}{$\gsf$}} finite frame property]\label{theorem:g4s_finite_model}
The logic \altterm{def:gsf}{$\gsf$} has the finite \term{birelational frame} property. That is, if a formula $\varphi$ is \term{falsifiable} (respectively~\term{satisfiable}) on a birelational \altterm{def:gsf}{$\gsf$} frame, then $\varphi$ is \term{falsifiable} (respectively \term{satisfiable}) on a finite birelational \altterm{def:gsf}{$\gsf$} frame.
\end{theorem}

 In order to use \Cref{theorem:g4s_finite_model} to prove decidability, we need to compute a bound on the size of a model $\mathcal M /{\approx}$ in terms of the size of $\Sigma$, when $\Sigma$ is finite. Note that $\nicefrac \peq \approx$ is given by the partial order
\[[w] \mathrel{\nicefrac \peq \approx} [v] \iff L(w) = L(v)\text{ and }\ell(w) \supseteq \ell(v).\]

\begin{lemma}\label{quasi:bound}
Suppose $\Sigma$ is finite. 
 Then, the cardinality of the domain of $\mathcal M /{\approx}$ is bounded by $(\lgt\Sigma +1)\cdot 2^{\lgt \Sigma (\lgt \Sigma +1)+1}$.
\end{lemma}

\begin{proof}
Each element of the domain of $\mathcal M /{\approx}$ may be identified with a pair $(\ell, L)$ where $L$ is a (nonempty) subset of $\raisebox{2pt}{$\wp$} \Sigma$ and $\ell \in L$.
There are at most $(2^{\lgt \Sigma})^i$ subsets of $\raisebox{2pt}{$\wp$} \Sigma$ of size $i$, so there are at most $\sum_{i = 1}^{\lgt \Sigma +1}(2^{\lgt \Sigma})^i$ distinct $L$. The sum is bounded by $2^{\lgt \Sigma (\lgt \Sigma +1)+1}$. The factor of $\lgt\Sigma +1$ corresponds to the choice of an $\ell \in L$, for each $L$.
\end{proof}

Thus we have an exponential bound on the size of $\mathcal M /{\approx}$. Hence any \term{falsifiable} formula is \term{falsifiable} on an effectively bounded model, and it follows that \altterm{def:gsf}{$\gsf$} is decidable.

\begin{theorem}[{\altterm{def:gsf}{$\gsf$}} complexity]\label{theorem:decide}
The logic \altterm{def:gsf}{$\gsf$} is decidable in \textsc{nexptime}.
\end{theorem}

\begin{proof}
 Since \altterm{falsifiable}{falsifiability} is the complement of \altterm{valid}{validity}, it suffices to show that it is decidable whether a formula $\varphi$ is \term{falsifiable} over the class of all \altterm{def:gsf}{$\gsf$} frames. Let $\Sigma$ be the set of subformulas of $\varphi$. The cardinality of $\Sigma$ is no greater than the length of $\varphi$. If $\varphi$ is \term{falsifiable} in a \altterm{def:gsf}{$\gsf$} frame of size at most $(\lgt\Sigma +1)\cdot 2^{\lgt \Sigma (\lgt \Sigma +1)+1}$, then in particular $\varphi$ is \term{falsifiable} in a \altterm{def:gsf}{$\gsf$} frame. Conversely, if $\varphi$ is falsified in a \altterm{def:gsf}{$\gsf$} frame, then by \Cref{quasi:bound}, $\varphi$ is falsified in a \altterm{def:gsf}{$\gsf$} frame of size at most $(\lgt\Sigma +1)\cdot 2^{\lgt \Sigma (\lgt \Sigma +1)+1}$.
Hence, it suffices to check falsifiability of $\varphi$ on the set of all \altterm{def:gsf}{$\gsf$} frames (with \altterm{valuation}{valuations} only for the variables appearing in $\varphi$) of size at most $(\lgt\Sigma +1)\cdot 2^{\lgt \Sigma (\lgt \Sigma +1)+1}$. It is clear that this check can be carried out within a computable time bound; hence the problem is decidable.

The \textsc{nexptime} bound is then obtained as in the proof of Theorem~\ref{thmCS4dec}. 
\end{proof}


\subsection{The finite frame property for \texorpdfstring{\altterm{def:gsfc}{$\gsfc$}}{GS4c}}\label{sec:godelcfinite} We now proceed to proving the finite \term{birelational frame} property for \altterm{def:gsfc}{$\gsfc$}. The strategy is the same as for \altterm{def:gsf}{$\gsf$}. However, the result does not follow immediately from our \altterm{def:gsf}{$\gsf$} constructions, because our construction of downward-linear \altterm{def:gsf}{$\gsf$} frames does not preserve \altterm{forth--down confluent}{forth--down confluence} and hence does not specialise to \altterm{def:gsfc}{$\gsfc$} frames. We thus use a variant of our previous bisimulation quotient construction. Before we define this, we need to show that without loss of generality, \altterm{def:gsfc}{$\gsfc$} frames have an additional property.

\begin{definition}
Let $(W, \peq)$ be a preordered set. We say a binary relation $R$ on $W$ is \define{pointwise convex}, if it validates the implication \[R(u,v_1), R(u,v_2),  v_1 \peq w \peq v_2 \implies R(u,w).\]
\end{definition}

\begin{lemma}\label{lemma:convex}
Let $(W, \peq)$ be a preordered set and $R$ a binary relation on $W$. Then, the relation $\overline R$ defined by 
\[\overline R (u,w) \iff \exists v_1, v_2 :   R(u,v_1), R(u,v_2), v_1 \peq w \peq v_2\]
is \term{pointwise convex}. We call this the \define{convex closure} of $R$. 
If $R$ is \altterm{forth--up confluent}{forth--up} and \term{forth--down confluent} and {transitive}, then so is $\overline R$. 
\end{lemma}

\begin{proof}
The claims are straightforward to prove; we give the details for the last claim (transitivity). So suppose $R$ is \altterm{forth--up confluent}{forth--up} and \term{forth--down confluent} and {transitive} and that $\overline R(u,v)$ and $\overline R(v,w)$. Then there is $v' \peq v$ with $R(u,v')$ and $w' \peq w$ with $R(v, w')$. By \altterm{forth--down confluent}{forth--down confluence} of $R$ there is $w'' \peq w'$ with $R(v',w'')$. Then by transitivity of $\peq$ and $R$ respectively, we have $w'' \peq w$ and $R(u,w'')$. By a symmetric argument using forth--up confluence, there also exists $z \seq w$ with $R(u,z)$. So by definition $\overline R(u,w)$.
\end{proof}

By \Cref{lemma:convex}, if $(W,{\peq},{\rel})$ is a \altterm{def:gsfc}{$\gsfc$} frame, then so is $(W,{\peq},{\overline\rel})$. Given a \term{valuation} $V$, it is straightforward to show, by structural induction on formulas and using monotonicity of $V$, that the satisfaction relation $\models$ on $(W,{\peq},{\overline\rel},V)$ is identical to that on $(W,{\peq},{\rel},V)$. Hence 
we may assume, without loss of generality, that \altterm{def:gsfc}{$\gsfc$} frames are \term{pointwise convex}.

\noindent\emph{Downward-linear \altterm{def:gsfc}{$\gsfc$} frames.} Let $\mathcal M = (W,{\peq},{\rel},V)$ be an arbitrary model based on a \term{pointwise convex} \altterm{def:gsfc}{$\gsfc$} frame. We define a model $\mathcal M'' = (W'',{\peq''},{\rel''},\allowbreak V'')$ based on a downward-linear \altterm{def:gsfc}{$\gsfc$} frame and falsifying exactly the same $\lanfull$-formulas as $\mathcal M$ as follows. 
(The only difference with $\mathcal M'$ from the previous subsection is in the definition of $\rel''$.)
\begin{align*}
W'' &= \{(v, w) \in W \times W \mid v \peq w\},\\
(v_1, w_1) \peq'' (v_2, w_2) &\iff v_1 = v_2\text{ and }w_1 \peq w_2,\\
(v_1, w_1) \rel'' (v_2, w_2) &\iff v_1 \rel v_2\text{ and } w_1 \rel w_2,\\
V''(v, w) &= V(w).
\end{align*}
It is clear that $\mathcal M''$ is an upward- and downward-linear \term{bi-preorder}. We now check the confluence conditions. For forth--up confluence, suppose $(u_1, v_1) \rel'' (u_2, v_2)$ and that $(u_1, v_1) \peq'' (u_1, w_1)$. Then $v_1 \rel v_2$ and $v_1 \peq w_1$, so by forth--up confluence of $\rel$ there exists $w_2$ with $w_1 \rel w_2$ and $v_2 \peq w_2$. By transitivity of $\peq$, we have $u_2 \peq w_2$, so that $(u_2, w_2) \in W''$. Then $(u_1, w_1) \rel'' (u_2, w_2)$ and $(u_2, v_2) \peq'' (u_2, w_2)$, as required. Back--up confluence is entirely analogous. For \altterm{forth--down confluent}{forth--down confluence}, suppose $(u_1, w_1) \rel'' (u_2, w_2)$ and that $(u_1, v_1) \peq'' (u_1, w_1)$. We have $v_1 \seq u_1 \rel u_2$, so by forth--up confluence, there exists $v_2$ with $v_1 \rel v_2 \seq u_2$. By \altterm{upward linear}{upward linearity}, either $v_2 \peq w_2$ or $w_2 \peq v_2$. In the first case $(u_2,v_2) \peq'' (u_2, w_2)$ and we are done. In the second case, apply \altterm{forth--down confluent}{forth--down confluence} to $v_1 \peq w_1 \rel w_2$ to obtain $z$ with $v_1 \rel z \peq w_2$. Then since $v_1 \rel z$ and $v_1 \rel v_2$, and also $z \peq w_2 \peq v_2$, we have, by \term{pointwise convex}ity, that $v_1 \rel w_2$. Then we have $(u_1,v_1) \rel'' (u_2,w_2) \peq'' (u_2,w_2)$, witnessing the \altterm{forth--down confluent}{forth--down confluence} of $\rel''$. We conclude that $(W'',{\peq''},{\rel''})$ is a downward-linear \altterm{def:gsfc}{$\gsfc$} frame. 

It is straightforward to check by structural induction on formulas that for any $\lanfull$-formula $\varphi$ and any $v, w \in W$, we have \[(\mathcal M'', (v, w)) \models  \varphi \iff (\mathcal M, w) \models  \varphi.\]
We verify the left-to-right implication for the case that $\varphi$ is of the form $\nec \psi$. So suppose $(\mathcal M'', (v, w)) \models  \nec\psi$ and, let $w \peq u$ and $u \rel z$. Then $(v, w) \peq'' (v,u)$ and $(v, u) \rel'' (z,z)$, so $(\mathcal M', (z, z)) \models  \psi$. Then $v \peq u \rel z$, so by \altterm{forth--down confluent}{forth--down confluence} there exists $y$ with $v \rel y \peq z$. Then $(v, w) \peq'' (v,u)$ and $(v, u) \rel'' (y,z)$, so $(\mathcal M'', (y, z)) \models  \psi$. Thus by the inductive hypothesis $(\mathcal M, z) \models  \psi$. Since $u$ and $z$ were arbitrary, subject to $w \peq u$ and $u \rel z$, we have $(\mathcal M, w) \models  \nec\psi$, as required.

Thus $\mathcal M''$ falsifies precisely the $\lanfull$-formulas falsified by $\mathcal M$. So in proving the finite frame property for \altterm{def:gsfc}{$\gsfc$} we may work exclusively with models based on downward-linear  \altterm{def:gsfc}{$\gsfc$} frames.

Now let $\Sigma$ be a subformula-closed subset of $\lanfull$. Given a downward-linear \altterm{def:gsfc}{$\gsfc$} frame $\mathcal M$, we know from \Cref{sec:godelfinite} that the quotient $\nicefrac{\mathcal M}\approx$ is a \altterm{def:gsf}{$\gsf$} frame of size at most $(\lgt\Sigma +1)\cdot 2^{\lgt \Sigma (\lgt \Sigma +1)+1}$ falsifying precisely the same formulas as $\mathcal M$. By \Cref{lemStrongQuotPreserve}, $\nicefrac{\mathcal M}{\approx}$ is also \term{forth--down confluent} and hence a \altterm{def:gsfc}{$\gsfc$} frame. This yields the following results.

\begin{theorem}[{\altterm{def:gsfc}{$\gsfc$}} finite frame property]\label{thm:gs4c_finite_frame}
	The logic \altterm{def:gsfc}{$\gsfc$} has the finite \term{birelational frame} property. That is, if a formula $\varphi$ is \term{falsifiable} (respectively~\term{satisfiable}) on a birelational \altterm{def:gsfc}{$\gsfc$} frame, then $\varphi$ is \term{falsifiable} (respectively \term{satisfiable}) on a finite birelational \altterm{def:gsfc}{$\gsfc$} frame.
\end{theorem}

\ignore{
\begin{proof}
	If $\varphi \not\in \gsfc$ then choose a \term{prime} filter $\Gamma$ avoiding $\varphi$.
	Then $(\Gamma,\Gamma) \notin \val\varphi_\Sigma$.
	Since $\mathcal M_\Sigma$ is a model based on a downward-linear birelational \altterm{def:gsfc}{$\gsfc$} frame, $\nicefrac{\mathcal M_\Sigma}{\approx_\Sigma}$ is a finite model based on a \altterm{def:gsfc}{$\gsfc$} frame also falsifying $\Sigma$.
\end{proof}
}

\begin{theorem}[{\altterm{def:gsfc}{$\gsfc$}} complexity]\label{theorem:g4sc_nexptime}
The logic \altterm{def:gsfc}{$\gsfc$} is decidable in \textsc{nexptime}. 
\end{theorem}



\subsection{Shallow frames}\label{sShallow}

Our finite frame property proof for  \altterm{def:sfi}{$\sfi$} follows some the same general pattern as those for G\"odel--Dummett logics, but there are a few subtleties.

First, once branching is involved, it is no longer the case that every model will be bisimilar to a finite one.
This is illustrated in \Cref{figShallow}.
On the left-hand side, for $\Sigma := \{p,q\}$, it can be checked that the two $p$-points are $\Sigma $-bisimilar to each other, but no other two points are $\Sigma $-bisimilar, so this model's $\Sigma $-bisimulation quotient will not be much smaller.
This model is not {\em too} large for the sake of illustration, but one can imagine the chain of $p$-worlds being much longer (or even infinite) or more complex branching below the $q$-world.
Regardless, we can simplify the model by making the following two observations:
\begin{enumerate}

\item When $w\prec v$ are such that $\ell(w) = \ell(v)$, we can modify the model so that $w\not\peq v$ without affecting $\Sigma$-labels.
In this case, we will say that $v$ is \define{redundant}.

\item When $w\not\peq v$, $v$ does not affect how formulas are evaluated on $w$, so $v$ can be removed from the model without affecting $\ell(w)$.

\end{enumerate}

These two observations lead to the middle model of Figure~\ref{figShallow}, where only one world with an empty label remains (the original evaluation point) and the two $p$-worlds are no longer related via $\peq$, although they are still both related to the root. 
The resulting model is a tree (or, in the general case, a forest) and has the property that if $w_0\prec w_1\prec \ldots \prec w_n$, then each $w_i $ has a distinct label.
Since there are finitely many formulas in $\Sigma$, this tells us that $n<|\Sigma|$, i.e., the \term{depth} of the new model is bounded by $|\Sigma| $.
In our particular example, we obtained a model of depth $1$, but this could be increased to $2$ by appending a $p,q$-world (i.e., one where both variables are true).
However, this is the maximal depth for our particular choice of $\Sigma$.

In the next step, we observe that many of the points from our \term{shallow} model are $\Sigma$-bisimilar, and we can thus identify them.
It is a general fact that any \term{shallow}, \term{forest-like} model is bisimilar to a finite one, although the bounds can be quite large (see Lemma~\ref{lemStrongBisBound}).

 \begin{figure}[h!]\centering
 	\resizebox{.7\textwidth}{!}{
 	\begin{tikzpicture}[node distance=1.5cm,auto]
 		\begin{scope}[node distance=1.5cm,local bounding box=scope1]
 			\node[state,fill=lightgray,line width=0.4mm] (i)  {};
 			\node[state,fill=turquoise,above left of=i] (p)  {$p$};
 			\node[state,fill=turquoise, above of=p] (r)  {$p$};	 	
 			\node[state,fill=lightgray,above right of=i] (x)  {};
 			\node[state,fill=tan,above of=x] (q)  {$q$};
 			\node[state,fill=lightgray, below right of=x] (z)  {};
 			\path[->,dashed] (i) edge[] node[] {} (p);
 			\path[->,dashed] (i) edge[] node[] {} (x);
 			\path[->,dashed] (x) edge[] node[] {} (q);
 			\path[->,dashed] (p) edge[] node[] {} (r);
 			\path[->,dashed] (z) edge[] node[] {} (x);
 		\end{scope}
 		 	
 		\begin{scope}[node distance=1.5cm,shift={($(i.east)+(6cm,0)$)},local bounding box=scope2 ]
 			\node[state,fill=lightgray,line width=0.4mm] (ii)  {};
			\node[state,fill=turquoise,above left of=ii] (pp)  {$p$};
			\node[state,fill=turquoise, above of=pp] (rr)  {$p$};	 	
			\node[above right of=ii] (xx)  {};
			\node[state,fill=tan,above of=xx] (qq)  {$q$};
			\path[->,dashed] (ii) edge[] node[] {} (pp);
			\path[->,dashed] (ii) edge[] node[] {} (qq);
			\path[->,dashed] (ii) edge[] node[] {} (rr);
 		\end{scope}	 	
 		
 		\begin{scope}[node distance=1.5cm,shift={($(ii.east)+(5cm,0)$)},local bounding box=scope3 ]
 			\node[state,fill=lightgray,line width=0.4mm] (iii)  {};
 			\node[state,fill=turquoise,above left of=iii] (ppp)  {$p$};
 			\node[state,fill=tan,above right of=iii] (qqq)  {$q$};
 			\path[->,dashed] (iii) edge[] node[] {} (ppp);
 			\path[->,dashed] (iii) edge[] node[] {} (qqq);
 		\end{scope}	 	 	
 		
 		\node[right of=scope1,node distance=3cm] (transition)  {\scalebox{1.9}{$\Longrightarrow$}};
 		\node[right of=scope2,node distance=3cm] (transition2)  {\scalebox{1.9}{$\Longrightarrow$}};
 	\end{tikzpicture}
 	}	
 	\caption{A two-step transformation, first from an intuitionistic model to a shallow model, then from a shallow model to its bisimulation quotient.
 	As before, the intuitionistic relation is the transitive, reflexive closure of the relation indicated by dashed arrows.
 	The evaluation point is indicated by a bold outline and each point is labelled by its true formulas from $\Sigma:=\{p,q\}$.}\label{figShallow}
 \end{figure}

In summary, our general strategy for proving the finite model property for \altterm{def:sfi}{$\sfi$} is to first convert any model to a \term{shallow} one and then apply a bisimulation quotient.
Let us make these notions precise. 
Given a frame $\mathcal F = ( W,{\peq} ,{\rel} )$ and $w\in W$, the \define{depth} of $w$ is the supremum, in $\mathbb N \cup \{\infty\}$, of all $n \in \mathbb N$ such that there is a sequence
\[w = w_0 \prec w_1 \prec \ldots \prec w_n,\] 
where $w\prec v$ means that $w\peq v$ but $v\not\peq w$.
The \define{depth2}[depth] of the frame $\mathcal F  $ is the supremum of all \altterm{depth}{depths} of elements of $W$.
The depth of a model is the \term{depth} of its underlying frame.
Note that the \term{depth} of worlds or frames/models could be $\infty$.
If a frame or model has \emph{finite} \term{depth}, we say that it is \define{shallow}.
Shallow models will provide an important intermediate step towards establishing the finite frame property, as the bisimulation quotient of a \term{shallow} model is finite.
Nevertheless, it can be quite large, as it is only superexponentially bounded.

Let $2^m_k$ be the superexponential function defined recursively by $2^m_0 = m$ and $2^m_{k+1} = 2^{2^m_k}$.
The following inequality will be useful in order to simplify some expressions involving superexponentials.

\begin{lemma}
	For all $m,n,k\geq 1$: \[2^m\cdot 2_{k}^{(n-1)m}\leq 2_{k}^{nm}.\]
\end{lemma}

\proof
Proceed by induction on $k$. If $k=1$, then
\[2^m\cdot 2_{k}^{(n-1)m}=2^m\cdot 2^{(n-1)m}=2_{k}^{n m}.\]
If $k>1$, then note that $1 \leq  2_{k-1}^{(n-1)m}$, so that
\[m+ 2_{k-1}^{(n-1)m} \leq (m+1) 2_{k-1}^{(n-1)m} \leq 2^m\cdot 2_{k-1}^{(n-1)m} \stackrel{\text{\textsc{IH}}} \leq 2_{k-1}^{nm} .\]
Then
\[2^m\cdot 2_{k}^{(n-1)m}=2^m \cdot 2^{ 2_{k-1}^{(n-1)m}}=2^{m+ 2_{k-1}^{(n-1)m}} \leq 2^{2_{k-1}^{nm}} = 2_{k}^{nm},\]
as needed.
\endproof

Much as was the case for G\"odel--Dummett logics, when proving the finite \term{bi-preorder}[bi-preordered] frame property for \altterm{def:sfi}{$\sfi$}, it is convenient to work with downward linear frames.
Since, in this case, they will not be \term{upward linearity}[upward linear], the downward linear \altterm{def:sfi}{$\sfi$} will be \term{forest-like}.
A bi-intuitionis\-tic model $\mathcal M = (W,{\peq},{\rel},V)$ is \define{forest-like} if for every $w\in W$, the set $\{v\in W \mid v\peq w\}$ is totally ordered by $\peq$.

\begin{lemma}\label{lemStrongBisBound}
	Given a \term{forest-like} \term{} \term{bi-preordered model} $\mathcal M = (W,{\peq},{\rel},\allowbreak V)$ of finite \term{depth} $n$ and finite $\Sigma \subseteq \lanfull$ with $\lgt\Sigma = s$, we have $ \lgt{ \nicefrac W{\approx_\Sigma}} \leq  2^{(n+1)s}_{n+1}$.
\end{lemma}

\proof[Proof sketch]
This is proven in some detail in for example~\cite{BalbianiToCL}, but we outline the main elements of the proof.
Let $W_n$ be the restriction of $W $ to the set of points of depth at most $n$.
Proceed by induction on $n\in \mathbb N$ to show that there are at most $2^{(n+1)s}_{n+1} $ \term{strong} $\Sigma$\term{-bisimulation} classes on $W_n$.
Since $W=W_n$ when $n$ is the \term{depth} of $\mathcal M$, the claim follows.

Let $w\in W_0$.
Then, its bisimulation class in $W_0$ is uniquely determined by its label $\ell(w) \in 2^\Sigma $, and there are at most $2^{s} = 2^{s}_1  $ choices for this label.

For the inductive step, let $w\in W_{n}$ with $n>0$.
Note that each immediate successor of $w$ belongs to $W_{n-1}$.
Let $\{[v_i] \mid i\in I\}$ enumerate the equivalence classes of the immediate successors of $w$ in $W_{n-1}$.
By the induction hypothesis, there are at most $2^{ns}_{n} $ choices for $[v_i]$ in $W_{n-1}$, and the strong bisimulation class of $w$ in $W_{n}$ is determined by its label, for which there are $ 2^{ s} $ choices, and a possible choice of  $\{[v_i]: i\in I\}$, of which there are at most $2^{2^{ns}_{n}} $ choices.
Hence there are at most
\[ 2^{ s} \cdot 2^{2^{ns}_{n}} = 2^{ 2s  + 2^{ns}_{n}} \leq 2^{2^{ (n+1) s}_{n }} = 2^{ (n+1) s}_{n+1} 
\]
choices for the strong bisimulation class of $w$ in $W_{n}$.
\endproof

\begin{remark}
	Note that the \term{forest-like} assumption in \Cref{lemStrongBisBound} is needed, as in general there may be infinitely many $\approx_\Sigma$ equivalence classes of points if this assumption fails.
	In a model consisting of an infinite sequence
	\[w_0 \succ v_0 \prec w_1 \succ v_1 \prec w_2 \succ \ldots \]
	where $p$ is true only on $w_0$, no two points are strongly ${\{p\}}$-bisimilar.
\end{remark}

We conclude this section by showing that the \term{shallow} frame property implies the finite model property for \altterm{def:sfi}{$\sfi$}.  

\begin{proposition}\label{thmShallowtoFin}Let 
 $\varphi\in \lanfull$.
If $\varphi$ is \term{falsifiable} (respectively \term{satisfiable}) on a \term{shallow}, \term{forest-like} \altterm{def:sfi}{$\sfi$} frame, then $\varphi$ is \term{falsifiable} (respectively \term{satisfiable}) on a finite \altterm{def:sfi}{$\sfi$} frame of size at most $2^{(n+1)s}_{n+1}
$, where $n$ is the \term{depth} of $\mathcal M$ and $s$ is the length of $\varphi$.
\end{proposition}

\begin{proof}
Let $\mathcal M$ be a model based on a \term{forest-like} \altterm{def:sfi}{$\sfi$} frame of \term{depth} $n<\infty$ that \term{falsifiable}[falsifies] (respectively \term{satisfiable}[satisfies]) $\varphi$. Define $\Sigma$ to be the set of subformulas of $\varphi$. So $\Sigma$ is finite, closed under subformulas, and $s:=|\Sigma|$ is bounded by the length of $\varphi$.
Let $\approx_\Sigma$ denote the greatest \term{strong} $\Sigma$\term{-bisimulation} on $\mathcal{M}$.
Then, the quotient structure $\nicefrac{\mathcal M}{\approx_\Sigma}$ has size at most $2^{(n+1)s}_{n+1}
$ by \Cref{lemStrongBisBound}, is based on a \altterm{def:sfi}{$\sfi$} frame by Lemmas~\Cref{lemQuotPreserv} and \ref{lemStrongQuotPreserve}, and \term{falsifiable}[falsifies] (respectively \term{satisfiable}[satisfies]) $\varphi$ by the truth lemma for quotients (\Cref{lemQuotTruth'}).
\end{proof}

Thus in order to prove the finite frame property for \altterm{def:sfi}{$\sfi$}, it suffices to show that it has the \term{shallow} frame property (with respect to its intended class of frames): that any non-\term{valid} formula is \term{falsifiable} on a \term{shallow}  \altterm{def:sfi}{$\sfi$} frame.
This is the strategy that we will employ in the sequel. 

Note that we can, in a similar manner, prove versions of \Cref{thmShallowtoFin} for \altterm{def:csf}{$\csf$}, \altterm{def:gsfc}{$\gsfc$}, \altterm{def:gsf}{$\gsf$}, and \altterm{def:isf}{$\isf$}.
In fact, the finite \term{bi-preorder} frame property for \altterm{def:csf}{$\csf$}, \altterm{def:gsf}{$\gsf$}, and \altterm{def:gsfc}{$\gsfc$} can be proven adapting our proof for \altterm{def:isf}{$\isf$}, but, as we have seen, more direct means are available.
Meanwhile, our techniques do not yield the \term{shallow} model property for \altterm{def:isf}{$\isf$}.


\subsection{The finite frame property for \texorpdfstring{\altterm{def:sfi}{$\sfi$}}{S4I}}\label{secFMPS4I}

Our aim is to prove that \altterm{def:sfi}{$\sfi$} has the \term{shallow} frame property, from where the finite frame property will follow.
As was the case for G\"odel--Dummett logics, it is convenient to enforce downward linearity on our models.
In the case of \altterm{def:sfi}{$\sfi$}, this amounts to them being \term{forest-like}, in order to apply \Cref{thmShallowtoFin}.
The following construction will ensure this.
In this subsection, $W_{\mathrm c} ,{\peq_{\mathrm c} }$, etc.~refer to the components of $\mathcal M^{\sfi}_{\mathrm c} $ (see \Cref{secCompGS4}), and $W_\Sigma,{\peq_\Sigma}$, etc.~will refer to the respective components of the model $\mathcal M^{\sfi}_\Sigma$ to be constructed below.

\begin{definition}
For a set of formulas $\Sigma$ closed under subformulas, define $W_\Sigma $ to be the set of all tuples $ \Gamma= (\Gamma _0,\ldots,\Gamma _n )$, where
\begin{enumerate}
\item each $\Gamma _i\in W_{\mathrm c} $,
\item $ \Gamma _i \peq_{\mathrm c}  \Gamma _{i+1}  $ if $i<n$,
\item for each $i<n$ there is a formula $\varphi \in \Sigma$ such that $\varphi \in \Gamma _{i+1}\setminus \Gamma _{i} $.
\end{enumerate}
We define $V_\Sigma$ by $\Gamma\in V_\Sigma(p)$ if and only if $p\in  \Gamma_n \cap \Sigma$.
Define $\Gamma \peq_\Sigma \Delta$ if $\Gamma$ is an initial sequence of $\Delta$, and for $ \Gamma= (\Gamma _0,\ldots,\Gamma _n )$, define $\ell(\Gamma) =\Gamma_n\cap \Sigma$.
\end{definition}

\begin{lemma}\label{lemS4IShallow}
The relation $\peq_\Sigma $ is a \term{forest-like} partial order on $W_\Sigma$ and any strict $\peq_\Sigma $-chain has length at most $|\Sigma|+1$.
Moreover, $V_\Sigma$ is monotone.
\end{lemma}

\begin{proof}
That $\peq_\Sigma $ is a \term{forest-like} partial order is easily checked from the definitions. 
Now, let $\Gamma^0\prec_\Sigma \Gamma^1 \prec_\Sigma \ldots \prec_\Sigma \Gamma^ n$ be a chain in $W_\Sigma$.
We note that each $\Gamma^i$ is an initial segment of $\Gamma^ n := (\Gamma^ n_0,\ldots,\Gamma^ n_k)$, so it suffices to show that $k\leq |\Sigma|$.
But by definition, we have that $\Gamma^ n_i\subsetneq \Gamma^ n_{i+1} $ for each $i<k$, and since each $\Gamma^ n_i\subseteq \Sigma $, we have that $k\leq |\Sigma|$, as needed.
The monotonicity of $V_\Sigma$ follows from the elements of $\Gamma$ being ordered by $\peq_{\mathrm c} $.
\end{proof}

\begin{lemma}\label{lemTruthImpS4I}
For all $\Gamma \in W_\Sigma$ and $ \varphi \imp \psi \in \Sigma $, we have that $\varphi\imp \psi \in \ell  ( \Gamma ) $ if and only if whenever $\Delta \seq_\Sigma \Gamma$ with $\varphi\in \ell  ( \Delta) $, it follows that $\psi \in \ell ( \Delta )$.
\end{lemma}

\begin{proof}
First assume that $\Gamma= (\Gamma _0,\ldots,\Gamma _n  )$ is such that $\varphi\imp \psi \in \ell ( \Gamma ) $ and let $\Delta = (\Delta _0,\ldots,\Delta _m )$ be such that $\Gamma \peq_\Sigma \Delta$.
Then $m\geq n$ and $\Delta_n = \Gamma_n$, which by transitivity of $\peq_{\mathrm c} $ implies that $\Gamma _n\peq_{\mathrm c}  \Delta_m$.
It follows that if $\varphi\in \ell ( \Delta) $ then $\psi \in \ell( \Delta )$.

Conversely, suppose that $\varphi\imp \psi \in \Sigma\setminus \ell ( \Gamma ) $.
We may further assume, without loss of generality, that if $\varphi\in \ell(\Gamma)$ then $\psi\in \ell(\Gamma)$, for otherwise we may take $\Delta = \Gamma$.
Then there is $\Delta_{n+1} \seq_{\mathrm c}  \Gamma_n$ such that $\varphi\in \Delta_{n+1}$ but $\psi \not \in \Delta_{n+1}$.
For $i\leq n$ set $\Delta_i = \Gamma_i$, and define $\Delta = (\Delta_i)_{i\leq n+1}$.
It should then be clear that $\Delta$ has the desired properties.
\end{proof}

If $\Gamma$ is the left-hand model of Figure~\ref{figShallow}, the middle model represents the generated submodel of $(\Gamma)$ (i.e., the sequence of length one with unique element $\Gamma$), which is now a shallow forest.
This is already a major step in establishing the finite frame property, but the accessibility relation $\rel_\Sigma$ for \altterm{def:sfi}{$\sfi$} requires some care. This is illustrated in Figure~\ref{figS4IRel}.

The top of the figure shows an \altterm{def:sfi}{$\sfi$} model (in fact, it is also an \altterm{def:isf}{$\isf$} model, a fact which will become relevant soon). 
We can imagine that this model is a substructure of $\mathcal M_{\rm c}$, and the bottom of the figure illustrates the respective substructure of $\mathcal M_\Sigma$.
The relation $\rel_\Sigma$ will be defined as the transitive, reflexive closure of a `generator' $\rel^1_\Sigma$, and the two boxes on the bottom illustrate two possible choices for this $\rel^1_\Sigma$.

Let us first point out that the top model is already forest-like, so the transformation to $\mathcal M_\Sigma$ in this case is rather simple.
It may be viewed as deleting `redundant' points (to be clear, copies of these points will appear elsewhere in $\mathcal M_\Sigma$, but we will not focus on them for this discussion).
The restriction of $\rel_{\rm c}$ to $W_\Sigma$ is neither \term{forth--up confluent}[forth--up] nor \term{back--up confluent}, so we must amend it if we want to preserve either of these properties.

As our goal is to show that \altterm{def:sfi}{$\sfi$} has the finite frame property, our priority is for the resulting relation to be \term{forth--up confluent}, as pictured on the left.
As can be seen, we added an arrow from the $b$-point on the left to the $a$-point on the right.
This arrow was inherited from the deleted $a$-point, i.e., the `recipe' is to set $ w \rel^1_\Sigma v$ if there is $v '$ such that $\ell(v')= \ell(v)$ and $ w \rel_{\rm c} v' \seq_{\rm c} v$.
This recipe will always preserve \term{forth--up} confluence when deleting \term{redundant} points from a forest-like model and is the core intuition behind Definition~\ref{defAccS4I} below.

Before stating this definition, it will be instructive to ask whether we can also preserve \term{back--up} confluence.
The figure on the bottom right shows that this can be achieved by a `dual' approach, if we instead set $ w \rel^1_\Sigma v$ if there is $w '$ such that $\ell(w')= \ell(w)$ and $ w \peq_{\rm c} w' \rel_{\rm c} v$.
The issue here is that \term{forth--up} confluence is no longer preserved. (This construction could be used to give an alternative finite frame property proof for \altterm{def:csf}{$\csf$}, although as we have seen, a much more direct approach is available in this case.)

This begs the question whether some optimal definition for $\rel^1_\Sigma$ could preserve both confluence properties simultaneously, thus allowing us to adapt our techniques to obtain the finite frame property for \altterm{def:isf}{$\isf$}.
But for the example we have chosen, there is {\em no} simultaneously \term{forth--up} and \term{back--up} confluent relation linking the root of the left tree of $W_\Sigma$ to any world on the right tree.
This can be seen from the following reasoning.
\begin{itemize}

\item Only the $c$-world can be linked to the $e$-world, since on the left tree, $\ps r$ only belongs to $ c$, but $r\in e$.
Similarly, only the $d$-world can be linked to the $f$-world, due to $\ps p$.

\item The $b$-world cannot be linked to the $e$-world nor to the $f$-world, as noted above.
Moreover, it cannot be linked to any of the other two worlds on the right, since \term{back--up  confluent}[back--up  confluence] would force some successor of the $b$-world to be linked to the $f$-world, which we have shown to be impossible.
Thus, the $b$-world cannot be linked to any world on the right.

\item If the left root is linked to {\em any} world on the right,  \term{forth--up} confluence would force the $b$-world to be linked to at least one world on the right.
However, as we have shown, this is not possible.
\end{itemize}
This shows that we cannot in general define $\rel^1_\Sigma $ in such a way that it preserves \term{forth--up} and \term{back--up} confluence simultaneously and in such a way that $ \Gamma \rel_{\mathrm c} \Delta$ implies that $ (\Gamma) \rel^1_{\Sigma} (\Delta)$, a crucial property we need of $\rel^1_\Sigma $; compare this to Figure~\ref{figBisQuo}, where a `true' bisimulation quotient does indeed preserve both \term{forth--up} and \term{back--up} confluence.
For this reason, our techniques cannot be applied to \altterm{def:isf}{$\isf$} without major modification.
A more promising approach seems to be bounding the \term{depth} of models with respect to $\rel$ rather than $\peq$, which aligns with the approach in~\cite{DBLP:conf/lics/GirlandoKMMS23}.

Now, let us return to our proof of the finite frame property for \altterm{def:sfi}{$\sfi$}.
 
 \begin{figure}[h!]
 	 \resizebox{.7\textwidth}{!}{
 	\begin{tikzpicture}[auto,node distance=1.5cm]
 		\begin{scope}[node distance=2.5cm,local bounding box=scope1]
 			\node[state,fill=tan,line width=0.4mm] (a1)  {$a$};
 			\node[state,fill=turquoise,above left of =a1] (b1)  {$b$};
 			\node[state,fill=babyblue,above of =b1] (c1)  {$c$};
 			\node[state,fill=tan,above right of=a1] (a2)  {$a$};
 			\node[state,fill=applegreen,above of=a2] (d1)  {$d$};
 			
 			\node[state,fill=tan,right of=a2] (a3)  {$a$};
 			\node[state,fill=amber,above of=a3] (c2)  {$e$};
 			\node[state,fill=tan,below right of=a3] (a4)  {$a$};
 			\node[state,fill=turquoise,above right of=a4] (c3)  {$b$};
 			\node[state,fill=atomictangerine,above  of=c3] (c4)  {$f$};
 			
 			\path[->,dashed] (a1) edge[] node[] {} (b1);
 			\path[->,dashed] (b1) edge[] node[] {} (c1);
 			\path[->,dashed] (a1) edge[] node[] {} (a2);
 			\path[->,dashed] (a2) edge[] node[] {} (d1);
 			\path[->,dashed] (a4) edge[] node[] {} (a3);
 			\path[->,dashed] (a3) edge[] node[] {} (c2);		 
 			\path[->,dashed] (a4) edge[] node[] {} (c3);
 			\path[->,dashed] (c3) edge[] node[] {} (c4);

 			\path[->] (a1) edge[] node[] {} (a4);
 			\path[->] (b1) edge[bend left] node[] {} (a3);
 			\path[->] (a2) edge[bend left] node[] {} (c3);
 			\path[->] (c1) edge[bend left] node[] {} (c2);
 			\path[->] (d1) edge[bend left] node[] {} (c4);
 			
 		\end{scope}
 		\node[below of=scope1, node distance=3.4cm] (transition)  {};
 		\node[left of=transition, node distance=1cm] (x1)  {\begin{turn}{-45}\scalebox{1.9}{$\Downarrow$}\end{turn}};
 		\node[right of=transition, node distance=1cm] (x2)  {\begin{turn}{45}\scalebox{1.9}{$\Downarrow$}\end{turn}};
 		\begin{scope}[node distance=1.8cm,  shift={($(a1)+(-3,-6cm)$)}, local bounding box=scope2]
 			\node[state,fill=tan,line width=0.4mm] (aa1)  {$a$};
 			\node[state,fill=turquoise,above left of =aa1] (bb1)  {$b$};
 			\node[state,fill=babyblue,above of =bb1] (cc1)  {$c$};
 			\node[above right of=aa1] (aa2)  {};
 			\node[state,fill=applegreen,above of=aa2] (dd1)  {$d$};
 			
 			\node[right of=aa2] (aa3)  {};
 			\node[state,fill=amber,above of=aa3] (cc2)  {$e$};
 			\node[state,fill=tan,below right of=aa3] (aa4)  {$a$};
 			\node[state,fill=turquoise,above right of=aa4] (cc3)  {$b$};
 			\node[state,fill=atomictangerine,above  of=cc3] (cc4)  {$f$};
 			
 			\path[->,dashed] (aa1) edge[] node[] {} (bb1);
 			\path[->,dashed] (bb1) edge[] node[] {} (cc1);
 			\path[->,dashed] (aa1) edge[] node[] {} (dd1);
 			\path[->,dashed] (aa4) edge[] node[] {} (cc2);
 			\path[->,dashed] (aa4) edge[] node[] {} (cc3);
 			\path[->,dashed] (cc3) edge[] node[] {} (cc4);
 			\path[->] (aa1) edge[] node[] {} (aa4);
 			\path[->] (bb1) edge[] node[] {} (aa4);
 			\path[->] (cc1) edge[bend left] node[] {} (cc2);
 			\path[->] (dd1) edge[bend left] node[] {} (cc4);	
 			\node [draw,line width=0.4mm, fit=(scope2)] {};
 			
 		\end{scope}

 		\begin{scope}[node distance=1.8cm,shift={($(a4)+(-1,-6cm)$)}, local bounding box=scope3]
 			\node[state,fill=tan,line width=0.4mm] (aaa1)  {$a$};
 			\node[state,fill=turquoise,above left of =aaa1] (bbb1)  {$b$};
 			\node[state,fill=babyblue,above of =bbb1] (ccc1)  {$c$};
 			\node[above right of=aaa1] (aaa2)  {};
 			\node[state,fill=applegreen,above of=aaa2] (ddd1)  {$d$};
 			
 			\node[right of=aaa2] (aaa3)  {};
 			\node[state,fill=amber,above of=aaa3] (ccc2)  {$e$};
 			\node[state,fill=tan,below right of=aaa3] (aaa4)  {$a$};
 			\node[state,fill=turquoise,above right of=aaa4] (ccc3)  {$b$};
 			\node[state,fill=atomictangerine,above  of=ccc3] (ccc4)  {$f$};
 			
 			\path[->,dashed] (aaa1) edge[] node[] {} (bbb1);
 			\path[->,dashed] (bbb1) edge[] node[] {} (ccc1);
 			\path[->,dashed] (aaa1) edge[] node[] {} (ddd1);
 			\path[->,dashed] (aaa4) edge[] node[] {} (ccc2);
 			\path[->,dashed] (aaa4) edge[] node[] {} (ccc3);
 			\path[->,dashed] (ccc3) edge[] node[] {} (ccc4);

 			\path[->] (aaa1) edge[] node[] {} (aaa4);
 			\path[->] (aaa1) edge[] node[] {} (ccc3);
 			\path[->] (ccc1) edge[bend left] node[] {} (ccc2);
 			\path[->] (ddd1) edge[bend left] node[] {} (ccc4);
 			\node [draw,line width=0.4mm, fit=(scope3)] {};
 		\end{scope}
 	\end{tikzpicture}
 }
 	\caption{An $\sf S4I$ model, where $\peq$ is the transitive, reflexive closure of the relations defined by the dashed arrows and $\rel$ of the relation defined by the solid arrows.
 	The letter/colour on each world represents its label, where $\Sigma:=\{p,q,r,\ps p,\ps r\}$, $a = \varnothing$, $b:=\{q\}$, $b:=\{q\}$, $c:=\{q,\ps r\}$, $d:=\{\ps p\}$,  $e:= \{r,\ps r\}$, and $f:= \{p,q,\ps p\}$.}\label{figS4IRel}
 	
 \end{figure}

\begin{definition}\label{defAccS4I}
For \term{theory}[theories] $\Phi$, $\Psi$ in $W_{\mathrm c} $, we define $\Phi \peq^0_\Sigma \Psi$ if $\Phi \peq_{\mathrm c}  \Psi$ and $\Phi \cap \Sigma = \Psi \cap \Sigma$.
We define $\Gamma \rel^0_\Sigma \Delta$ if there is $\Theta$ so that $\Gamma\rel_{\mathrm c}  \Theta \seq^0_\Sigma \Delta$.

We define, for $\Gamma= (\Gamma _0,\ldots,\Gamma _n  )$ and $\Delta = (\Delta _0,\ldots,\Delta _m  ) \in W_\Sigma$, $\Gamma \rel^1_\Sigma \Delta$ if there is a non-decreasing sequence $j_1,\ldots j_n $ with $j_n = m$ such that $ \Gamma _i  \rel^0_\Sigma  \Delta _{j_i} $ for $i\leq n$.
We define $\rel_\Sigma$ to be the {transitive} closure of $\rel^1_\Sigma$.
\end{definition}

With this, we define $\mathcal M^{\sfi}_\Sigma = (W_\Sigma,{\peq_\Sigma},{\rel_\Sigma},V_\Sigma)$.
It is easy to check using the above lemmas that $\mathcal M^{\sfi}_\Sigma$ is a \term{bi-preordered model}.
Next we show that it is indeed a model based on an \altterm{def:sfi}{$\sfi$} frame.

\begin{lemma}\label{lemIsS4IModel}
The relation $\rel_\Sigma $ is \altterm{forth--up confluent}{forth--up} and \term{forth--down confluent} on $ S$.
\end{lemma}

\begin{proof}
By Lemma~\ref{lemmClosure}, it suffices to show that $\rel^1_\Sigma $ is \altterm{forth--up confluent}{forth--up} and \term{forth--down confluent}.
Suppose that $\Phi \peq_\Sigma \Phi' $ and $\Phi \rel^1_\Sigma \Delta$ and write $\Phi' = (\Phi_i)_{i\leq m'}$, so that $\Phi=(\Phi_i)_{i\leq m}$ for some $m\leq m'$, and $\Delta = (\Delta_i)_{i\leq n}$.
We construct $\Delta' = (\Delta_i)_{i\leq n'}$ for some $n'\geq n$ such that $\Phi' \rel_\Sigma \Delta'$ and $\Delta \peq_\Sigma \Delta' $.
Given $k \in [m,m']$, we assume inductively that $(\Delta_i)_{i\leq r}$ have been built so that $(\Phi_i)_{i\leq k} \rel_\Sigma (\Delta_i)_{i\leq r}  $.
The base case with $r=m$ is already given by the assumption that $\Phi \rel_\Sigma  \Delta $.
For the inductive step, assume that $(\Phi_i)_{i\leq k} \rel_\Sigma (\Delta_i)_{i\leq r}  $.
Then in particular, $\Phi_k \rel^0_\Sigma \Delta_r$.
By the definition of $\rel_\Sigma^0$, there is $\Theta \in W_{\mathrm c} $ so that $\Phi_k \rel_{\mathrm c}  \Theta \seq^0_\Sigma \Delta_r$.
By forth--up confluence of $\mathcal M^{\sfi}_{\mathrm c} $, there is $\Upsilon$ such that $\Phi_{k+1} \rel_{\mathrm c}  \Upsilon \seq_{\mathrm c}  \Theta$.
If $ \Delta_r \cap \Sigma =  \Upsilon \cap \Sigma$, we observe that $\Delta_r \peq^0_\Sigma \Upsilon$.
Hence $\Phi_{k+1} \rel^0_\Sigma \Delta_r$, so that $(\Phi_i)_{i\leq k + 1} \rel^1_\Sigma (\Delta_i)_{i\leq r}  $.
In this case, we may simply set $r'=r $.
Otherwise, we set $r'=r+1$  and $\Delta_{r+1} = \Upsilon$.
In this case, we also have that $(\Phi_i)_{i\leq k} \rel^1_\Sigma (\Delta_i)_{i\leq r}  $.

For \term{forth--down} confluence, as above, if $\Gamma \rel^1_\Sigma \Delta$, then there is a non-decreasing sequence $j_1,\ldots j_m $ with $j_m = n$ such that $ \Gamma _i  \rel^0_\Sigma  \Delta _{j_i} $ for $i\leq m$.
Then, any $\Gamma' \peq_\Sigma \Gamma$ is of the form $\Gamma= (\Gamma _0,\ldots,\Gamma _{m' }  )$ with $m'\leq m$, and setting $\Delta'  = (\Delta _0,\ldots,\Delta _{j_{m' }}  )$, we readily see that $\Gamma' \rel^1_\Sigma \Delta' \peq_\Sigma\Delta $.
\end{proof}

\begin{lemma}\label{lemTruthDiamS4I}
Let $\Gamma \in W_\Sigma$.
\begin{enumerate}

\item If $\ps\varphi \in \Sigma$ then $\ps\varphi \in \ell( \Gamma ) $ if and only if there is $\Delta \ler _\Sigma\Gamma $ such that $ \varphi\in \ell( \Delta )$.

\item If $\nec\varphi \in \Sigma$ then $\nec\varphi \in \ell( \Gamma ) $ if and only if for every $\Delta \ler _\Sigma\Gamma $, $ \varphi\in  \ell(  \Delta )$.

\end{enumerate}

\end{lemma}

\begin{proof}
Let $\Gamma= (\Gamma_i)_{i\leq n}$.
For the first claim, we first prove the easier right-to-left direction.
Suppose that $\Delta = (\Delta_i)_{i\leq m} \ler^1 _\Sigma\Gamma $ is such that $ \varphi\in \ell( \Delta )$, meaning that $\varphi \in \Delta_m$.
Then $\Gamma_n \rel^0_\Sigma \Delta_m $, which means that for some $\Theta$, $\Gamma_n \rel_{\mathrm c}  \Theta \seq^0_\Sigma \Delta_m $.
Then $\varphi \in \Theta$; hence $\ps\varphi \in \Gamma_n = \ell(\Gamma)$ by the truth lemma.
If instead $\Delta\ler_\Sigma \Gamma$, then there is a sequence
\[\Gamma = \Upsilon_0 \rel^1_\Sigma \Upsilon_1 \rel^1_\Sigma \ldots \rel^1_\Sigma \Upsilon_n = \Delta,\]
and backwards induction on $i$ shows that $\ps\varphi\in \ell(\Upsilon_i)$, so that in particular $\ps\varphi\in \ell(\Gamma)$.

For the left-to-right direction, suppose that $\ps\varphi \in  \ell( \Gamma ) $.
By backwards induction on $k\leq n$ we prove that there exists some sequence $ (\Theta_i)_{i = k}^n$ such that $ \Gamma_i \rel _{\mathrm c}  \Theta_i $ for each $i\in [k,n]$ and $\Theta_i\peq_{\mathrm c} \Theta_{i+1}$ if $i\in[k,n)$ (interval notation should be interpreted over the natural numbers).
In the base case $k=n$ and, by the truth lemma, there is $\Theta_n \ler_{\mathrm c}  \Gamma_n$ such that $\varphi\in \Theta_n$.
So, suppose that $ (\Theta_i)_{i = k+1}^n$ has been constructed with the desired properties.
Since the canonical model is \term{forth--down confluent}, there is some $\Theta_k$ such that $\Gamma_k \rel_{\mathrm c} \Theta_k\peq_{\mathrm c}  \Theta_{k+1} $, yielding the desired $\Theta_k$.

The problem is that the sequence $\Theta=(\Theta_i)_{i\leq n}$ may not be an element of $W_\Sigma$.
We instead choose a suitable subsequence $\Delta = (\Theta_{j_i})_{i\leq m}$, where $j_k$ is defined by (forward) induction on $k$.
For the base case, we set $j_0 = 0$, and $m_0 = 0$ (meaning that $\Delta$ currently has $m_0+1$ elements).
Now, suppose that $m_k$ has been defined as has been $j_i$ for $i\leq m_k$, in such a way that $(\Theta_{j_i})_{i \leq m_k} \in  W_\Sigma$ and $(\Gamma_i)_{i\leq k} \rel_\Sigma (\Theta_{j_i})_{i \leq m_k}$.
To define $m_{k+1} \geq m_k$ and $(\Theta_{j_i})_{i \leq m_{k+1}}$, we consider two cases.
First assume that there is $\psi\in \Sigma$ such that $ \psi \in  \Theta_{k+1} \setminus \Theta_{j_k}  $.
In this case, setting $m_{k+1} = m_k + 1$ and $j_{m_{k+1}} = k+1 $ we see that $(\Theta_{j_i})_{i \leq m_{k+1}}$ has all desired properties.
In particular, the existence of the formula $\psi$ guarantees that $(\Theta_{j_i})_{i \leq m_{k+1}} \in W_\Sigma$.
Otherwise, set $m_{k+1} = m_k$.
In this case we see that $\Gamma_{k+1} \rel_{\mathrm c}  \Theta_{k+1} \seq^0_\Sigma \Theta_{j_{m_k}}$, so that $\Gamma_{k+1} \rel^0_\Sigma \Theta_{j_{m_k}}$ and thus $(\Gamma_i)_{i\leq k+1} \rel_\Sigma (\Theta_{j_i})_{i \leq m_{k+1}}$.
That $(\Theta_{j_i})_{i \leq m_{k+1}} \in W_\Sigma$ follows from the fact that $m_{k+1} = m_k$ and the induction hypothesis.
The claim then follows by setting $m = m_n$ and $\Delta = (\Theta_{j_i})_{i \leq m  }$.

The second claim is proven similarly, but by contraposition. We leave the details to the reader.
\end{proof}

From Lemmas \ref{lemTruthImpS4I} and \ref{lemTruthDiamS4I}, we immediately obtain the following.

\begin{lemma}\label{lemTruthS4I}
For $\Gamma \in W_\Sigma$ and $\varphi\in \Sigma$, we have $(\mathcal M^{\sfi}_\Sigma,\Gamma) \models\varphi $ if and only if $\varphi\in \ell(\Gamma)$.
\end{lemma}

\begin{lemma}\label{lemS4IShallowComp}
For $\varphi\in \Sigma$, we have $ \varphi \in \sfi$ if and only if $\mathcal M^{\sfi}_\Sigma  \models\varphi$.
\end{lemma}

\begin{proof}
We know that $ \varphi \in \sfi$ if and only if $ \mathcal M^{\sfi}   \models\varphi$.
Now, given $\Gamma \in W_\Sigma$, we have $\ell(\Gamma)=\Theta \cap \Sigma$, for some \term{prime} set $\Theta$ containing all derivable formulas, so $\varphi \in \ell(\Gamma)$. \Cref{lemTruthS4I} yields that $ \varphi \in \sfi$ implies $(\mathcal M^{\sfi}_\Sigma,\Gamma)   \models\varphi$.
Conversely, if $\varphi \not\in \sfi$ then there is $\Gamma \in W_{\mathrm c} $ such that $\varphi\not\in \Gamma $, which setting $\Gamma'\! =\!(\Gamma)$ to be a singleton sequence yields $(\mathcal M^{\sfi}_\Sigma, \Gamma')   \not \models\varphi$.
\end{proof}

\begin{theorem}[{\altterm{def:sfi}{$\sfi$}} finite frame property]\label{thm:s4i_finite_frame}
The logic \altterm{def:sfi}{$\sfi$} has the finite frame property. That is, if a formula $\varphi$ is \term{falsifiable} (respectively~\term{satisfiable}) on a \altterm{def:sfi}{$\sfi$} frame, then $\varphi$ is \term{falsifiable} (respectively \term{satisfiable}) on a finite \altterm{def:sfi}{$\sfi$} frame. 
\end{theorem}

\begin{proof}
In view of \Cref{thmShallowtoFin}, it suffices to show that \altterm{def:sfi}{$\sfi$} has the \term{shallow}, \term{forest-like} model property.
Fix a formula $\varphi$ and let $\Sigma$ be the set of subformulas of $\varphi$.
By Lemmas \ref{lemS4IShallow} and \ref{lemIsS4IModel}, $\mathcal M^{\sfi}_\Sigma$ is a model based on a \term{shallow}, \term{forest-like} \altterm{def:sfi}{$\sfi$} frame, and by \Cref{lemS4IShallowComp}, $\mathcal M^{\sfi}_\Sigma  \models\varphi$ if and only if $\varphi \in \sfi$, as needed.
\end{proof}

\Cref{lemStrongBisBound} yields a bound on the cardinality of the domain of  $\mathcal M^{\sfi}_\Sigma  $ of $2^{(|\Sigma|+1)|\Sigma|}_{|\Sigma|+1}$.
Since this bound is computable, we obtain decidability of \term{def:sfi}[$\sfi$].

\begin{corollary}[{\altterm{def:sfi}{$\sfi$}} decidability]
The logic \term{def:sfi}[$\sfi$] is decidable.
\end{corollary}

As for other logics, we could bound the complexity of \term{def:sfi}[$\sfi$], although the best we obtain is \textsc{tower}.

\section{Concluding remarks}\label{sec:conclusions}


We have settled the long-standing problems of the finite frame property for \altterm{def:csf}{$\csf$} and the decidability of \altterm{def:gsf}{$\gsf$} and \altterm{def:gsfc}{$\gsfc$}.
We also introduced a logic \altterm{def:sfi}{$\sfi$} closely related to \altterm{def:isf}{$\isf$} that enjoys the finite frame property as well.
The logics we considered correspond to classes of models with combinations of the \term{forth--up confluent}[forth--up], \term{back--up confluent}[back--up], and \altterm{forth--down confluent}{forth--down confluence} properties.
We have already sketched in the discussion regarding \Cref{figS4IRel} why we do not obtain the same result for \altterm{def:isf}{$\isf$}, but there are a handful of other logics that may be defined in this fashion for which our techniques should be applicable.

We could also consider logics with the \term{forth--up confluent}[forth--up confluence] property alone, the basis of what we call \emph{\altterm{regularlogic}{$\ps$-regular}} logics.
We have not provided an axiomatisation for the logic of {reflexive} {transitive} \altterm{regular}{$\ps$-regular} frames, as our completeness proofs rely on either Axiom~\ref{ax:fs} or Axiom~\ref{ax:cd} being available.
Note that this class does not validate $\nec p\to\nec\nec p$. First results on logics of frames satisfying only forth--up confluence were recently presented in \cite{bagageol23}, where the $\mathsf K$ version was axiomatised and shown to be decidable.

Similarly, we may consider the class of \term{forth--down confluent} frames.
While Axiom~\ref{ax:cd} holds on frames that are \term{forth--up} \emph{and} \term{forth--down confluent}, it is unclear how the class of \term{forth--down confluent} frames may be axiomatised.
On the other hand, it seems that the \term{shallow} model construction we have provided for \altterm{def:sfi}{$\sfi$} should readily adapt to frames satisfying only the \term{forth--up confluent}[forth--up confluence] or the \altterm{forth--down confluent}{forth--down confluence} property, so our techniques should be applicable to these classes of frames.
All that is needed is to find the respective axiomatisations.

In total, there are eight logics that could be obtained from combinations of the three confluence frame conditions, and eight more for their \altterm{upward linear}{upward-linear} variants.
The study of \altterm{def:sfi}{$\sfi$}-like logics whose frames satisfy both \altterm{forth--up confluent}{forth--up} and \altterm{forth--down confluent}{forth--down confluence} has recently been expanded to the enriched \emph{bi}-intuitionistic modal language \cite{KR2023-26}.
Note that the decidability results in that paper are not obtained directly via the \term{shallow} model property, given that models can no longer be assumed to be forest-like in the bi-intuitionistic setting.

It is worth noting that G\"odel modal logics are of independent interest (see for example~\cite{Caicedo2010StandardGM,MetcalfeO09}), with the G\"odel variants of the modal logics $\mathsf{K}$ and $\mathsf{S5}$ enjoying the finite model property~\cite{CaicedoMRR13}.
However, the decidability of \altterm{def:gsf}{$\gsf$} and \altterm{def:gsfc}{$\gsfc$} remained challenging problems, since these logics do not enjoy a finite model property for real-valued semantics.
We have shown how birelational semantics do not have this issue, and indeed the logics enjoy the finite frame property for this semantics, and hence are decidable.
We believe that this work could lead to a systematic treatment of other prominent G\"odel modal logics whose decidability was previously unknown, by first providing them with a birelational interpretation and then establishing the finite frame property under this semantics.

Finally, an interesting question left open is that of the complexity of the logics we have studied.
The exponential bound on the size of \altterm{def:gsf}{$\gsf$} and \altterm{def:gsfc}{$\gsfc$} models yields a \textsc{nexptime} upper bound on complexity, but it is possible that validity is in fact \textsc{pspace}-complete since this is the case of G\"odel linear temporal logic defined by similar frame conditions \cite{gtlwollic}.
In the case of \altterm{def:csf}{$\csf$}, the only known lower bound is \textsc{pspace}, inherited from the intuitionistic propositional fragment~\cite{statman1979intuitionistic}.
Closing this gap remains an interesting open problem.
Our conjecture is that \altterm{def:csf}{$\csf$} is indeed \textsc{pspace}-complete.
Meanwhile, the finite models we have produced for \altterm{def:sfi}{$\sfi$} are superexponential in size, but it may be possible to obtain smaller models via a subsequent selection.
In view of the non-primitive-recursive complexity of bimodal logics with similar frame properties~\cite{pml}, it is unclear whether an elementary upper bound may be found.
We expect the results we have presented to be an important step in determining the complexity of these four logics.
\color{black}

\section*{Acknowledgments}
David Fern\'andez-Duque and Brett McLean were partially funded by the SNSF--FWO Lead Agency Grant 200021L\_196176 (SNSF)\slash G0E2121N (FWO).
David Fer\'andez-Duque was also partially funded by the Spanish Ministry of Science and Innovation Grant PID2023-149556NB-I00.
Martin Di\'eguez was partially funded by the research project CTASP (\textit{Région Pays de la Loire}, France).

\end{document}